\documentclass[12pt,a4paper]{article}

\usepackage[margin=2.58cm]{geometry}
\usepackage{setspace}

\usepackage[utf8]{inputenc}
\usepackage{amssymb}
\usepackage{mathtools,amsfonts,amsthm, mathrsfs}
\usepackage{nth}

\usepackage{dcolumn}

\usepackage{csquotes}

\usepackage{xr-hyper}

\usepackage{authblk}

\newcolumntype{d}[1]{D{.}{.}{#1} } 

\usepackage{comment}
\usepackage{bbm}

\usepackage{graphicx,wrapfig, tikz, cancel, hyperref, multirow,rotating}
\usetikzlibrary{fit,shapes.geometric}
\usetikzlibrary{shapes.misc}

\usepackage{caption, subcaption}
\usepackage{natbib}
\usepackage{enumerate}
\usepackage{colortbl}
\definecolor{blue(pigment)}{rgb}{0.2, 0.2, 0.6}

\hypersetup{
	colorlinks=true,
	citecolor=blue,
	linkcolor=blue,
	filecolor=magenta,      
	urlcolor=blue,
	pdftitle={Sharelatex Example},
	pdfpagemode=FullScreen,
}

\usepackage{blkarray}
\usepackage{diagbox}
\usetikzlibrary{positioning}
\usetikzlibrary{calc}
\usetikzlibrary{patterns,decorations.pathreplacing}
\usetikzlibrary{intersections, angles, quotes}
\definecolor{green}{rgb}{0.0, 0.5, 0.0}

\usepackage{pdflscape} 

\newtheoremstyle{break}
{}
{}
{\itshape}
{}
{\bfseries}
{.}
{\newline}
{}

\theoremstyle{break}
\newtheorem{teo}{Theorem}
\newtheorem{defi}{Definition}[section]
\newtheorem{prop}{Proposition}

\theoremstyle{plain}

\theoremstyle{definition}

\newtheorem{defi*}{Definition}
\newtheorem{example}{Example}
\newtheorem{cor}{Corollary}[section]

\newtheorem{lemma}{Lemma}[section]
\newtheorem{oss}{Remark}[section]

\usetikzlibrary{positioning}
\usetikzlibrary{calc}

\newcommand{\dd}{\mathrm{d}}

\newcommand{\EE}{\mathbb{E}}
\newcommand{\RR}{\mathbb{R}}

\newcommand{\eps}{\boldsymbol{\varepsilon}}

\newcommand{\xv}{\boldsymbol {x}}

\newcommand{\pv}{\boldsymbol {p}}
\newcommand{\qv}{\boldsymbol {q}}
\newcommand{\vv}{\boldsymbol {v}}

\newcommand{\ev}{\boldsymbol {e}}

\newcommand{\oomega}{\boldsymbol{\omega}}

\usepackage{comment}
\usepackage{mathtools}
\usepackage{subfiles}
\usepackage{wrapfig}

\begin{document}

\title{\vspace{-0mm}Multilateral Market Power in Input-Output Networks\thanks{I wish to thank Fernando Vega-Redondo for his guidance throughout this project. I wish to thank for valuable comments Christoph Carnehl, Vasco Carvalho, Jér\^ome Dollinger, Matt Elliott, Dino Gerardi, Ben Golub, Basile Grassi, P\"ar Holmberg, Margaret Meyer,  Pavel Molchanov, Ignacio Monz\'on, Marco Ottaviani, Marco Pagnozzi, Marzena Rostek, Ariel Rubinstein, Alex Teytelboym, Flavio Toxvaerd, Cole Williams, Ji He Yoon. Refine.ink was used to check the paper for consistency and clarity. 
\\
This study was funded by the European Union -
NextGenerationEU, in the framework of the GRINS -Growing Resilient,
INclusive and Sustainable project (GRINS PE00000018 - CUP: E63C22002140007). The views and opinions expressed are solely those of the authors and do not necessarily reflect those of the European Union, nor can the European Union be held responsible for them.
	}
}

\author{\vspace{-3mm}Matteo Bizzarri\thanks{Universit\`a  Federico II di Napoli and CSEF. matteo.bizzarri@unina.it}
}

\date{November 2025
}

\maketitle

\begin{abstract}
This paper models firm-to-firm trade in a production network as a set of double auctions. Firms have \enquote{multilateral market power}, namely, can affect prices in both input and output markets. The size and division of surplus are endogenous and depend only on technology, network position, and consumer preferences. The standard simplifying assumption of price-taking on input markets (\enquote{unilateral market power}) has systematic effects: it underestimates the final price and overestimates the surplus going upstream. 
These phenomena affect the model predictions for the welfare impact of mergers.


\end{abstract}

\bigskip




\section*{Introduction}

In a network of supplier-customer relationships, who has the power to decide the prices? How does this power depend on the network position and the other primitives? 
The amount and heterogeneity of market power across firms is at the core of our understanding of many economic phenomena, and consequent policy responses, e.g., the welfare impact of mergers and the diffusion of shocks.
Market power is present in output markets \citep{de2020rise}, and the interaction with the input-output network is a fundamental source of amplification, because distortions accumulate through the supply chain \citep{baqaee2020productivity}.
Moreover, also input-market power is sizable, both in labor \citep{azar2024monopsony} and intermediate inputs \citep{morlacco2019market, dhyne2022imperfect,alviarez2023two}. Input-market power has recently received attention from competition authorities, particularly in relation to labor \citep{kariel2024competition}.
However, with the exceptions noted in the literature, most customarily used models of firm-to-firm trade impose the simplifying assumption that in each market some firms (e.g., the sellers) are price-setters, while others (typically the buyers) are price-takers.


The main contributions of this paper are three. The first is to show how competition in schedules can offer a model of oligopolistic firms connected by an input-output network with the ability to affect prices both in input and output markets, in an endogenously determined way: I call this \emph{multilateral market power}. The model can accommodate general network structures and is relatively tractable thanks to nice technical properties discussed below. The second contribution is to show that competition in schedules, with proper restrictions on the price impact, can offer a unifying framework for many standard models of oligopoly in networks.
The third contribution is to characterize the direction of the bias generated by constraining market power to be unilateral, as often done. This generates two systematic biases. It decreases the market power distortions, and it predicts too much surplus going upstream.
Both these characteristics stem from a basic economic mechanism: the ability to charge high markups (on the output) or markdowns (on the inputs) stems from the perceived elasticity of the residual demand (when selling) and supply (when buying). If firms are price-takers in input markets, the elasticity of the supply is larger, and so distortions are smaller. In principle, buyer power can decrease upstream prices: however, in equilibrium, the distortion is passed to the consumer. So, ignoring buyer power leads to underestimate the final price. Moreover, unilateral market power creates a systematic asymmetry, forcing firms to have zero markdown. Under some conditions, this means that the model with unilateral market power predicts more profit going to upstream firms than it should.



I briefly describe the model. Firms have a set of input goods and produce a single output. Some outputs are the input of other firms, and these trade relationships, or \emph{input-output links}, are exogenous. 
Firms trade using a uniform-price double auction for each good, as in models of the financial market \citep{malamud2017decentralized}.  
More specifically, firms play a simultaneous game in which the available actions are supply and demand schedules, relating quantities of the traded goods to prices. The realized price on every trade relationship is the one where demand and supply cross. The key feature generating market power is that firms, being non-infinitesimal, fully internalize the mechanism and choose their schedules to affect prices in their favor. The classic metaphor for the price-taking general equilibrium behavior is that a \enquote{Walrasian} auctioneer proposes prices and collects supply and demand \enquote{bids}, until all markets clear. The approach followed in this paper takes this metaphor one step further, applying it to non-infinitesimal firms. The auctioneer acts as a market maker in financial markets, collecting firms' conditional schedules. The competition in schedules is not meant as a literal description of the workings of the market (although they are in some cases, e.g., the electricity or financial markets), but as an abstraction of a bargaining procedure, parsimonious but powerful enough for the complexity of the problem. In practice (as \cite{klemperer1989supply} suggest), the choice of a schedule can be thought of as representing all the organizational or managerial decisions that affect how a firm responds to different market conditions. In some industries, e.g., natural gas \citep{hubbard1991efficient},
 contracts commonly include clauses allowing price/quantity changes conditional on some events.



The model is tractable thanks to parametric assumptions that ensure that profits are linear-quadratic in quantities. This is classic in models of competition with schedules (\cite{klemperer1989supply}, \cite{malamud2017decentralized}), because it means that there exists an equilibrium where schedules are linear. For simplicity of exposition, in the main text, I set up the game directly as a choice of a linear schedule: the analyzed equilibrium remains an equilibrium even if firms are allowed to choose general schedules. 
Moreover, the linear equilibrium remains an equilibrium even when the intercept of the linear schedules is a random variable and, moreover, it is an ex-post equilibrium. These technical properties are classic in similar settings, and I discuss them in the Supplemental Appendix.

The key technical departure of my model from financial market models such as \cite{malamud2017decentralized} is that intermediate inputs are perfect complements. This leads to a key simplification because the choice of a linear schedule boils down to the choice of a single number, representing the slope of both the supply and all demand schedules. This makes the firms' optimization problem unidimensional and allows me to prove that the game is a supermodular potential game. In turn, the concavity of the potential implies the uniqueness of the linear equilibrium, which, to the best of my knowledge, was not known in general networks. This is the content of Theorem \ref{thm:existence}. However, perfect complementarity is by no means essential for all the results: the Supplemental Appendix shows that existence, the centrality interpretation of markups, and the comparison with local market power are still valid when intermediate inputs are imperfect complements or substitutes.


In this model, market power cannot be captured by a single index, because firms charge both a markup when selling and potentially heterogeneous markdowns when buying. Theorem \ref{thm:markups} is the main result connecting markup and markdowns to the network position. It compares two firms $i$ and $j$ such that $i$ is a supplier to $j$, with the simplifying assumption that price impact matrices are diagonal, ruling out indirect effects such as, e.g., some supplier of $i$ that trades with $j$. The result says, in words, that the upstream firm $i$ has a larger markup than $j$ if the other customers are \enquote{small} with respect to the importance of the input-output connection with $j$. Other details of the network do not matter. The mechanism is the pass-through effect mentioned above: the upstream firm perceives a smaller elasticity of demand, because the elasticity compounds the elasticity perceived by $j$ plus the imperfect pass-through from $i$ to $j$.
Markdowns behave symmetrically: if the vertical connection is the most important for the customer $j$, then the downstream firm $j$ has a larger markdown. 
The two effects push in different directions, so which firm has higher profit depends on the details of the network and the parameters.
For example, in the symmetric situation of a supply chain with homogeneous layers, each firm makes the same profit (Proposition \ref{prop:markup}).

In general, the vector of markup and markdowns of each firm can be seen in terms of a Bonacich-like centrality measure in the \emph{goods network}, a network in which the nodes are the goods, and two goods are linked if some firms trade both goods. 
In particular, if the downstream markets are disconnected from the upstream markets, the price impact on the output is proportional to the number of loops out of the output in the goods network. The number of direct and indirect connections measures the strength of the pass-through effect.

To single out the effect of multilateral market power, I then generalize the model to the \enquote{Generalized SDFE}, in which firms still choose schedules, but where the price impact functions are a given primitive. The generalization also has independent interest because, for different choices of price impact functions, we obtain as special cases not only the model of the previous paragraphs, but also other standard models: the classic Cournot oligopoly (without input-output dimension), the sequential monopoly \`a la \cite{spengler1950vertical}, the sequential Cournot, models that adopt price-taking assumptions on some markets, and, in general, all the combinations thereof. I study in detail two relevant special cases: (i) firms take input prices as given (which I call unilateral market power) and (ii) firms take as given all prices of markets where they are not directly involved (local market power). These two sets of assumptions are relevant because they are often used in quantitative models of input-output networks, as discussed in the Literature below. Proposition \ref{teo_comparative} shows that the key property of strategic complementarity is still valid: this is the key technical tool allowing the following results.

  

Theorem \ref{thm:main_comparative} characterizes the effect of constraining market power to be local, or unilateral. Part 1 shows that the constraints affect the total size of surplus: the final price is larger with multilateral market power. Part 2 shows that unilateral market power also has specific effects on the distribution of surplus: when the vertical dimension is the most prominent, profit is mechanically larger for upstream firms.
 
The intuition for both results stems again from the effect of the assumptions on the pass-through effects. Part 1 shows that buyer (or countervailing) market power has the effect of increasing, rather than decreasing, distortions.
Whenever a firm takes input prices as given (unilateral), it perceives a higher elasticity of supply. Strategic complementarity means that, in equilibrium, the firm chooses a higher slope, so this translates to higher elasticity of demand for the buyers too. As a consequence, if the network is symmetric enough, the seller charges a lower markup and the buyer charges a lower markdown. So, without buyer power, distortions are actually smaller.
The effect is analogous in the case of local market power: the firm takes as given prices in downstream or upstream markets, which implies smaller pass-through effects and higher elasticities, leading again to lower distortions. 
Concerning part 2, the mechanism is the same as Theorem \ref{thm:markups}: if one vertical connection is important enough with respect to other customers, the markup is be higher upstream. However, with unilateral market power, there is no opposite effect on markdowns. As a consequence, the upstream firm has higher profits, purely as a result of the unilateral assumption. Again, a homogeneous supply chain with layers is the clearest example: all firms have the same profit with multilateral market power, but, with unilateral market power, profits are higher the more upstream firms are.



The feedback effects generated by the adjustment of markups (and markdowns) are among the key channels that mediate the model's predictions of more complex economic phenomena. To illustrate this, I explore the role of multilateral market power in a classic applied supply chain question: the evaluation of the welfare impact of vertical mergers. Constraining market power to be unilateral changes qualitatively the model's predictions (Proposition \ref{prop_vertical}). The reason is that with unilateral market power, the model leads to underestimating the inefficiency generated by double marginalization, thus predicting a smaller efficiency gain from the merger. This is true both using the unilateral SDFE and the sequential Cournot.

Whether multilateral market power is quantitatively important depends on the specifics of the application. If the object of interest is a specific market where the details of the competition mechanism (e.g., timing, negotiation procedures, etc) can be precisely observed, then the model should reflect those details.
 However, if the analyst is trying to derive insights from a large input-output network,  it is presumably much harder to observe a lot of details on bargaining procedures and the timing of the offers.
For example, in competition policy, it is already standard to analyze vertical connections, but considering more detailed information on input-output interconnections may generate sharper insights: this is argued, e.g., by \cite{elliott2019networks} and especially \cite{elliott2019role}. 
The message of the paper is that in these contexts, assumptions that exogenously restrict firms' market power push results in a specific direction, and so may be misleading. In such cases, the Supply and Demand function equilibrium may be a useful tool.



\subsection*{Related literature}

This paper contributes to three lines of literature: the literature on competition in supply and demand functions, the literature on production networks or networked markets, and the literature on countervailing market power.

My contribution to the literature on competition on supply and demand functions is to introduce the technique to the modeling of general equilibrium oligopoly with firm-to-firm trade. \cite{flynn2025theory} studies a general equilibrium model of monopolistic competition with supply functions, without firm-to-firm trade. 
The closest setting to firm-to-firm trade is trading of heterogeneous financial assets in decentralized markets: \cite{malamud2017decentralized} show how the strategic complementarity property extends to the network setting, and characterizes an equilibrium in a general network. As discussed above, the main technical departure is the assumption of perfect complementarity, which buys the uniqueness of the equilibrium. The assumption 
is standard for intermediate inputs, but less natural for assets. 
Moreover, my focus is on the effect of multilateral market power, a question that comes from the comparison with standard supply chain models, whereas \cite{malamud2017decentralized} study the effect of decentralization (as \cite{wittwer2021connecting}).
In a similar setting, \cite{rostek2021exchange} and \cite{rostek2025financial} study restrictions on the clearing mechanism, rather than different assumptions on firms' market power. Technically, clearing restrictions mean that the equilibrium is not ex-post, which means that price impacts are affected by inference. Instead, in the present paper the equilibrium is always ex-post (in the version with uncertainty of Appendix \ref{sec:full_blown}).
The implications are different: restrictions on the clearing mechanism may increase or decrease the price impact, depending on how they affect inference, 
while both unilateral and local market power always decrease the price impact (by Theorem \ref{thm:main_comparative}).  

The literature has studied the situation where the demand firms receive comes from a network structure, in \cite{wilson2008supply} and \cite{holmberg2018supply}. \cite{holmberg2025multi} study the case of multi-product firms. These papers do not consider firm-to-firm trade. 
 \cite{ausubel2014demand} also study uniform price auctions in different settings. \cite{vives2011strategic} studies market power arising from asymmetric information, rather than network position.

The paper also contributes to the literature studying the effect of buyer and countervailing power. \cite{azar2021general} study the effect of firms having market power both on the output good and labor, showing a non-trivial welfare effect of ownership. I focus on multilateral market power, specifically on firm-to-firm trade, keeping the ownership structure standard.
The literature on countervailing market power typically uses Nash bargaining, or Nash-in-Nash models: to this category belong \cite{hart1990vertical}, \cite{alviarez2023two}. \cite{avignon2025markups} and \cite{demirer2025welfare}  also explore the effect of buyer power in a supply chain. Nash bargaining models differ from the supply and demand function competition because the division of the surplus is determined (or heavily affected) by the exogenously chosen bargaining parameters. In the supply and demand function equilibrium, both the size and the division of the surplus are endogenously determined without the need for additional parameters.
In particular, Nash bargaining selects by construction a constrained efficient allocation, given the constraints on what can or cannot be bargained on.\footnote{
This modeling choice already generates a lot of implications, as made clear by the review of \cite{toxvaerd2024bilateral}. Nash-in-Nash models feature typically multiple interrelated bargaining problems, so the overall equilibrium is typically inefficient} Instead, in the SDFE, ability to extract higher surplus comes at the cost of reduction in production: the standard monopolistic effect. This explains the different results: \cite{avignon2025markups} and \cite{demirer2025welfare} study the efficient mix of buyer and seller power, whereas in the SDFE the only efficient configuration is the one without any market power.
Countervailing market power between buyers and sellers without Nash bargaining is studied in \cite{weretka2011endogenous} and \cite{hendricks2010theory}. They do not consider a general production network, and do not focus on the effect of multilateral market power. 

My contribution to the production networks literature is to provide a model of competition in an input-output network in which all firms have market power on both input and output markets and are fully strategic, internalizing their position in the supply chain. Many models explicitly assume that firms have the power to decide/affect prices only on one side of the market (the assumption I call below \emph{unilateral market power}).
To this class belong the workhorse sequential oligopoly games in 
 \cite{salinger1988vertical} and \cite{ordover1990equilibrium}, and many others. 
Sometimes the assumption takes the form of firms choosing output price, with a fixed marginal cost. This implicitly assumes unilateral market power: first, because the marginal cost is typically taken as given, and also because it considers the relative input prices fixed and independent of demand.
The familiar formula of the price equal to a (possibly constant) markup as a constant markup times a marginal cost is a prominent example.
Examples in this category are \cite{grassi2017io}, \cite{bernard2022origins}, \cite{baqaee2019macroeconomic}, \cite{baqaee2020productivity}, 
\cite{magerman2020pecking}, \cite{dhyne2022imperfect}, \cite{bizzarri2024common}. 
An exception is \cite{acemoglu2025macroeconomics}, who use Nash bargaining. As discussed above, Nash bargaining is conceptually different and, in particular, it is efficient among the allocations satisfying the constraints. Indeed, in the benchmark of \cite{acemoglu2025macroeconomics} with no exit, the equilibrium is efficient. Instead, the SDFE can generate markup and markdown distortions even without exit.  
Except for \cite{acemoglu2025macroeconomics}, all mentioned production networks papers also feature the implicit or explicit assumption that firms do not internalize the effect of their decisions on sectors/firms further downstream besides the direct customers (an assumption I call \emph{local market power}).  Sometimes this is a consequence of the assumption of a continuum of firms in each sector (so that sector-level aggregates are taken as given by every individual firm, as in \cite{baqaee2018cascading}), other times it is explicitly assumed (e.g., in \cite{grassi2017io}, \cite{dhyne2022imperfect}).
More in general, many models of networked markets have studied the network defined by the demand: \cite{galeotti2024robust}, \cite{pellegrino2025product}, \cite{bimpikis2019cournot}, but they do not focus on input-output connections. In all these contexts, prices have some centrality interpretation, analogous to Equation \eqref{prices_centrality}. However, the markup-markdown \eqref{markup_centrality} is more novel because, since most models only have unilateral market power, the markdowns are zero. Moreover, markups have a centrality interpretation (e.g., in \cite{bimpikis2019cournot}, \cite{pellegrino2025product}) only in models in which market power is not local; otherwise, by construction they cannot depend on global network characteristics.


The rest of the paper is organized as follows. Section \ref{sec:simple} illustrates the model through a simple example of a vertical economy.
Section \ref{model} defines the benchmark model, the Supply and Demand Function Equilibrium (SDFE). Section \ref{solution} describes the solution and the existence theorem. Section \ref{sec:network} characterizes the connection of markups, markdowns, and network position. Section \ref{sec:multilateral} introduces the Generalized SDFE and explores the effect of multilateral market power. 
Section \ref{conclusion} concludes. The main proofs are in the Appendix.

\section{A simple example}

\label{sec:simple}

In this section, I illustrate the model and the main take-aways in the simplest network where the concept of multilateral market power is non-trivial: a supply chain consisting of one intermediate producer, $U$, and a final good producer $D$. This is represented in Figure \ref{Fig:line_simple}. The intermediate good producer $U$ produces good $U$ using only labor, and sells it to both the final producer $D$ and consumers. In turn, the final producer $D$ uses good $U$ to produce the final output $D$. The consumers consume both goods $U$ and $D$.\footnote{The fact that consumers also consume good $U$ is necessary for the model to have a non-trivial equilibrium. This is a well-known technical feature of competition in schedules: a non-trivial linear equilibrium exists if goods are traded by at least 3 agents \citep{malamud2017decentralized}.}

\begin{wrapfigure}{l}{0.35\textwidth}

	\begin{tikzpicture}[mycircle/.style={
		circle,
		draw=black,
		fill=gray,
		fill opacity = 0.3,
		text opacity=1,},
	node distance=  2cm
	]

	\node [mycircle] (1) {$U$};

	
	\node [mycircle, below of =1] (0) {$D$};
	\node [ below of =0] (c) {Consumers};

	\foreach \i/\j in {
		1/0,
		0/c}
	\draw [->] (\i) -- node[above] {} (\j);
	\draw [->] (1.east) to [out=-45,in=45] (c.north);
	
\end{tikzpicture}

\caption{A simple vertical economy.}

\label{Fig:line_simple}
\end{wrapfigure}

Firms have a standard linear technology: $F_U(q_U)=q_U$ and $F_D(q_D)=q_D$, so that profits are: $\pi_{U}=p_Uq_U$, $\pi_D=(p_D-p_U)q_D$. Consumer demand for both goods is linear: $D_{c,D}(p_D)=1-p_D$, $D_{c,U}(p_U)=1-p_U$. 
The firms play a simultaneous game in which the strategic variables are the (slopes of the) linear schedules connecting prices and quantities. Formally:
\begin{enumerate}
	
	\item  firm $U$ submits a supply function $S_U(p_U)=B_{U}p_U$, where $B_{U}$ is any positive real number;
	
	\item firm $D$ submits a function $S_{D}(p_U,p_D)=	B_{D}(p_D-p_U)
$	 indicating both its supply of output, and its demand for the input, where $B_{D}$ is, again, any positive real number.
	 	 
\end{enumerate} 

Whichever choice of the firms, the prices $p_U$, $p_D$ and quantities $q_U$, $q_D$ must satisfy the market-clearing conditions:
\begin{align}
	q_D&=1-p_D=B_D(p_D-p_U)  	\label{simple_market_clearing_D}\\
	q_U&=1-p_U+B_D(p_D-p_U)=B_Up_U
	\label{simple_market_clearing_U}
\end{align}
We look for the Nash equilibrium of this game.

Focus on firm $U$. It is going to be convenient to express the best reply problem of firm $U$ using the inverse demand $p^r_U(q_U)$. To obtain it, first use the market-clearing equation for good $D$ \eqref{simple_market_clearing_D} to express $p_D$ as a function of $p_U$, then substitute into \eqref{simple_market_clearing_U}. Doing so, we obtain:
\[
p^r_{U}(q_U)=1-\left(1+\left(\dfrac{1}{B_D}+1\right)^{-1}\right)^{-1}q_U.
\]
Similarly, we obtain the inverse demand and supply faced by firm $D$, denoted $p^r_{D,U}(q_D)$ and $p^r_{D,D}(q_D)$.
 When taking the FOC for firm $U$, we get:
 \[
\dfrac{\partial}{\partial B_U}\pi_U= \dfrac{\partial q_U}{\partial B_U}\left(p_U+q_U\dfrac{\partial p_{U,U}}{\partial q_U}\right)=0
 \]
From the market-clearing conditions, it is easy to conclude that $\dfrac{\partial q_U}{\partial B_U}>0$, and so the FOC are equivalent to: $p_U+q_U\dfrac{\partial p_{U,U}}{\partial q_U}=0$. Doing the analogue for firm $D$, we obtain the equilibrium equations:
\begin{subequations}
\begin{align}
&p_U+q_U\dfrac{\partial p_{U,U}}{\partial q_U}=0 \label{eq:simple_U}\\
& p_D-p_U+q_D\left(\dfrac{\partial p_{D,D}}{\partial q_D}-\dfrac{\partial p_{D,U}}{\partial q_D}\right)=0 \label{eq:simple_D}
\end{align}
\label{eq:simple}
\end{subequations}
Since schedules are linear, the derivatives are just constants. So, the above equations imply the following schedules relating prices and quantities:
\begin{align*}
	S_U(p_U)&=\left(-\dfrac{\partial p_{U,U}}{\partial q_U}\right)^{-1}p_U\\
	S_D(p_D,p_U)&=\left(\dfrac{\partial p_{D,D}}{\partial q_D}-\dfrac{\partial p_{D,U}}{\partial q_D}\right)^{-1}(p_D-p_U)&
\end{align*}
So, the slopes $B_D^*$, $B_U^*$ that constitute an equilibrium of the game must be equal to the slope of the above functions, and satisfy:
\begin{align}
B_U^*& = \left(-\dfrac{\partial p_{U,U}}{\partial q_U}\right)^{-1}=1+\left(\dfrac{1}{B_D^*}+1\right)^{-1} \label{foc_simple_U}\\
B_D^*& = \left(\dfrac{\partial p_{D,D}}{\partial q_D}-\dfrac{\partial p_{D,U}}{\partial q_D}\right)^{-1}=\left(1+\dfrac{1}{B^*_U+1}\right)^{-1}.
\label{foc_simple_D}
\end{align}
The expression highlights the role of the price impacts (the derivatives of the inverse demand/supply), and in particular, the fact that firm $D$ has a price impact on both the input and the output market. The equations can be solved analytically (eliminating denominators, we are left with a quadratic equation), and it can be checked that the solution is: $B_D^*=1/\sqrt{2}$, $B_U^*=\sqrt{2}$, which using the market-clearing equations \eqref{simple_market_clearing_U}, \eqref{simple_market_clearing_D} yield prices $p_U=1/2$ and $p_D=(3-\sqrt{2})/2$. 

In this model both firms have multilateral market power. Now suppose that for the sake of simplicity we want to adopt the more standard assumption that firm $D$ is a price-taker on the input market. This means that, from $D$'s perspective, its choice does not affect the input price, so $\dfrac{\partial p^r_{D,U}}{\partial q_D}=0$. The equilibrium equations \eqref{foc_simple_D},\eqref{foc_simple_U} become:
\begin{align*}
	B_U^{**}& =1+\left(1+\dfrac{1}{B_D^{**}}\right)^{-1} \\
	B_D^{**}& =1.
\end{align*}
It turns out that this solution is exactly the same we would get solving the model as a standard sequential monopoly, as shown in Example \ref{ex:vertical_merger}.
 The solution in this case is $B_D^{**}=1$, $B_U^{**}=3/2$, yielding prices of $p_U=1/2$ and $p_D=3/4$. Comparing to the solution with multilateral market power, we see that the final price is larger. This is a general fact (see Theorem \ref{thm:main_comparative} below): constraining market power to be unilateral leads to underestimate distortions due to market power.
 

The situation with multilateral market power is represented graphically in Figure \ref{Fig:equilibrium}: firm $U$ faces a residual demand $D_{U}^r(p_U)$, that is the blue line in the graph, depending on the slope of consumers and the slope chosen by $D$. This residual demand induces a profit as a function of the price $p_U$. Firm $U$ wants to charge $p_U^*$, the monopoly price for this residual demand, and so sets a slope that achieves that price: this is the red line in Figure \ref{Fig:equilibrium1}.  But, in doing so, it affects the slope of the
\emph{residual supply} that firm $D$ faces. As a consequence, firm $D$ changes its choice of schedule, changing the transaction price to $(p^*_U)_2$, the optimal \emph{monopsony} price for firm $D$. This, in turn, leads to a new residual demand and a new profit function for firm $U$ (as in Figure \ref{Fig:equilibrium2}): as a consequence, the previous optimal price $p_U^*$ is not optimal anymore, and firm $U$ adjusts its slope again. This adjustment process continues until the slopes are such that the optimal price sellers want to charge is equal to the optimal price for the buyers.
 In a sense, both the buyer and the seller of good $U$ set their \enquote{optimal price}. This seems a contradiction, since sellers would want to raise $p_U$ while buyers would want to decrease it. The tension is resolved by the fact that firms \enquote{implement} a price by modifying the slope of their schedule that, in turn, \emph{changes other firms' incentives to raise prices}. Without multilateral market power, firm $D$  takes the input price as given: this means that the slope of the residual demand $D^r$ is fixed, and the equilibrium slopes are higher: this leads to lower prices.

\begin{figure}[h!]
		
		\begin{subfigure}{0.5\textwidth}

\hspace*{-6em}
\begin{tikzpicture}[scale=0.9,rotate=-90, xscale=-1]		
	
	\coordinate (O) at (0,0);
	\draw[->, name path=x] (-0.3,0) -- (7,0) coordinate[label = {above: price}] (xmax);
	\draw[->, name path=y] (1,-0.3) -- (1,7) coordinate[label = {below: quantity}] (ymax);
	
	\node[] at ([yshift=10pt] current bounding box.north) {\textcolor{green}{\textbf{Market for good $U$}}};
	
		\draw[blue,-, name path=res] (1,5) -- (6,0);
		\node[blue, below left, align=center] at (1,6) {Residual\\demand};
		
		\draw [blue, domain=3:6,name path=Rev1 
		] plot(\x, {(6-\x)*\x-(6-\x)*(6-\x)});
		\draw[dashed,-, name path=horizontal1] (4.5,0)--+(0,4.5);			
	
	\node[red,left] at (4.5,0) {$p_U^*$};

		
		\path[name path=PC] (1,0) --   (6,5);
		
		\path[name intersections ={of= res and horizontal1, by=Q1}];	
		
			\draw[shorten >=-2cm, red,-, name path=SD] (1,0) -- (Q1);
			\node[red, above=10pt] at (6,2.5) {Best reply of U};
		
		\path[name path=horizontal2] (Q1)--+(-4,0);
		\path[name intersections ={of=horizontal2 and PC, by =MC1}];
		
		\path[name path=vertical11] (MC1)--+(0,-4);		
		\path[name path=vertical12] (Q1)--+(0,-4);	
		\path[name intersections ={of=vertical11 and x, by =mc1}];	
		\path[name intersections ={of=vertical12 and x, by =P1}];
		
			\draw[dashed] (Q1)--(P1);
			

\end{tikzpicture}
\caption{}
\label{Fig:equilibrium1}
		\end{subfigure}%
\begin{subfigure}{0.5\textwidth}
	
	\hspace*{-1em}
		\begin{tikzpicture}[scale=0.9,rotate=-90, xscale=-1]		
			
			\coordinate (O) at (0,0);
			\draw[->, name path=x] (-0.3,0) -- (7,0) coordinate[label = {above: price}] (xmax);
			\draw[->, name path=y] (1,-0.3) -- (1,7) coordinate[label = {below: quantity}] (ymax);
			
			\node[] at ([yshift=10pt] current bounding box.north) {\textcolor{green}{\textbf{Market for good $U$}}};
			
				\draw[blue!35,-, name path=res] (1,5) -- (6,0);
				
				\draw [blue!35, domain=3:6,name path=Rev1 
				] plot(\x, {(6-\x)*\x-(6-\x)*(6-\x)});
							
			

				
				
				\path[name path=horizontal1] (4.5,0)--+(0,5);
				\path[name intersections ={of= res and horizontal1, by=Q1}];	
				
					\draw[shorten >=-2cm, red!35,-, name path=SD] (1,0) -- (Q1);
					\node[red, above=10pt] at (6,2.5) {Best reply of U};
				
				\path[name path=horizontal2] (Q1)--+(-4,0);
				

			
				
				\draw[blue,-, name path=res2] (1,1.5) -- (6,0);
				\draw[dashed,->, blue] (0.8,5)-- (0.8,1.5);
				\node[blue, below] at (0.8,2.2) {Best reply of $D$} ;	
				
				\path[name intersections ={of= SD and res2, by=Q2}];
				\path[name path=horizontal1] (Q2)--+(0,-2);
				\path[name intersections ={of= horizontal1 and x, by=p2}];
				\node[blue,left] at (p2) {$(p_U^*)_2$};	
				\draw[dashed,-, name path=horizontal1] (Q2)--(p2);
			
				
				\draw [blue, domain=2:6,name path=Rev2 
				] plot(\x, {0.6*(6-\x)*\x-0.36*(6-\x)*(6-\x)});
				
			
				\draw[dashed,-, name path=horizontal3] (4.1,0)--+(0,3.37);					
				\node[red, left] at (4.1,0) {$(p_U^*)_3$};
				
			
				\path[name intersections ={of= res2 and horizontal3, by=Q3}];	
				
				\draw[shorten >=-2cm, red,-, name path=SD] (1,0) -- (Q3);
				\draw[dashed,->, red] (6, 2) -- (6,1);
			
		\end{tikzpicture}
\caption{}
\label{Fig:equilibrium2}		
\end{subfigure}	
			
\caption{Graphical representation of the choice of schedule by firm $U$. On the left (a): the best reply for firm $U$ to the residual demand given by the blue line. On the right (b): the optimal choice of firm $U$ leads other firms to adjust, modifying firm $U$ residual demand and optimal price: so firm $U$ further adjusts its best reply. }		
\label{Fig:equilibrium}
		
	\end{figure}
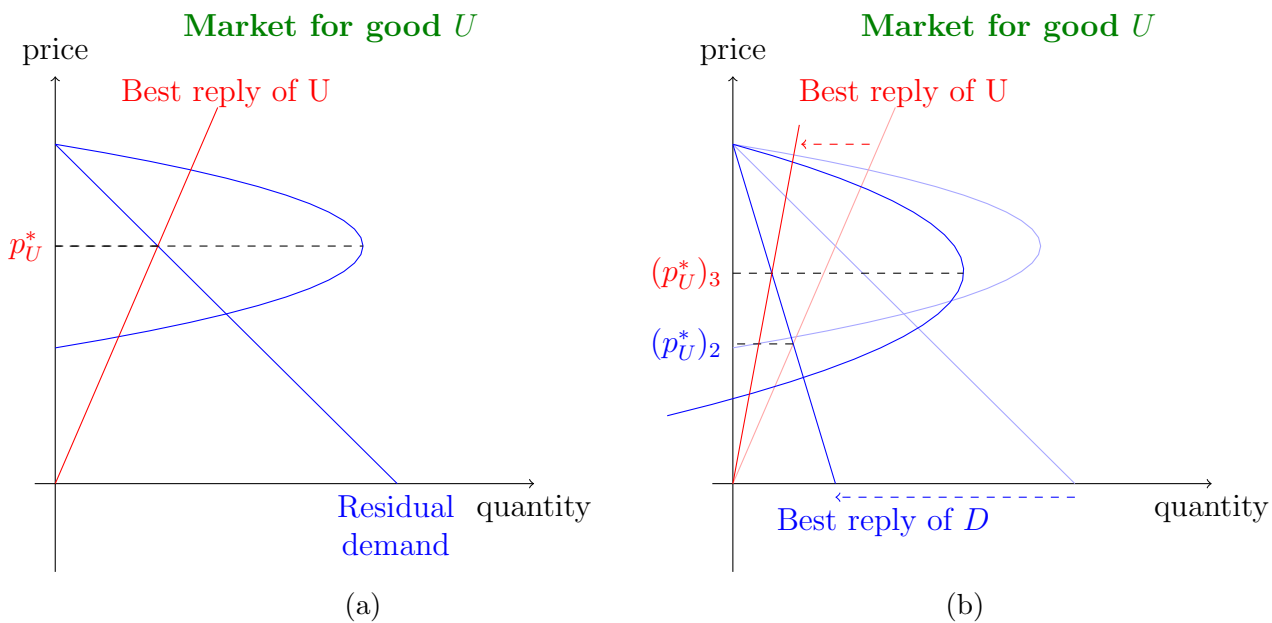

\section{The model}
\label{model}


\subsection{Setting}

\paragraph{Notation}  

Bold symbols are used to denote vectors: $\pv\in \mathbb{R}^m$ is the vector of all prices of goods. Labor is the numeraire, so the wage is normalized to 1. If $\vv$ is a vector indexed on goods and $\mathcal{A}$ is a subset of goods, $\vv_{\mathcal{A}}=((v_g)_{g \in \mathcal{A}})$ denotes the vector that selects only the entries relative to goods belonging to $\mathcal{A}$. If $\mathcal{A}=\mathcal{N}(i)$, for brevity, I use the notation $\pv_i:=\pv_{\mathcal{N}(i)}$ (the prices of all input and output goods of firm $i$). For quantities, $\qv_i:=\qv_{\mathcal{N}(i)}$ is the vector of input and output \emph{net} quantities traded by firm $i$: positive quantities represent outputs and negative quantities represent inputs. Similarly, $\pv_c:=(p_g)_{g \in \mathcal{C}}$ and $\qv_c:=(q_g)_{g \in \mathcal{C}}$ are the vectors of prices and quantities consumed by the consumer. Also for simplicity, I write $p_i^{out}:=p_{out(i)}$ and $q_i^{out}=q_{out(i)}$ for the price and quantity of the output good of $i$. I use $\qv_i$ for quantities when they are the argument of a function. When I am talking about quantities as functions of prices, I denote the function $\mathcal{S}_i(\pv_i)$ and call it a \emph{schedule}. Similarly, $\pv_i$ are prices as variables, while $\pv_i^r(\qv_i)$ is the \emph{residual schedule} (defined below) representing prices as a function of quantities. 

The quantity of labor used by firm $i$ is $\ell_i$.
If $M$ is a matrix, $M_{\mathcal{A}}$ denotes the square submatrix where both row and columns are indexed by goods in $\mathcal{A}$; $M_{\mathcal{A}^c}$ denotes the square submatrix indexed by the complement; and $M_{\mathcal{A},\mathcal{A}^c}$ denotes the off-diagonal block.
If $B_i\in \mathbb{R}^{d_i\times d_i}$ is a coefficient matrix, the notation $\hat{B}_i$ denotes the \emph{lifting} of the matrix $B_i$, namely the matrix in $\mathbb{R}^{m\times m}$ such that the nonzero elements are exactly those corresponding to the elements of matrix $B_i$, and the rest is filled with zeros. Analogously, if $\vv_i\in \mathbb{R}^{d_i}$, the notation $\hat{\vv}_i$ denotes the vector in $\mathbb{R}^m$ in which the additional elements are filled with zeros.

\paragraph{Firms and Production Network} There are $n$ firms and $m$ goods: their sets are respectively denoted $\mathcal{N}$ and $\mathcal{M}$. Each good might be produced by more firms, but each firm produces only one good. Each firm produces using labor and a set of inputs produced by other firms, which I denote as $\mathcal{N}^{in}(i)\subseteq \mathcal{M}$. If firm $i$ has a unique input, I denote it $in(i)$. If good $g$ is the output of firm $i$, I denote it as $out(i)$. The set of all goods traded by firm $i$ is $\mathcal{N}(i)=\mathcal{N}(i)^{in}\cup\{out(i)\} \subseteq \mathcal{M}$. The number of non-labor goods traded by $i$ is $d_i=|\mathcal{N}(i)|$, so the number of intermediate inputs is $d_i-1$. 
The consumers' utility depends potentially only on a subset of goods, denoted $\mathcal{C} \subseteq \mathcal{M}$.  The connections $(\mathcal{N}(i))_{i\in \mathcal{N}}$ defined above define a bilateral network between firms and goods, which is the \emph{input-output network} of this economy.  

\paragraph{Technology}
\label{technology}

Intermediate inputs are perfect complements, so that to produce a quantity of output $q_i$  firm $i$ needs $-q_{ij}=f_{ih}q_i$ units of input $h$, with $f_{ih}\ge 0$ (remember that input quantities are negative, $q_{ij}<0$). Denote $F\in \mathbb{R}^{n\times m}$ the matrix with entries $f_{ih}$. The technology has to be \emph{viable}, namely
there must exist a positive vector $\pv\in \mathbb{R}^m_+$ such that $p_i>\sum_h f_{ih}p_h$ for each $i$. 
\footnote{If $n=m$ so that $F$ is square, this is equivalent to the standard condition that $I-F$ must be invertible with nonnegative inverse \citep{horn1994topics}. With $n>m$, this is reminiscent of the \emph{multi-activity} model from standard input-output analysis, with the important difference that the classic input-output analysis is done with perfect competition.} Moreover, each good must be connected to consumers, that is, for each good $g$ there must be a sequence of goods $g_1,\ldots, g_k$ such that each pair of consecutive goods represents the input and output of some firm $i$, and the final good $g_k$ is consumed by the consumer, $g_k\in \mathcal{C}$.
To produce $q^{out}_i$ units of output, the firm also needs $\ell_i=f_{i,L}q^{out}_i+\dfrac{1}{2k_i}(q_i^{out})^2$ labor units, with $f_{i,L}>0$. This slightly generalizes the Leontief technology to allow decreasing returns: as illustrated below, this facilitates the existence of an equilibrium. Formally, the production function is: $q_i^{out}=\min\{F_i(\ell_i),(-q_{ij})/f_{ij}\}$, where $F_i(\ell_i):=f_{i,L}k_i\left(\sqrt{2\ell_i/(f_{i,L}^2k_i)+1}-1\right)$ is the inverse of the labor expression above.
If $k_i\to \infty$, the technology becomes the standard Leontief. Denote the vector of labor requirements $f_{i,L}$ as $\boldsymbol{f}_L$. It is convenient to define the vector $\vv_i=(1,-(f_{ig})_{g \in \mathcal{N}^{in}(i)}) \in \mathbb{R}^{d_i}$.
So, we can write the technology constraint of firm $i$ as:
\begin{align}
\qv_i&=q_i^{out}\vv_i \nonumber\\
\ell_i&=f_{i,L}q^{out}_i+\dfrac{1}{2k_i}(q_i^{out})^2
 \label{tech_constraint_min}
 \end{align}

\paragraph{Consumers} The utility function of the consumers is quadratic in consumption and (quasi-)linear in the disutility of labor $L$:
\begin{equation}
	U(\qv_c,L)=\boldsymbol{A}'B_c^{-1}\qv_c-\frac{1}{2}\qv_c'B_c^{-1}\qv_c-L,
	\label{utility}
\end{equation}
where $\boldsymbol{A}$ is a positive vector and $B_c\in \mathbb{R}^{|\mathcal{C}|\times |\mathcal{C}|}$ is a symmetric positive definite matrix. Throughout the paper, prices $\pv$ are always in labor terms, so the wage is normalized to 1.
This means that the consumer demand has the form: $D_{c}(\pv_c)=\boldsymbol{A}-B_c\pv_{c}$.

\subsection{The game}
\label{sec:thegame}
\paragraph{Schedules}

Firms compete by choosing a supply function for the output, and demand functions for intermediate inputs and labor, respecting the technology constraint \eqref{tech_constraint_min}. The players of the game are the firms: $i=1,\ldots, N$, and the actions available to each firm $i$ are \emph{schedules} mapping prices to quantities. Denote the schedule for the intermediate inputs and output as: $\mathcal{S}_i:\mathbb{R}^{d_i}\to \mathbb{R}^{d_i}$ and the one for labor as $\mathcal{S}_{\ell,i}:\mathbb{R}^{d_i}\to \mathbb{R}$.  
The key simplifying assumption is that the intermediate input schedule $\mathcal{S}_i$ is linear. Linearity means that there exists a matrix of coefficients $B_i\in \mathbb{R}^{d_i\times d_i}$ and a vector $B_{i,f}\in \mathbb{R}^{d_i}$, such that the intermediate input schedule is linear:
\begin{equation}
\mathcal{S}_i(\pv_i)=B_i\pv_i-B_{i,f}
\label{linear_schedule}
\end{equation}
Notice that I do not restrict prices or quantities to be positive. With the linear functional form conditions that guarantee that equilibrium values are nonnegative are complicated (even if they normally are). For simplicity, I do not restrict the sign of prices and quantities, and simply interpret a negative sign as trade happening in the opposite direction (as customarily done in the financial-market interpretation of the model). 
However, when discussing the results relating market power and network position in Sections \ref{sec:network} and \ref{sec:multilateral}, I assume that output quantities are positive.

The technology constraints \eqref{tech_constraint_min} imply that the supply function $\mathcal{S}^{out}_i$ determines the whole input schedule, including labor: $\mathcal{S}_i=\mathcal{S}_i^{out}\vv_i$ and $\mathcal{S}_{\ell,i}=f_{i,L}\mathcal{S}_i^{out}+\dfrac{1}{2k_i}(\mathcal{S}_i^{out})^2$. Moreover, it turns out (proven in Lemma \ref{linear-residual} below) that for the schedule \eqref{linear_schedule} to be a best reply, it must be that for some $\overline{B}_i \in [0,k_i]$:
\begin{equation}
B_i=\overline{B}_i\vv_i\vv_i',\quad B_{i,f}=\overline{B}_i\vv_if_{i,L},
\label{coefficient_equation}
\end{equation}
so that the schedule must have the form:
\begin{equation}
\qv_i=\mathcal{S}_i(\pv_i)=\overline{B}_i\left(\boldsymbol{v}'_i\pv_i-f_{i,L}\right)\vv_i.
    \label{schedule_expression}
\end{equation}
It follows that, to analyze the equilibrium, it is sufficient to restrict firms' choice to the choice of the coefficient $\overline{B}_i$ in $[0,k_i]$. In the main text, I set up the game directly as the choice of this coefficient $\overline{B}_i$, because it simplifies the analysis.
 Sometimes, for convenience of notation, I still use the $B_i$, $B_{i,f}$ matrices as shortcuts for the expressions in \eqref{coefficient_equation}.
Denote $B=(B_i)_{i\in \mathcal{N}}$ a profile of coefficient matrices, $B_f=(B_{i,f})_{i\in \mathcal{N}}$ a profile of intercept vectors, and $\overline{B}=(\overline{B}_i)_{i\in \mathcal{N}}$ a profile of slope coefficients.
When needed, we index the matrices $B_i$ directly with the relevant goods, for simplicity. So, $B_{i,gh}$ means the entries of the schedule relative to the effect of the price of $h$ on demand for the input goods $g$. 

 Most proofs also only use properties of $B_i$ directly, and so they generalize to the case of imperfect complements and substitutes, explored in the Supplemental Appendix.

\paragraph{Prices}
\label{sec:pricing}

The market prices are, by assumption, those satisfying the market-clearing conditions. The market-clearing equations for non-labor goods are:
\begin{align}
	\sum_{i\,:\,g \in \mathcal{N}_i}\mathcal{S}_{ig}(\pv_i)&=D_{c,g}(\pv_{c})\quad \forall g \in \mathcal{C} \quad \sum_{i\,:\,g \in \mathcal{N}_i}\mathcal{S}_{ig}(\pv_i)=0 \quad \forall g \in \mathcal{M}\setminus \mathcal{C}
    \label{sub_mktclear}
\end{align}
Since the demand derived by \eqref{utility} satisfies Walras's law, it is standard that one of the market-clearing conditions is redundant: we leave out the labor market-clearing equation $\sum_i \mathcal{S}_{\ell,i}(\pv_i)=L(\pv_c)$.

Lemma \ref{pricingfunction} below shows that the system has a unique solution, a \emph{pricing function} mapping coefficient matrices to prices: $\pv:(\overline{B})\to \pv(\overline{B})$. This function is crucial: it embeds the information about competition and network interconnections.

\paragraph{Payoffs}

To complete the definition of the game, it remains to define the payoffs. These are, in short, the profits, calculated in the prices that satisfy the market-clearing conditions \eqref{sub_mktclear} (remember that the wage is normalized to 1):
\begin{align}
	\pi_{i}(\overline{B})&:=\pv_i'(\overline{B})\mathcal{S}_i(\pv_i(\overline{B}))-\mathcal{S}_{\ell,i}(\pv_i(\overline{B})) \nonumber\\
    &=\overline{B}_{i}\left(1-\dfrac{1}{2k_i}\overline{B}_{i}\right)(\vv_i'\pv_i(\overline{B})-f_{iL})^2
    \label{profits}
\end{align}
In particular, since $\overline{B}_{i}\le k_i$, we get that $\pi_i(\overline{B})$ is always nonnegative, so we do not need to worry about firm exit.

\begin{defi}
A Supply and Demand Function Equilibrium (SDFE) is a Nash equilibrium of the game $G=(\mathcal{N}, ([0,k_i], \pi_{i})_{i \in \mathcal{N}})$, where the players are the firms, actions are slopes, and the payoffs are defined in \eqref{profits}.

\end{defi}

\begin{example}{\textbf{Horizontal economy/Standard Supply Function Equilibrium}}
	
	\label{example1}
Consider the case of $n$ firms, producing the same output good, without input-output connections (producing using only labor): $\boldsymbol{v}_{i}=1$ for $i=1,2, \ldots, n$. The demand function in this case is $D_c=A_c-B_cp_c$, where $A_c$, $B_c\in \mathbb{R}_+$. This is an instance  of the Supply Function competition by \cite{klemperer1989supply} (in the parametric case of the quadratic cost function). 
	
\end{example}	

\begin{example}{\textbf{The vertical economy}}
    \label{example0}
The vertical economy of Section \ref{sec:simple} is an example with two goods, $U$ and $D$, and two firms with the same indices. The network is defined by: $\mathcal{N}(U)=\{D\}$, $\mathcal{N}(D)=\{U,D\}$ and $\mathcal{C}=\{D\}$. For general parameters the schedules are: $\mathcal{S}_U=\overline{B}_U(p_U-f_{U,L})$ and $\mathcal{S}_D=\overline{B}_D(p_D-f_{DU}p_U-f_{D,L})$. Section \ref{sec:simple} explores the case in which the parameters satisfy: $\boldsymbol{v}_U=1, \boldsymbol{v}_D=(1,-1)$, $B_c=\begin{pmatrix}
    1 & 0\\
    0& 1
\end{pmatrix}$, $f_{i,L}=0$ and $k_i\to \infty$, so that the marginal costs are zero. T
    
\end{example}

The next example is also very useful because it is very tractable, so I am going to use it intensively to illustrate the model in the next sections.

\begin{example}[\textbf{Supply chain with layers}]
	\label{ex:line}

	A layered supply chain is a production structure in which firms are divided in $N$ \emph{layers} $i=1,\ldots, N$, There are $n_i$ firms in layer $i$. Each layer $i$ produces good $i$, as in Figure \ref{Fig:line}.  
	Firms in layer $i$ use as intermediate input the good produced in layer $i+1$ and sell to firms in layer $i-1$, except firms in layer 1, which sell to the consumers ($\mathcal{C}=\{1\}$). Firms in layer $N$ only use labor for production. 
    If $N=1$, we obtain the standard Supply Function equilibrium as in \cite{klemperer1989supply}, and the example above.
For simplicity, assume that for all firms $f_{j,L}=0$, $k_j=k$, and $\vv_j=(1,-1)$ for each firm in layers $i<N$, and $\vv_j=1$ for firms in the last layer $N$. With these assumptions, all the firms in each layer are ex-ante identical. It turns out that they will also choose identical schedules in equilibrium, so we can index schedules directly by the layer index: $\mathcal{S}_i$ is the schedule of all firms in layer $i$, for $i<N$, and it is equal to:
\begin{align*}
    \mathcal{S}_i(\pv_i)=\begin{pmatrix}
        \overline{B}_i(p_i-p_{i+1})\\
        -\overline{B}_i(p_i-p_{i+1})
    \end{pmatrix}
\end{align*}

\end{example}

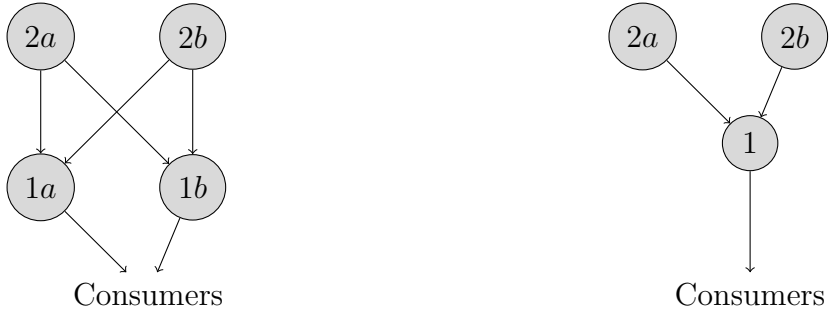
\begin{figure}[t]
	
	\centering

	\begin{subfigure}{0.5\textwidth}
		
		\hspace{2cm}
		\begin{tikzpicture}[mycircle/.style={
				circle,
				draw=black,
				fill=gray,
				fill opacity = 0.3,
				text opacity=1,},
			node distance=  2cm
			]

			\node [mycircle] (1) {$2a$};
			\node [mycircle, right of=1] (2) {$2b$};		
			
			\node [mycircle, below of =1] (0) {$1a$};
			\node [mycircle, below of =2] (3) {$1b$};		
			\node [ below right of =0] (c) {Consumers};

			\foreach \i/\j in {
				1/0,
				2/0,
				1/3,
				2/3,
				3/c,
				0/c}
			\draw [->] (\i) -- node[above] {} (\j);

		\end{tikzpicture}
		\caption{homogeneous case, with $n_1=n_2=2$.}
	\label{Fig:linea}

	\end{subfigure}%
    \begin{subfigure}{0.5\textwidth}

		\hspace{2cm}
		\begin{tikzpicture}[mycircle/.style={
				circle,
				draw=black,
				fill=gray,
				fill opacity = 0.3,
				text opacity=1,},
			node distance=  2cm
			]

			\node [mycircle] (1) {$2a$};
			\node [mycircle, right of=1] (2) {$2b$};		
			
			\node [mycircle, below right of =1] (0) {$1$};
			\node [ below of =0] (c) {Consumers};

			\foreach \i/\j in {
				1/0,
				2/0,
				0/c}
			\draw [->] (\i) -- node[above] {} (\j);

		\end{tikzpicture}
		
		\caption{Asymmetric layers, with $n_1=1$, $n_2=2$.}
	\label{Fig:lineb}
	
	\end{subfigure}
	
	\caption{A supply chain with 2 layers and 1 consumer good: $\mathcal{C}=\{1\}$.  }
	
	\label{Fig:line}
	
\end{figure}

\section{Solution: existence and uniqueness}
\label{solution}

\subsection{Residual schedule and price impact}

First of all, I show the pricing function  and the payoffs defined in \ref{sec:pricing} are well-defined, so the payoffs are indeed uniquely defined as a function of the slope coefficients. 
We need the lifting notation introduced in the previous section: if $B_i\in\mathbb{R}^{d_i\times d_i}$ is a coefficient matrix of a schedule, $\hat{B_i}\in \mathbb{R}^{m\times m}$ extends the same matrix to dimension $m\times m$ by filling in zeros; the same for vectors $\hat{B}_{i,f}$.  

\begin{lemma}

Define the matrices $M:=\sum_j\hat{B}_j+\hat{B}_c$ and $M_{f}:=\sum_j \hat{B}_{j,f}$. The market-clearing conditions \eqref{sub_mktclear} have a unique solution:
\begin{equation}
\pv=M^{-1}\overline{\boldsymbol{A}}
\label{mktclear_linear}
\end{equation}
where $\overline{\boldsymbol{A}}:=\boldsymbol{\hat{A}}+M_f$.
  

    \label{pricingfunction}
\end{lemma}

 The result follows from the assumption of viability of the network, and the linear functional form of \eqref{linear_schedule} and \eqref{coefficient_equation}. The
 proof is in the Appendix  \ref{proof_pricingfunction}.

 As in Section \ref{sec:simple}, it is  convenient to express the optimization in terms of a monopolist optimizing against a (inverse) residual demand and supply (or schedule). This formulation is useful for clarifying the role of multilateral market power, both intuitively and in the generalization of Section \ref{sec:multilateral}. Moreover, it is also useful for the generalizations in which firms are allowed to choose general coefficient matrices and general schedules (in the Supplemental Appendix). 
   
\begin{defi}
Define the (inverse) \emph{residual schedule} of firm $i$ as the function:
\begin{equation}
\pv^r_i(\qv_i; \overline{B}_{-i})=\Lambda_i(\overline{B}_{-i})(\boldsymbol{\tilde{A}}_i(\overline{B}_{-i})-\qv_i)
\label{residual_schedule}
\end{equation}
where $\Lambda_i(\overline{B}_{-i}):-\partial_{\qv_i} \pv^r_i=[(M-\hat{B}_i)^{-1}]_{\mathcal{N}(i)}$ and $\boldsymbol{\tilde{A}}_i$ is a vector that is also a function of $\overline{B}_{-i}$. 
    \label{def_residual}
\end{defi}
The coefficient matrix $\Lambda_i$ is the submatrix of $(M-\hat{B}_i)^{-1}$ relative to the goods traded by firm $i$ (the set $\mathcal{N}(i)$). It collects the slope coefficients of the (inverse) supply and demand schedules, describing the effect of firm's $i$ quantity decision on its input and output prices: its \emph{price impact}. The following properties are going to be useful.

\begin{lemma}
  For all $i$, $\Lambda_i$ exists, is positive definite, and is decreasing in the positive semi-definite order in each $\overline{B}_j$ with $j\neq i$.
    \label{lemma_priceimpact}
\end{lemma}

The next Lemma shows the equivalence between the monopoly problem against the residual schedule and the best-reply problem. This is a standard result in the context of competition in schedules.  

\begin{lemma}

 Consider the monopoly problem with the \emph{inverse demand} $\pv_i^r$:
\begin{align}
    &\max_{\qv_i, \ell_i}\boldsymbol{q}_{i}'\boldsymbol{p}^r_{i}(\qv_i;\overline{B}_{-i})-\ell_i     \label{BestReply}\\ 
    &\text{subject to the technology constraints \eqref{tech_constraint_min}} . \nonumber 
\end{align}
is equivalent to the best reply problem of firm $i$, in the sense that  $\overline{B}_i$ is the best reply to $\overline{B}_{-i}$ in the game $\mathcal{G}$ if and only if the optimum $\qv^*_i$ of the above problem \eqref{BestReply} satisfies Equations \eqref{schedule_expression}: 
\begin{equation}
\qv_i^*=\overline{B}_i(\vv_i'\pv^r_i(\qv_i^*; \overline{B}_{-i})-f_{i,L})\vv_i
\label{optimality_condition}
\end{equation}
Moreover, the expression of the best reply is:
\begin{equation}
\overline{B}_i=\left(\boldsymbol{v}_i'\Lambda_i\boldsymbol{v}_i+\dfrac{1}{k_i}\right)^{-1}
    \label{bestreply}
\end{equation}

    \label{linear-residual}
\end{lemma}

 To solve \eqref{BestReply}, the easiest way is to substitute the technology constraints to transform the problem into an optimization over $q_i^{out}$:
\begin{align}
    \max_{q_i^{out}} q_i^{out}\boldsymbol{v}_{i}'\boldsymbol{p}^r_{i}(q_i^{out}\vv_i;\overline{B}_{-i})-f_{i,L}q_i^{out}-\dfrac{1}{2k_i}(q_i^{out})^2     
    \label{FOC_restricted}
\end{align}
This is strictly concave because $\Lambda_i$ is positive definite and so the second-order condition is: $-\left(2\vv_i'\Lambda_i\vv_i+\dfrac{1}{k_i}\right)<0$. The FOC is:
  \[
  \vv_i'\pv_i^r-f_{i,L}-\dfrac{1}{k_i}q_i^{out}-\vv_i'\Lambda_i\vv_iq_i^{out}=0
  \]
  that, reordered, becomes:  \[
q_i^{out}=\left(\boldsymbol{v}_i'\Lambda_i\boldsymbol{v}_i+\dfrac{1}{k_i}\right)^{-1}(\vv_i'\pv_i^r(\qv_i^{out}\vv_i;\overline{B}_{-i})-f_{i,L})
  \]
which, remembering that $\qv_i=q_i^{out}\vv_i$, is exactly \eqref{optimality_condition}, where the coefficient $\overline{B}_i$ satisfies \eqref{bestreply}.

Note that the \enquote{monopolist} optimization over quantities \eqref{BestReply} does not impose directly that the relation between $\qv_i^*$ and $\pv_i^r(\qv_i^*)$ be linear. In principle, any functional form may arise. However, solving it, we get that the optimal schedule is linear, and has the expression in \eqref{schedule_expression}. The details are in Supplemental Appendix \ref{proof_linear_residual}.
Indeed, this is exactly the reason why the linear equilibrium remains an equilibrium in the generalized game where firms are allowed to choose general schedules, explored in the Supplemental Appendix \ref{sec:full_blown}. 
However, once firms are allowed to choose general schedules, there is a subtlety.
It is a standard observation that the optimization \eqref{BestReply} only constrains the schedule in the optimal point $\qv_i^*$, $\pv_i^r(\qv_i^*)$. So, if we allow generic instead of linear schedules, any schedule passing through the point $(\qv_i^*$, $\pv_i^r(\qv_i^*))$ would be a best reply. It is also standard that, following \cite{klemperer1989supply}, introducing uncertainty in a sufficient number of parameters, there is a \emph{unique} surviving best reply, which in this case is exactly  \eqref{schedule_expression}. 
The Supplemental Appendix \ref{sec:full_blown} shows this formally, motivating the restriction to linear strategies in the main text.




\subsection{Existence Theorem}

The next Theorem formally proves the existence and uniqueness of the equilibrium (i.e., the fixed point of \eqref{bestreply}). The proof is in Appendix \ref{existence_proof}.

\begin{teo}

There exists a unique Nash equilibrium of the game $\mathcal{G}$, and is in pure strategies. The game $\mathcal{G}$ is supermodular, so the equilibrium is also the unique rationalizable action profile. The equilibrium coefficients  $\overline{B}_1, \ldots, \overline{B}_n\in \mathbb{R}_+^n$ satisfy \eqref{bestreply}. 

\label{thm:existence}
\end{teo}

Existence follows from the fact that the coefficients $\overline{B}_i$ are bounded above by $k_i$.\footnote{This is the main reason why we introduce the quadratic term in the labor constraint. If we allow $k_i\to \infty$, the equilibrium may still exist, as in the vertical economy example in Section \ref{sec:simple}, but in general the slopes of the firms may tend to infinity. This happens, in particular, if more than one firm produces the same good. However, from an economic perspective, the limit is well defined: those firms simply behave as perfectly competitive. This is not surprising because, as highlighted by \cite{klemperer1989supply}, with constant marginal costs, the supply function equilibrium behaves as price competition.} 
Uniqueness follows from considering a modified game $\mathcal{G}'$ where firms choose (logarithm of) the slope coefficients $x_i=\ln \overline
{B}_i$. The new game corresponds to a reparameterization of the strategies of the original game, and a monotonic transformation of the payoffs. As such, any Nash equilibrium of the game $\mathcal{G}$ corresponds to one and only one Nash equilibrium of the game $\mathcal{G}'$. The game $\mathcal{G}'$ is a potential game, and the potential is strictly concave: as a consequence, the game has a unique Nash equilibrium. 

The key intuition of the fixed-point equation \ref{bestreply} is that the slope depends inversely on the quadratic form $\boldsymbol{v}_i'\Lambda_i\boldsymbol{v}_i$, which is an aggregate index of the strength of the price impact of the firm. 
 For example, in the horizontal economy of Example \ref{example1}, the price impact is simply the inverse of the slope of the residual demand:
\begin{equation}
\Lambda_i=\left(B_c+\sum_{j\neq i}\overline{B}_j\right)^{-1}.
\label{price_impact_sfe}
\end{equation}
It is a number, so the aggregation is trivial: $\boldsymbol{v}_i'\Lambda_i\boldsymbol{v}_i=\Lambda_i$. In general, the price impact is a matrix that can be decomposed into blocks as:
\[
\Lambda_i=\begin{pmatrix}
    \Lambda_i^{out} & \Lambda_i^{out, in}\\
    (\Lambda_i^{out, in})' & \Lambda_i^{in}
\end{pmatrix},
\]
where $\Lambda_i^{out}\in \mathbb{R}_+$ is the coefficient relative to output, while $\Lambda_i^{out}\in \mathbb{R}^{(d_i-1)\times (d_i-1)}$ is the submatrix relative to inputs; and $\Lambda_i^{out, in}\in \mathbb{R}^{1\times (d_i-1)}$ contains the terms relative to the cross effects of inputs on outputs. In this case, the aggregate price impact can be written:
\[
\vv_i'\Lambda_i\vv_i=\Lambda_i^{out}+\boldsymbol{f}_i'\Lambda_i^{in}\boldsymbol{f}_i-2\Lambda_i^{out, in}\boldsymbol{f}_i
\] Since $\vv_i$ has negative elements for the inputs, the links \emph{between} outputs and inputs are negatively weighted  (if the off-diagonal elements of $\Lambda_i$ are positive).
This is because, whenever input and output markets are connected, an increase in the output price reflects negatively on the input prices. So, it \enquote{subtracts} from the aggregate price impact. 

For example, in the vertical economy of Section \ref{sec:simple}, firm $D$ is the only firm trading both goods: this means that conditional on the input quantity, output has no effect on input prices and vice-versa, so $\Lambda_D^{out,in}=0$. So, the price impact of $D$ is diagonal, and the diagonal elements are simply the inverse of the demand slope and supply slope:
\begin{align}
 \Lambda_D=&\begin{pmatrix}
        \dfrac{-\partial p^r_{D,D}}{\partial q_D} & 0\\
        0& -\dfrac{\partial p^r_{D,U}}{\partial q_D}
    \end{pmatrix}=\begin{pmatrix}
        1/B_{c,D} & 0\\
        0& 1/(B_U+B_{c,U})
    \end{pmatrix}
    \label{priceimpact_D}
\end{align}
So, we have that the aggregate price impact is $\boldsymbol{v}_D'\Lambda_D\boldsymbol{v}_D=\Lambda_D^{out}+\Lambda_D^{in}$, which is what appears in Equation \eqref{eq:simple}: the equilibrium slope $\overline{B}_D$ depends inversely on the sum of the price impacts on both the input and the output.

\section{Market power and network position}

\label{sec:network}

This section analyzes how, in equilibrium, the network affects market power. The first step is to understand the connection between position, pass-through, and price impact.

\subsection{The price impact and the pass-through}
\label{sec:pass-through}

The price impact $\Lambda_i$ is the key object containing the information on multilateral market power: how much each firm can affect input and output prices.
Naturally, it depends on the number and slopes chosen by competitors, as well as on the higher-order network connections: for given schedules, a price change in a market passes through all the indirectly connected markets. This \emph{pass-through} effect is the key mechanism through which vertical (i.e. input-output) connections affect the price impact. 

To understand the pass-through of prices in vertical connections, focus on the vertical economy of Section \ref{sec:simple}, and the price impact of $U$. It is easier to understand the intuition by focusing on the \emph{direct} residual demand (quantity as function of price), which is: 
\[
D_U^r(p_U):=\underbrace{A-B_{c,U}p_U}_{\text{consumer}}+\underbrace{\overline{B}_D(p_D-p_U)}_{\text{demand from } D}.
\]
The price impact on $p_U$ is the inverse of the slope of this curve:
\[
\Lambda_U=-\left(\dfrac{\dd D^r_U}{\dd p_U}\right)^{-1}=\left(\underbrace{B_{c,U}}_{\text{consumers}}+\overbrace{\overline{B}_D-\underbrace{\overline{B}_D\dfrac{\partial p_D}{\partial p_U}}_{\text{pass-through}}}^{\text{demand from } D}\right)^{-1}
\]
We can see that this expression is the sum of two terms, coming from the two sources of demand (horizontal dimension): consumer demand, and demand from firm $D$. However, the demand from firm $D$ does not only depend on $p_U$, but also on $p_D$. By assumption, firms understand the relation between equilibrium schedules and prices (the pricing function in \eqref{pricingfunction}), so they internalize that a change in $p_U$ generates a \emph{pass-through} to $p_D$. Mathematically, we see this because the market-clearing equation for the downstream market implies that: $p_D=(A+\overline{B}_Dp_U)/(B_c+\overline{B}_D)$. Since $\dfrac{\partial p_D}{\partial p_U}>0$, the term: $\overline{B}_D-\overline{B}_D\dfrac{\partial p_D}{\partial p_U}=\left(B_{c,D}^{-1}+B_D^{-1}\right)^{-1}$, representing the importance of the \enquote{vertical} input-output connection on the price impact, is \emph{smaller} than both $\overline{B}_D$ and $B_{c,D}$: it amplifies the price impact. But $B_{c,D}^{-1}=\Lambda^{out}_{D}$. As a consequence, if $B_{c,U}$ is small enough, the effect of the vertical connection dominates and the price impact on the output of the upstream firm $U$ is larger: $\Lambda_U>\Lambda_D^{out}$.
 This is the key mechanism  linking price impacts and network position, decomposed into the horizontal and vertical dimensions. Theorem \ref{thm:markups} in the next paragraph shows that this mechanism is true in general networks.

\subsection{Markups and markdowns}

So far we talked about price impacts, but the standard \enquote{measure} of market power is the markup, or markdowns. In a context where firms have market power both on inputs and outputs, both the marginal cost and marginal revenue product depend on the price impact. 
The easiest way to obtain the markup and markdowns is to rewrite the technology constraint \eqref{tech_constraint_min} as $\qv_i=F_i(\ell_i)\vv_i$. Calling $\boldsymbol{\lambda}_i$ the vector of multipliers of this constraint, the FOC of the best reply problem become:
\begin{align}
\Lambda_i\qv_i&=\pv_i-\boldsymbol{\lambda}_i \label{pricecostmargin}\\
1&=\dfrac{\partial F_i}{\partial \ell_i}\vv_i'\boldsymbol{\lambda}_i
\end{align}
Both sides of \eqref{pricecostmargin} represents the gap between the price and the shadow value of the output and each input: we define these as the (signed) markup-markdown vector $\boldsymbol{\mu}_i$.\footnote{In the Supplemental Appendix, I show that the same values are recovered from the explicit computation of marginal cost and marginal revenue product.} So, the markup and markdowns in equilibrium satisfy: 
\begin{equation}
     \boldsymbol{\mu}_i=q_i^{out}\Lambda_i\vv_i,
     \label{eq:markup}
\end{equation} 
The magnitude of the markup and markdowns depends, not surprisingly, on the price impact: equation \eqref{eq:markup} is nothing else but the standard equation connecting the slope (or elasticity) of demand to the price charged. Indeed, under perfect competition, we have $\Lambda_i=0$ and so $\boldsymbol{\mu}_i=\boldsymbol{0}$.
Each firm potentially has a different price impact in each of the markets in which it is involved: the vector $\boldsymbol{\mu}_i$ provides a measure of how much surplus the firm extracts from each market.  

The next Theorem is the main result relating the markup-markdowns and the network position. It compares the markups and markdowns of two firms $i,j$, such that $i$ is a supplier of $j$: $i \in \mathcal{N}^{in}(j)$. For simplicity and analogy with quantities, I denote the markup on the output good as $\mu_i^{out}:=\mu_{i,out(i)}$. If there is a unique input $g$, I denote the (absolute value)  of the markdown as $\mu_i^{in}:=-\mu_{i,g}$. In the diagonal case, the markup is simply $\mu_{i}^{out}=\Lambda_i^{out}q_i^{out}>0$. If there is only one input $g$, the markdown is simply: $\mu_{i}^{in}=\Lambda^{in}_if_{ig}q_i^{out}>0$.

The comparison of markups of two firms $i$ and $j$ is clean, informally, when the set of goods can be partitioned into goods upstream and downstream, and the two have no other connections via other firms, as in the vertical economy of Section \ref{sec:simple}. In this way, the price impact is diagonal: a change in, e.g., the output does not indirectly affect input prices via indirect pass-through effects. Formally, assume the following.
\begin{description}
    \item[Assumption 1] Consider two firms $i,j$ such that $out(i) \in \mathcal{N}^{in}(j)$ and the technology coefficient is: $f_{j,out(i)}=1$. Moreover, there is a partition of the set of goods $\mathcal{N}_i^{down}$, $\mathcal{N}_i^{up}$ such that $\mathcal{N}_i^{in}\subseteq \mathcal{N}_i^{up}$ and $out(i)\in \mathcal{N}_i^{down}$, such that the two subsets are disconnected, in the sense that for any pair of goods $g \in \mathcal{N}_i^{down}$ and  $h \in \mathcal{N}_i^{up}$ it is impossible to find a sequence of goods $g_1=g, g_2,\ldots, g_k=h$ such that for each consecutive pair of goods there is a firm $j\neq i$ trading both.\footnote{That is, the two subsets are disconnected in the \enquote{goods network} defined below, after removing firm $i$.} 
\end{description}

\begin{teo}
 Consider two firms $i$ and $j$ satisfying Assumption 1.
\begin{enumerate}

\item Suppose $i$ is the only supplier of $j$ and that $\mathcal{N}_i^{down}$ can itself be partitioned into disconnected subsets $\mathcal{N}_{i,j},\mathcal{N}_{i,-j}$, where $\mathcal{N}_{i,j}$ contains $out(j)$ and $\mathcal{N}_{i,-j}$ contains all the outputs of the other customers $k\neq j$. Define the equilibrium slope of the other customers of $i$ as $\overline{B}_{-j}=\sum_{k:\,i\in \mathcal{N}^{in}(k),\, k\neq j}\overline{B}_k+B_{c,i}$. If $\overline{B}_{-j}$ is small enough, then the markup is larger for the upstream firm: $\mu_{i}^{out}>\mu_{j}^{out}$.
\item  If $j$ is the only customer of $i$ and $i$ has only 1 input $in(i)$, and that $\mathcal{N}_j^{up}$ can itself be partitioned into disconnected subsets $\mathcal{N}_{j,i},\mathcal{N}_{j,-i}$, where $\mathcal{N}_{j,i}$ contains $in(i)$ and $\mathcal{N}_{j,-i}$ contains all the inputs of the other suppliers $k\neq i$.
Define the equilibrium slope of the other suppliers of $j$ as $\overline{B}_{-i}=\sum_{k:out(k)\in \mathcal{N}^{in}(j),\, k\neq i}\overline{B}_k$. If $\overline{B}_{-i}$ is small enough, then it is also true that the markdown is larger for the downstream firm: $\mu_{i}^{in}<\mu_{j}^{in}$.
\end{enumerate}

    \label{thm:markups}
\end{teo}
The proof is in Appendix \ref{proof:markups}. If $i$ is the only supplier of $j$ and $f_{ji}=1$, firm $j$ demands a quantity equal to its output $q_j^{out}$, so it must be that $q_i^{out}\ge q_j^{out}$. 
The interesting part is that the price impacts of $i$ and $j$ on the output are related by the equation:
\[
\Lambda_i^{out}=\left(\Lambda_{-j}^{-1}+\left(\dfrac{1}{\overline{B}_j}+\Lambda_j^{out}\right)^{-1}\right)^{-1}
\]
where $\Lambda_{-j}$ represents the part of the price impact of $i$ that comes from the other customers (either firms or the final customer), different from $i$, which satisfies $0\le \Lambda_{-j}^{-1}<\overline{B}_{-j}$. For example, if the only other customer of $j$ are final consumers, $\overline{B}_{-j}=B_{c,j}$.
The above equation generalizes Equation \eqref{foc_simple_U} for the vertical economy, where $i=U$, $j=D$ and the only other customer of $i$ are consumers: $\Lambda_{-j}^{-1}=1$. In fact, the mechanism remains the pass-through mechanism discussed in the previous paragraph: the elasticity of the demand from the downstream firm $j$ to $i$ is less elastic than the demand faced by firm $j$. So, the vertical term $\left(\overline{B}_j^{-1}+\Lambda_j^{out}\right)^{-1}$ pushes the price impact to be larger upstream.
If $\overline{B}_{-j}$ (and consequently $\Lambda_{-j}^{-1}$) is small, the \enquote{vertical} dimension is more important, and we can conclude $\Lambda_i^{out}>\Lambda_j^{out}$: so the upstream firm has a higher markup. 

So, the ability of firms to extract surplus is heterogeneous and depends on the market analyzed: e.g., upstream firms may have more power on output markets, downstream firms on input markets. Who then makes the higher profit? 
The aggregate price impact appearing in \eqref{bestreply} can also be rewritten as aggregate markup/markdown:
\[
\boldsymbol{v}_i'\Lambda_i\boldsymbol{v}_i=(q_i^{out})^{-1}\vv_i'\boldsymbol{\mu}_i
\]
where $\vv_i'\boldsymbol{\mu}_i=\mu_i+\sum_g f_{i,g}\mu_{i,g}$ aggregates the markup and markdowns: the larger this sum, the smaller the equilibrium slope. Using the best reply equation \eqref{bestreply} to express the slope as a function of the markups, the profit equation \eqref{profits} becomes:
\begin{equation}
\pi_i=\left(\dfrac{1}{\overline{B}_i}-\dfrac{1}{2k_i}\right)(q_i^{out})^2=\left(\mu_i+\sum_g f_{i,g}\mu_{i,g}\right)q_i^{out}+\dfrac{(q_i^{out})^2}{2k_i}
\label{profit_markup}
\end{equation}
The two effects described in Theorem \ref{thm:markups} go in opposite directions, so which profit is higher depends on which effect is stronger, even for the same output quantity.
In the supply chain with layers, it is possible to characterize profits in more detail, providing a particularly clean example.

\subsubsection{The supply chain with layers}
\label{sec:supplychain}

Consider the supply chain with layers introduced in Example \ref{ex:line}. 
Notice that if $n_j>1$ this is \emph{not} a special case of the setting of the Theorem \ref{thm:main_comparative}, because each firm has at least another with the same input and output, so the set of goods cannot be partitioned in that way. Nonetheless, the conclusion holds. 
 Since in the homogeneous case all firms have the same output $q$, from the profit Equation \eqref{profit_markup} becomes:
\[
\pi_i=\left(\dfrac{1}{\overline{B}_i}-\dfrac{1}{2k}\right)q^2=\left(\mu_i^{out}+\mu_{i}^{in}\right)q+\dfrac{q^2}{2k}.
\]
Moreover, it turns out that in the homogenenous case, the
equilibrium slope is also homogeneous: $\overline{B}_i =B^*$ for all layers $i$,
from which we can see that markups and markdowns must move in opposite direction while keeping their sum constant. 
The next Proposition makes this formal. The proof is in Appendix \ref{proof_markups}.

\begin{prop}
Consider a layered supply chain as in Example \ref{ex:line}. Suppose $N\ge 2$ and $n_i\ge 2$, so the equilibrium slope is nonzero, and $k_i=k$. We have:

\begin{enumerate}

    \item     If $n_i=n^* \ge 2$ for all layers $i$, the markups are larger the more upstream the layer is, while markdowns are larger the more downstream a layer is;

    \item If $n_i\ge n_j$ then the aggregate profit of firms in layer $j$ is larger than the aggregate profit of firms in layer $i$. In particular, if $n_i=n^*$ for all $i$, then all the firms make the same profit.

\end{enumerate}	

    \label{prop:markup}
\end{prop}

So, the result of Theorem \ref{thm:main_comparative} holds: markups are increasing upstream, markdowns are increasing downstream. 
 What is the balance? Part 2 answers. If $n_i$ is constant across layers, as in Figure \ref{Fig:linea}, the situation is completely symmetric, and so the increase in markups and decrease in markdowns exactly offset each other, and the firms all have the same profits. Hence, each layer extracts the same surplus.  Otherwise, if some layer is more competitive ($n_i$ is larger), as in as in Figure \ref{Fig:lineb}, the corresponding firms have lower profits. Not only, but the layer with less firms as a whole also extracts less aggregate surplus than other layers. For example, if in the asymmetric example of Figure \ref{Fig:lineb}  $n_2\to \infty$, the firms in the upstream behave as price takers, and the only firm in the downstream layer behaves as a monopolist both on the output and the input market.

\subsection{The goods network}

What can we say about the relation between the network position and market power in general? We summarize the discussion in the following remarks.

For simplicity, let us focus on the case in which $\hat{B}_c=I$, and $n=m$: each firm produces a distinct good. In this case, the matrix $F$ is square. So, the matrix $M$ can equivalently be written as:
\[
M=(I-F')D(I-F)+I
\]
where $D:=diag(\overline{B})+I$ is the matrix containing only the diagonal of $M$. By the market-clearing equations \eqref{mktclear_linear}, we can write:
\begin{equation}
\pv=D^{-1}(I-G)^{-1}\boldsymbol{\overline{A}}
\label{prices_centrality}
\end{equation}
for $G$ the matrix with elements $G_{ij}=(-M_{ij})/M_{ii}$ if $i\neq j$, and $G_{ii}=0$ for each $i$. The matrix $G$ can be interpreted as the adjacency matrix of a weighted projection of the original bipartite graph $F$ onto the set of goods: a link is present if both goods are traded by at least one firm.
In matrix form, $G=(F'D+DF-F'DF)D^{-1}$. The three terms reflect three types of links, which carry a different sign: input-output connections ($F'D$ and $DF$) have positive weight, while being inputs of the same firm (the term $F'DF$) has negative weight. This is distinct from the classic \emph{one-mode projection}, which would have adjacency $F'F$, and only counts as linked goods that are simultaneous \emph{inputs} to the same firm. Here, input-output connections matter, and have an opposite sign than \enquote{horizontal connections}.
The links are weighed by the equilibrium slopes: the lower the slope, the weaker the link. 
Prices are high when a good has many direct or indirect connections with goods in high demand $\boldsymbol{A}$, or high cost $\boldsymbol{f}_L$ (this is because both forward and backward connections enter). The pre-multiplication by $D$ represents the direct effect of the slope: the smaller the slope, the higher the price impact, and so the price. The network part represents the effect of indirectly connected goods.

Equation \eqref{prices_centrality} is also valid in perfect competition (with different equilibrium slopes and so network weights), but prices depend only on connections to cost parameters. However, the direct and indirect connections also determine the magnitude of market power distortions through Equation \eqref{markup_centrality}: this is the object of the next Remark \ref{rmk_centrality}.

\begin{oss}

       In the setting above, with $n=m$, defining $D_i:=diag(\overline{B}_{-i})+I_{-i}$ is the diagonal of $M-\hat{B}_i$,
        the markup-markdowns vector can be rewritten as:
\begin{equation}
\boldsymbol{\mu}_{i}=q_i^{out}[D_{-i}^{-1}(I-G_{-i})^{-1}]_{\mathcal{N}(i)}\vv_i,
\label{markup_centrality}
\end{equation}        
where the matrix $G_{-i}$ has elements $G_{-i,gh}=-[M_{gh}-\hat{B}_{i,gh}]/[M_{gg}-\hat{B}_{i,gg}]$ for $g\neq h$ and $G_{-i,gg}=0$. It has analogous intuition as $G$, with the only difference that now all the network is the projection \emph{after removing firm $i$ from the network}. This is because the price impact is computed fixing the input and output goods of $i$. 
To understand better Equation \eqref{markup_centrality}, notice that elements of $(I-G_{-i})^{-1}_{\mathcal{N}(i)}$ count all the direct and indirect connections between each pair of goods in $\mathcal{N}_i$: only paths with endpoints in $\mathcal{N}(i)$ matter, even if possibly passing through all other nodes. The division by $D_{-i}$ is a normalization.
So, the markup-markdowns vector is proportional to a restriction of Bonacich centrality of the respective good, in the good network excluding firm $i$, limited to counting the paths that have endpoints in $\mathcal{N}(i)$. For example, if the price impact is diagonal, $\mu_i^{out}$ is proportional of the weighted number of \emph{loops} centered in the output good $out(i)$.


\label{rmk_centrality}
\end{oss}

\section{The role of multilateral market power}
\label{sec:multilateral}

\subsection{General price impacts}

The key feature of the model studied so far is that firms have multilateral market power: they can affect prices in all the markets they are involved in.
What are the implications of multilateral market power? To answer this question, in this section, I introduce a model that simultaneously generalizes the supply and demand function competition and various other classic models of oligopolistic competition, with and without networks. 

\begin{defi}
	
	Consider a profile of functions $\Lambda=(\Lambda_1,\ldots, \Lambda_n)$, where:
	\[
	\Lambda_i:\overline{B}_{-i}\to \Lambda_i(\overline{B}_{-i})\in \mathbb{R}^{d_i\times d_i}
	\]
	such that, for all $i$ and $\overline{B}_{-i}$, $\Lambda_i(\overline{B}_{-i})$ is positive semidefinite, and the function $\Lambda_i$ is continuous and decreasing in the positive semidefinite ordering in each $\overline{B}_{j}$ with $j\neq i$.
	
	A Generalized SDFE a profile of slopes $\overline{B}$ such that for each firm $i$:
    \begin{enumerate}
        \item $\overline{B}_i$ satisfies Equation \eqref{bestreply} were $\Lambda_i$ is the above price impact function;
        \item equilibrium prices and quantities are determined by the market-clearing conditions \eqref{mktclear_linear}.
    \end{enumerate}
    		
\label{def_generalized}		
\end{defi}

In other words, a Generalized SDFE is a model in which firms choose their slopes optimizing against a modified residual schedule, that satisfies: $	\partial_{\qv_i} \pv^r_i(\qv_i,\overline{B}_{-i})=-\Lambda_i(\overline{B}_{-i})$, where $\Lambda_i$ is now a primitive.
The Supply and Demand function competition of the previous sections is a Generalized SDFE, because $\Lambda_i$ is continuous and decreasing (Lemma \ref{lemma_priceimpact}). Many other standard models are also special cases. For example, Walrasian equilibrium is the special case in which $\Lambda_i=0$ for each $i$. Cournot oligopoly is also a special case.
\begin{example}{\textbf{Cournot oligopoly is a Generalized SDFE}}

In the horizontal economy of \ref{example1} consider the Generalized SDFE with:
\begin{equation}
\Lambda_i(\overline{B}_{-i})=\dfrac{1}{B_c}
\label{priceimpact_cournot}
\end{equation}
This yields the FOC: 
\[
p^{out}-\dfrac{1}{B_c}q_i^{out}-\dfrac{1}{k_i}q_i^{out}=0
\]
The own slope $\overline{B}_i$ disappears, and this is exactly the FOC of Cournot competition. This is not surprising, because the price impact function \eqref{priceimpact_cournot} is exactly the one obtained from \eqref{price_impact_sfe} by assuming that other firms choose schedules with zero slope $\overline{B}_j=0$: but a schedule with zero slope is a quantity commitment, as in Cournot competition. 

\end{example}
What is perhaps more striking is that some \emph{sequential} models are also special cases of the Generalized SDFE, as shown in the next example. 
\begin{example}[\textbf{Sequential Monopoly is a Generalized SDFE}]
\label{sequential_monopoly}

In the textbook model of sequential monopoly (see, e.g. \cite{belleflamme2015industrial}), by backward induction firm $D$ chooses $p_D$ taking $p_U$ as a given, yielding the FOC:
\begin{equation}
    p_D-p_U+\dfrac{\partial p_D}{\partial q_D}q_D=0
    \label{sequential_d}
\end{equation}
This is precisely the same equation we get in a Generalized SDFE where $\Lambda_{D,UU}=\dfrac{\partial p_U}{\partial q_D}=0$. Indeed, it can be checked that the best reply of the sequential monopolist $D$: $p_D(p_U)$, being linear, is exactly identical to the price that realizes when firm $D$ chooses the best reply slope $\overline{B}_D$, for given $p_U$.  

In the sequential monopoly, the firm upstream again acts as a monopolist but, by construction, when optimizing it internalizes the pass-through effect (the best response of $D$). So, by construction, firm $U$ has the residual demand $p_U^r$: it follows that the sequential monopoly is a Generalized SDFE, where the price impact functions are identical to the standad SDFE, except for the value of $\Lambda_{D,UU}$, which is set to zero. In particular, the matrix $\Lambda_D$ is independent of $\overline{B}_U$.
In the Supplementary Appendix, I prove that also Sequential Cournot oligopoly is an instance of Generalized SDFE.

\end{example}

The next Proposition proves the existence of a Generalized SDFE, and that the slopes remain strategic complements.

\begin{prop}
	
A Generalized SDFE exists. Moreover, the best replies are increasing, so there are always a maximal and a minimal equilibrium (possibly identical).


\label{teo_comparative}	
\end{prop}
The key idea for the proof comes once again by strategic complementarity: a lower price impact means higher slopes that, in turn, trigger higher best responses, and equilibrium slopes. This is the fundamental tool that allows to do comparisons of different restrictions on market power in the next section.
The proof is in Appendix \ref{proof_teo_comparative}.


\subsection{Comparison with unilateral and local market power}

As discussed in the Literature section, many papers in the production network literature assume, as a simplification, that input prices are taken as given, and that prices in other markets are taken as given: in the next definition, I characterize them, in words. The formal definition can be found in the proof of Proposition \ref{thm:specialcases}, in the Supplemental Appendix.

\begin{defi}

\begin{enumerate}
    \item The SDFE with \emph{unilateral market power} is the model in which firms optimize \eqref{BestReply} with the additional constraint that they take input prices as given. To be more precise, there is a partition of the set of goods $(\mathcal{N}^{down}(i),\mathcal{N}^{up}(i))$ such that $\mathcal{N}^{in}(i)\subseteq \mathcal{N}^{up}(i)$ and $out(i)\in \mathcal{N}^{down}_i$, and the firm takes the prices of goods in $\mathcal{N}^{up}(i)$ as given. 

\item The SDFE with \emph{local market power} is the model in which firms optimize \eqref{BestReply} with the additional constraint that they take as given the prices of all goods that are not their output or inputs. 
\end{enumerate}
\label{special_cases}
\end{defi}  
In the unilateral case, the choice of $\mathcal{N}^{up}(i)$ is not unique. It turns out that these details do not really matter for the following results. For example, in the supply chain with layers, it turns out that all the choices satisfying the above assumption generate the same price impact function \eqref{impact_unilateral}. The intuition is simple: if firms take inputs as given in a chain, then they automatically also take as given the further upstream prices. If the network is more complicated, this is not necessarily true anymore, but the results below hold anyway.

The next Proposition shows that these two are Generalized SDFEs. The proof is in Appendix \ref{proof:thm:specialcases}.
\begin{prop}

The models of Definition \ref{special_cases} are special cases of the Generalized SDFE, for different choices of the price impact functions $\Lambda_i$. 
  
\begin{enumerate}
\item \textbf{Unilateral market power.} 
\begin{equation}
\Lambda_i^{\text{unilateral}}(\overline{B}_{-i})=\begin{pmatrix}
    [(M-\hat{B}_i)^{-1}_{\mathcal{N}^{down}(i)}]_{out(i),out(i)} & \boldsymbol{0}' \\
    \boldsymbol{0} & O
\end{pmatrix},
\label{impact_unilateral}
\end{equation}
where $\boldsymbol{0}$ is a vector of length $d_i-1$ (if $d_i=1$ then the matrix reduces to the upper left element).

\item \textbf{Local market power}: 
\begin{equation}
\Lambda^{\text{local}}_i(\overline{B}_{-i})=([M-\hat{B}_i]_{\mathcal{N}(i)})^{-1}
\end{equation}
	
\end{enumerate}

\label{thm:specialcases}
\end{prop}
The Proposition generalizes the calculation done naively in Section \ref{sec:simple}, where we simply defined the price impact on inputs to be zero: the above expression $\Lambda_i^{\text{unilateral}}$ shows that this is true in general.
We can see that $\Lambda_D^{\text{unilateral}}\le \Lambda_D^{\text{multilateral}}$ in the p.s.d order. This turns out to be a general fact.

In principle, one could also define unilateral market power in the reverse way, assuming that firms are price-takers on the \emph{output} market, but price-setters in \emph{input} markets. This is very rarely done, so I present the analysis with price-taking on inputs: the mechanisms would be analogous to the reverse assumption, complicated by the fact that there may be many inputs. In Proposition \ref{prop:markup_generalized}
I analyze the reverse case for the layered supply chain.

Now we show the main result of the section, illustrating the qualitative implications of the assumptions of local and unilateral market power.
\begin{teo}
\begin{enumerate}
    \item   In both the unilateral and local SDFE, the equilibrium slopes are higher (so the aggregate price impacts are smaller) than with multilateral market power: for all $ i \in \mathcal{N}$, in any equilibrium $\overline{B}_i^{\text{unilateral}}>\overline{B}_i^{\text{multilateral}}$ and $\overline{B}_i^{\text{local}}>\overline{B}_i^{\text{multilateral}}$.
    \item Consider two firms, $i$ and $j$, with $i$ selling to $j$, and consider all the other assumptions of Theorem \ref{thm:markups} part 1. 
The result of Theorem \ref{thm:markups} remains true in the (minimal equilibrium of the) unilateral SDFE, so that $i$ has a higher markup: $\mu_{i}^{out}>\mu_{j}^{out}$. Moreover,  $i$ also has a higher profit: $\pi_{i}>\pi_{j}$. 
Moreover, the same result remains true if there are $n^{*}$ identical firms producing good $i$, and their only customers are $n^{*}$ identical firms producing good $j$, and in the multilateral SDFE $\overline{B}_{i}<\overline{B}_{j}$. 

\end{enumerate}
  \label{thm:main_comparative}
\end{teo}

The proof is in Appendix \ref{proof:main_comparative}. In words, the Theorem illustrates the consequences of constraining market power to be unilateral, or local. The consequences affect both the size and the distribution of surplus. 

Part 1 shows that, in both the local and unilateral cases, the equilibrium aggregate price impacts are smaller, meaning that market power distortions are less powerful. 
This shows that countervailing market power actually harms consumers. Out of equilibrium, buyer power indeed makes the seller less willing to increase the price. However, this does not consider the fact that the increase in price is passed through to consumers. 
The mechanism is the one illustrated in \ref{sec:pass-through}: when the firm takes some prices as given, the perceived elasticity of demand (and supply) is lower, which, through the standard mechanism, leads to lower markups and markdowns. Strategic complementarity means that this effect remains true in equilibrium. 
Technically, the result refers to equilibrium slopes, because it is the result available in higher generality. Adding more structure, it is possible to show more precise implications on the final price or welfare, as formalized in the following Corollary.

\begin{cor}

	\label{localteo}
Suppose for simplicity that $\boldsymbol{f}_{L}=0$. 
Suppose that the consumer only consumes one good, say good 0: $\mathcal{C}=\{0\}$, then the price of the final good $p_0$ is smaller in with local or unilateral market power than with multilateral.

In the supply chain with layers, the welfare is lower with multilateral market power.


\end{cor}

Part 2 shows that, furthermore, in the unilateral case, surplus \enquote{moves upstream} with respect to the multilateral case. This follows from the asymmetry of the pass-through when input prices are taken as given. Theorem \ref{thm:markups} showed that, if the vertical connection between $i$ and $j$ is \enquote{more important} than the horizontal one (as made precise by the assumptions of part 1 of that theorem), then
markups and markdowns tend to move in opposite directions when moving upstream the supply chain. With unilateral market power, the result on the ranking of markups remains true. However, since markdowns are now zero, there is no contrasting force, so the result translates to profits: the upstream firm has higher profit too.



\begin{example}{\textbf{The division of surplus in the layered supply chain}}

The supply chain with homogeneous layers $n_i=n^*$ dramatically shows how the assumption of unilateral market power can completely change the surplus distribution in the supply chain.
Each pair of consecutive layers $i$ and $j:=i-1$ satisfies the assumptions of the second part of Theorem \ref{thm:main_comparative}. So, we can conclude that profits and markups are increasing going upstream: $\pi_i>\pi_{i-1}$, $\mu^{out}_i>\mu_{i-1}^{out}$. The order is diametrically opposite if we consider unilateral market power on inputs, as the Proposition below formalizes. Consider the difference with the result of Proposition \ref{prop:markup} showing that, with multilateral market power, all firms have the same profit. In that context, the markups are increasing upstream as with unilateral market power, but the presence of the markdowns balances the asymmetry, so that all firms have the same profit. 
\begin{prop}
Consider the supply chain with homogeneous layers: $n_i=n^*$ of Proposition \ref{prop:markup}.	
\begin{enumerate}
    \item 	If firms take the \emph{input} price as given (the price impact is as in \eqref{chain_output} below), then for all $i$: $\mu_i^{out}>\mu_{i-1}^{out}$ and $\pi_i>\pi_{i-1}$.
    
  \item  If firms take the \emph {output} price as given, then for all $i$: $\mu_i^{in}<\mu_{i-1}^{in}$ and $\pi_i<\pi_{i-1}$.
\end{enumerate}

\label{prop:markup_generalized}	
\end{prop}
Moreover, part 1 is also true in the classic Sequential Cournot model.
\begin{prop}
   In the layered supply chain of Example \ref{ex:line}, the Sequential Cournot model is a Generalized SDFE where, for all $i<N$: 
\[
\Lambda_i^{sequential}=\begin{pmatrix}
    \overline{\Lambda}_i^{out} & 0\\
    0 & 0
\end{pmatrix}
\]
Moreover, the markups and profits are increasing upstream.

\label{sequential_cournot}
\end{prop}
To understand the intuition, compare with the Generalized SDFE with unilateral market power. Equation \eqref{impact_unilateral}, in the case of the supply chain, becomes:
\begin{align}
\Lambda^{unilateral}_i=\begin{pmatrix}
   \left( (\overline{\Lambda}_i^{out})^{-1}+(n_i-1)\overline{B}_i\right)^{-1}  & 0 \\
0 & 0
\end{pmatrix}
\label{chain_output}
\end{align}
The sequential Cournot model combines the assumption of unilateral market power (because input prices are taken as given) with setting to zero the slope of competitors, as in the Cournot model in Equation \eqref{priceimpact_cournot}. The markups are $\mu_i^{out}=\overline{\Lambda}_i^{out}q_i$, and are increasing upstream with analogous intuition. The markdowns are zero, so the profits are increasing upstream.

\end{example}

\begin{example}{\textbf{Local market power and the size of distortions}}

Theorem \ref{teo_comparative} is a qualitative result. Again, the distortion is larger when the vertical network dimension is more important, because the distinction between local and multilateral market power bites more. Indeed, the next Proposition shows that, in the case of a supply chain with $N$ layers, the relative welfare of the two models can be unbounded. The proof is in Appendix \ref{proof_relativeloss}.
\begin{prop}

Consider a layered production chain of $N$ layers, with $2$ firms per layer. Denote $W^{\text{multilateral}}$ the welfare in the standard SDFE, and $W^{\text{local}}$ the welfare in the model with local market power. 

If $N$ goes to infinity, the ratios of quantities $\dfrac{Q^{\text{local}}}{Q^{\text{multilateral}}}$ and of welfares $\dfrac{W^{\text{local}}}{W^{\text{multilateral}}}$ go to $\infty$.

\label{relativeloss}
\end{prop}

\end{example}

\begin{oss}[\textbf{The general network interpretation}]

    In terms of the network interpretation of Section \ref{sec:network}, unilateral market power can be understood as the situation where only the subnetwork containing the output good is considered. 
    Local market power amounts to the situation in which only the subnetwork containing direct inputs and the output is considered.
    
\end{oss}


\subsection{Application: welfare impact of vertical mergers}
\label{sec:applications}

This section illustrates an applications in which multilateral market power qualitatively affects the welfare evaluation of vertical mergers. The previous general results compared multilateral market power with the Generalized SDFE with unilateral market power, for a clean comparison. Since the most commonly used model is the sequential Cournot, in this section, I compare the multilateral market power to both the unilateral model and the sequential Cournot.


\label{ex:vertical_merger}    

Consider the of the supply chain with layers of Example \ref{ex:line}. In particular, suppose that $N=2$, $n_2=1$, so that there is only one firm in the upstream layer. Suppose that there is a merger between the upstream firm and one of the downstream firms. Suppose further that the merged firm does not sell its intermediate good to others, but it keeps it all to produce the final output, so that the economy becomes a monopoly. The pre and post-merger situations are illustrated in Figure \ref{Fig:vertical}. 

\begin{prop}
Consider a layered supply chain with $N=2$, $n_2=1$, $f_{1,L}=f_{2,L}=0$, $k_2\to \infty$ and $k_1=k$. In the described merger setting, there is an interval $(n_{\ast},n^{\ast})$ such that if $n_1 \in (n_{\ast},n^{\ast})$ the merger is welfare-increasing with multilateral market power, but welfare-decreasing with unilateral market power. The same happens comparing multilateral market power and Sequential Cournot.
\label{prop_vertical}
\end{prop}

The formal proof is in Appendix \ref{proof_prop_vertical}. A vertical mergers can be welfare improving or not, depending on the standard trade-off between decreasing double marginalization and decreasing competition through foreclosure \citep{belleflamme2015industrial}. Multilateral market power affects the balance of the two forces. Both the unilateral market power model and the sequential Cournot model share the feature that input prices are taken as given, leading to underestimate the extent of double marginalization. As a consequence both unilateral SDFE and (in some cases) the sequential Cournot tend to overestimate the welfare loss by a vertical merger with respect to the multilateral SDFE.



\begin{figure}
    
    \begin{subfigure}{0.5\textwidth}
        
        \begin{tikzpicture}[mycircle/.style={
                circle,
                draw=black,
                fill=gray,
                fill opacity = 0.3,
                text opacity=1,},
            node distance=  1.5cm,
            every newellipse node/.style={inner sep=0pt}
            ]
            
            \tikzset{cross/.style={cross out, draw=black, fill=none, minimum size=20pt, inner sep=0pt, outer sep=0pt}, cross/.default={2pt}}
            
            \node [mycircle] (1) {2};
            
            \node [mycircle, draw=black, below of =1] (0B) {1b};
            \node [mycircle, draw=black, left of =0B] (0A) {1a};
            \node [mycircle,draw=black, right of =0B] (0C) {$1,n_1$};	
            
            \draw (0B) node[cross,red, rotate=8] {};
            \draw (0C) node[cross,red, rotate=8] {};
            
            \path let
            \p1 = ($(0A.south west)-(1.north east)$),
            \n1 = {veclen(\p1)}
            in
            (0A)--(1)	 	
            node[ellipse, draw, sloped, midway, blue, minimum width=\n1+2pt, minimum height=\n1/2, inner sep=2pt]  {}; 	
            
            \node [ below of =0B] (c) {Consumers};

            \foreach \i/\j in {
                1/0A,
                1/0B,
                1/0C,
                0A/c,
                0B/c,
                0C/c}
            \draw [->] (\i) -- node[above] {} (\j);

        \end{tikzpicture}

    \end{subfigure}
    \begin{subfigure}{0.5\textwidth}
        \begin{tikzpicture}[mycircle/.style={
                circle,
                draw=black,
                fill=gray,
                fill opacity = 0.3,
                text opacity=1,},
            node distance=  1.5cm
            ]
            
            \node [mycircle, draw=blue] (1) {2+1a};
            
            \node [above of =1, node distance=1cm] (text) {Merged firm};
            \node [ below of =1] (c) {Consumers};

            \foreach \i/\j in {
                1/c}
            \draw [->] (\i) -- node[above] {} (\j);

        \end{tikzpicture}	
        
    \end{subfigure}
    \caption{Left: pre-merger economy. The blue circle indicated the merging firms 2 and $1a$. Right: the economy after the merger: $1b$ and $1c$ are driven out of the market because the merged firm does not sell them the necessary input anymore, and the merged firm becomes a monopolist.}
    \label{Fig:vertical}
    
\end{figure}
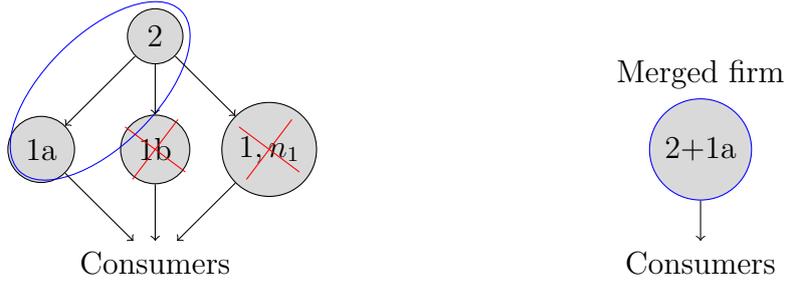

\section{Conclusion}
\label{conclusion}

This paper shows that competition in schedules provides a tractable way to model oligopoly in general equilibrium, and provides a unified framework for different competition models. Moreover, simplifying assumptions on market power can affect systematically both the magnitude and the distribution of surplus. This suggests that, when modeling complex networks of large firms with market power, a model with endogenous market power may be desirable.

\appendix

\section*{Appendix}

\section{Proofs of Section \ref{solution}}

\subsection{Proof of Lemma \ref{pricingfunction}}
\label{proof_pricingfunction}

For this and the following proof, we need the following Lemma.
\begin{lemma}
    The viability and consumer connectedness assumptions imply that for any $i$, the matrix $M-\hat{B}_i=\sum_{j\neq i}\overline{B}_j\hat{\vv}_j\hat{\vv}_j'+\hat{B}_c$ is positive definite. 

    \label{lemma-invertibility}
\end{lemma}

Using the schedules \eqref{linear_schedule}, the market-clearing conditions \eqref{sub_mktclear} can be written as:
\[
\left(\sum_i \hat{B}_i+\hat{B}_c\right)\pv-\sum_i \hat{B}_{i,f}=\boldsymbol{A}
\]
So, Equation \eqref{mktclear_linear} follows, with the definitions of $M$ and $M_f$ in the text. In particular, using \eqref{coefficient_equation} we get: $M=\sum_i\overline{B}_i\hat{\vv}_i\hat{\vv}_i'+\hat{B}_c$.

Now since $\hat{B}_i=\overline{B}_i\hat{\vv}_i\hat{\vv}_i'$ is positive semidefinite, Lemma \ref{lemma-invertibility} implies that $M=M-\hat{B}_i+\hat{B}_i$ is positive definite, and so invertible. \qed

 \subsection{Proof of Lemma \ref{lemma_priceimpact}}
\label{proof_lemma_priceimpact}

By Lemma \ref{lemma-invertibility}, $M-\hat{B}_i$ is positive definite, so the inverse exists and is positive definite. Then, $\Lambda_i$ is positive definite because it is a principal submatrix of a positive-definite matrix.
Concerning the monotonicity, $\hat{B}_j=\overline{B}_j\hat{\vv}_j\hat{\vv}_j'$, so $\hat{B}_j$ is increasing in the coefficient $\overline{B}_j$. Moreover, $(M-\hat{B}_i)^{-1}=(\sum_{j\neq i} \hat{B}_j+\hat{B}_c)^{-1}$ is decreasing in each $\hat{B}_j$, and passing to a principal submatrix preserves the positive semidefinite ordering, so $\Lambda_i$ is decreasing in each $\hat{B}_j$.  \qed

\subsection{Proof of Theorem \ref{thm:existence}}
\label{existence_proof}

Now, consider the game $\mathcal{G}'=(\mathcal{N},(X_i,U_i)_{i\in\mathcal
N})$ with payoffs $U_i=\ln \pi_i(e^{x_1},\ldots,e^{x_n})$, where $X_i=\mathbb{R}$. Since both $\ln(\cdot)$ and $exp(\cdot)$ are monotone, a profile $\overline{B}$ is a pure Nash equilibrium of $\mathcal{G}$ if and only if the profile $x=(\ln \overline{B}_1, \ldots, \ln \overline{B}_n)$ is a pure Nash equilibrium of  $\mathcal{G}'$. It follows that the Nash equilibrium of $\mathcal{G}'$ is unique if and only if the Nash equilibrium of $\mathcal{G}$ is unique. Moreover, since the log is order-preserving, the game $\mathcal{G}$ is also supermodular. In the following, we analyze the game $\mathcal{G}'$.

\paragraph{Existence}
 
The best reply equation \eqref{bestreply} shows that $\overline{B}_i\in [0,k_i]$ and is continuous. So, by Brouwer's fixed-point theorem, there exists an equilibrium.


\paragraph{Potential}

To show that the modified game $\mathcal{G}'$ is a potential game, we show that the second cross-derivatives of the payoffs are equal. To compute the derivative of the payoffs, we must differentiate the matrix $M^{-1}$:
\begin{equation}
    \dfrac{\partial}{\partial \overline{B}_i} M^{-1}=-M^{-1}\dfrac{\partial}{\partial \overline{B}_i} \left( \sum_j \overline{B}_j\boldsymbol{\hat{\vv}}_j\boldsymbol{\hat{\vv}'}_j +\hat{B}_c\right)M^{-1}=-M^{-1}\boldsymbol{\hat{v}}_i\boldsymbol{\hat{v}'}_iM^{-1}
\label{derivative_rule}
\end{equation}
Since $M_f=\sum_j \hat{B}_{j,f}$ and $B_{j,f}=f_{j,L}\overline{B}_j\vv_j$, we have: $M_f=\sum_j \overline{B}_j\hat{\vv}_j f_{j,L}=\sum_j \overline{B}_j\hat{\vv}_j \hat{\ev}_j\boldsymbol{f}_{L}$, where $\hat{\ev}_j$ is the $j$-th canonical basis vector in $\mathbb{R}^n$. So:
\[
\dfrac{\partial}{\partial \overline{B}_i} M_f=\dfrac{\partial}{\partial \overline{B}_i} \sum \overline{B}_j \hat{\vv}_j\hat{\ev}_j'\boldsymbol{f}_L=\hat{\vv}_i\hat{\ev}_i'\boldsymbol{f}_L
\]
Moreover, $\hat{\vv}_i'\pv=\hat{\vv}_i'M^{-1}(\boldsymbol{A}+M_f)$. So:
\begin{align}
\dfrac{\partial }{\partial \overline{B}_i}\hat{\vv}_j'\pv&=-\hat{\vv}_j'M^{-1}\hat{\vv}_i\hat{\vv}'_iM^{-1}(\boldsymbol{A}+M_f)+\hat{\vv}_j'M^{-1}\hat{\vv}_i\hat{\ev}_i'\boldsymbol{f} \nonumber\\
&=-\hat{\vv}_jM^{-1}\hat{\vv}_i(\hat{\vv}'_i\pv-f_{i,L})\nonumber 
\end{align}

Using this,  the derivative of the profit is:
\begin{align*}
  \dfrac{\partial\pi_{i}}{\partial\overline{B}_{i}}		&=\overline{B}_{i}\left(1-\dfrac{1}{2k_i}\overline{B}_{i}\right)\left(\hat{\vv}_i'\pv-f_{i,L}\right)^{2}\left(\dfrac{1-\overline{B}_{i}/k_i}{\overline{B}_{i}\left(1-\dfrac{1}{2k_i}\overline{B}_{i}\right)}-2\hat{\vv}_{i}'M^{-1}\hat{\vv}_{i}\right)
\end{align*}
Notice that with our reparameterization $\overline{B}_i=e^{x_i}$ and $\dfrac{\partial U_{i}}{\partial x_{i}}=\dfrac{\partial\ln\pi_{i}}{\partial\ln\overline{B}_{i}}$. Since $\pi_i=\overline{B}\left(1-\dfrac{1}{2k_i}\overline{B}_{i}\right)\left(\hat{\vv}_i\pv-f_{i,L}\right)^{2}$, the derivative of $U_i$ becomes: 
\begin{align*}
    \dfrac{\partial U_{i}}{\partial x_{i}}
	&=1-\dfrac{1}{2k_i}\dfrac{e^{x_i}}{\left(1-\dfrac{1}{2k_i}e^{x_i}\right)}-2e^{x_i}\hat{\vv}_{i}'M^{-1}\hat{\vv}_{i}=1-\dfrac{1}{2k_i}\dfrac{\overline{B}_{i}}{\left(1-\dfrac{1}{2k_i}\overline{B}_{i}\right)}-2\overline{B}_{i}\hat{\vv}_{i}'M^{-1}\hat{\vv}_{i}
    \end{align*}
    Differentiating again, and using \ref{derivative_rule}:
    \begin{align}
\dfrac{\partial^2 U_{i}}{\partial x_{j}\partial x_{i}}	&=\begin{cases}
2\overline{B}_{i}\overline{B}_{j}\hat{\vv}_{i}'M^{-1}\hat{\vv}_{j}\hat{\vv}_{j}'M^{-1}\hat{\vv}_{i}=2\overline{B}_{i}\overline{B}_{j}\left(\hat{\vv}_{i}'M^{-1}\hat{\vv}_{j}\right)^2& i\neq j\\
 =-\dfrac{1}{2k_i}\dfrac{\overline{B}_{i}}{\left(1-\dfrac{1}{2k_i}\overline{B}_{i}\right)^{2}}-2\overline{B}_{i}\hat{\vv}_{i}'M^{-1}\hat{\vv}_{i}  +2\overline{B}_{i}^2\left(\hat{\vv}_{i}'M^{-1}\hat{\vv}_{i}\right)^{2}  & i=j
\end{cases}
\label{hessian}
\end{align}
Since $\dfrac{\partial^2 U_{i}}{\partial x_{j}\partial x_{i}}=\dfrac{\partial^2 U_{j}}{\partial x_{i}\partial x_{j}}$, the game is a potential game. This means that there exists a twice differentiable function $\Psi$ such that: $\dfrac{\partial \Psi}{ \partial x_{i} }=\dfrac{\partial U_i}{\partial x_{i} }$ for each $i$ and each profile $x$. In particular, this means that, even without knowing the expression of $\Psi$, we know its Hessian matrix $H$, we have $H_{ij}=\dfrac{\partial^2\Psi}{\partial x_{i}\partial x_j}=\dfrac{\partial^2U_i}{\partial x_{i}\partial x_j}$.

\paragraph{Uniqueness}

Now, we prove that the potential is strictly concave. This proves that the game can have at most one Nash equilibrium. To prove it, we prove that the Hessian matrix $H$ is negative definite, by proving that $-H$ is strictly diagonally dominant. Sum the off-diagonal entries:
\begin{align*}
\sum_{i\neq j}\left|H_{ij}\right|&=\sum_{i\neq j}H_{ij}\\
    &=\sum_{i\neq j}2\overline{B}_{j}\overline{B}_{i}\left(\hat{\vv}_{i}'M^{-1}\hat{\vv}_{j}\right)^{2}	\\
&=2\overline{B}_{j}\hat{\vv}_{j}'M^{-1}\left(\sum_{i\neq j}\overline{B}_{i}\hat{\vv}_{i}\hat{\vv}_{i}'\right)M^{-1}\hat{\vv}_{j}\\
&=2\overline{B}_{j}\hat{\vv}_{j}'M^{-1}\left(M-\hat{B}_{c}-\overline{B}_{j}\hat{\vv}_{j}\hat{\vv}_{j}'\right)M^{-1}\hat{\vv}_{j}	\\
&=2\overline{B}_{j}\hat{\vv}_{j}'M^{-1}MM^{-1}\hat{\vv}_{j}-2\overline{B}_{j}\hat{\vv}_{j}'M^{-1}\left(\hat{B}_{c}+\overline{B}_{j}\hat{\vv}_{j}\hat{\vv}_{j}'\right)M^{-1}\hat{\vv}_{j}	\\
&<2\overline{B}_{j}\hat{\vv}_{j}'M^{-1}\hat{\vv}_{j}-2\overline{B}_{j}\hat{\vv}_{j}'M^{-1}\left(\overline{B}_{j}\hat{\vv}_{j}\hat{\vv}_{j}'\right)M^{-1}\hat{\vv}_{j},
\end{align*}
where the strict inequality is because $\hat{B}_{c}$ is positive semidefinite, and there must be at least a path from each firm $j$ to the consumer, so that $\left[M^{-1}\hat{\vv}_{j}\right]_{c}\neq0$. Comparing with the expression for $H_{jj}$ in Equation \eqref{hessian}, we conclude that  $\sum_{i\neq j}\left|H_{ij}\right|<-H_{jj}$. This implies that $-H$ is strictly diagonally dominant and so $H$ is negative definite. Hence, $\Psi$ is strictly concave and the Nash equilibrium is unique. Since $H$ is negative definite, in particular, it has negative diagonal: this means that the payoffs are strictly concave, so the FOCs \eqref{bestreply} are necessary and sufficient for the equilibrium. Since the payoffs are concave and the cross derivatives are positive, the game is also a supermodular game, and we can conclude that the unique Nash equilibrium is also the unique rationalizable action profile. 

\qed

\section{Proofs of Section \ref{sec:network}}	

\subsection{Proof of Theorem \ref{thm:markups}}
\label{proof:markups}

Because of the assumption that goods can be partitioned in two disconnected sets, the matrix $M-\hat{B}_i$ must be block diagonal: one block $M_{\mathcal{N}^{down}_i}$, corresponding to $\mathcal{N}^{down}_i$, and the other $M_{\mathcal{N}^{up}_i}$ to  $\mathcal{N}^{up}_i$. The same happens for the matrix $M-\hat{B}_j$. Then, also the price impact matrices are block diagonal, with one block $\Lambda_i^{out}$ corresponding to the output, and one block $\Lambda_i^{in}$, corresponding to the inputs.

We need the following Lemmas, proven in the Supplemental Appendix. Since firm $j$ has a unique input, denote it as: $in(j)$. The vector  $\boldsymbol{e}^{down}_j$ is the vector of the canonical basis with length $|\mathcal{N}^{down}(j)|$, with 1 in the position $out(j)$.  The vector  $\boldsymbol{e}^{up}_i$ is the vector of the canonical basis with length $|\mathcal{N}^{up}_i|$, with 1 in the position $in(i)$.

\begin{lemma}
Firm $j$ satisfies (actually, also $i$):
\begin{align}
\Lambda_{j}^{out}	&=\left(\left(M_{\mathcal{N}^{\text{down}}(j)}^{-1}\right)_{out(j),out(j)}^{-1} - \overline{B}_{j}\right)^{-1}\\
\Lambda_{i}^{in}&	=\left(\left(M_{\mathcal{N}^{\text{up}}(i)}^{-1}\right)_{in(i),in(i)}^{-1} - \overline{B}_{i}\right)^{-1}
\label{downstream_j}
\end{align}

\label{linear_algebra}
\end{lemma}

\begin{lemma}
   Consider any profile of slopes $\overline{B}$, and assume that $i$ and $j$ satisfy the assumptions of Theorem \ref{thm:markups} part 1. The elements relative to the output block of the (diagonal) price impact matrices are:
\begin{align}   
\Lambda_{i}^{out}	&=\begin{pmatrix}
M_{\mathcal{N}^{\text{down}}(j)} & -\overline{B}_{j}\boldsymbol{e}^{down}_{j}\\
-\overline{B}_{j}(\boldsymbol{e}^{down}_{j})' & \overline{B}_{j}+\Lambda_{-j}^{-1}
\end{pmatrix}_{out(i),out(i)}^{-1}=\left(\Lambda_{-j}^{-1}+\overline{B}_{j}-\overline{B}_{j}^{2}\left(M_{\mathcal{N}^{\text{down}}(j)}\right)_{out(j), out(j)}^{-1}\right)^{-1} \label{downstream_i}
\end{align}
where $0\le \Lambda_{-j}^{-1}<\overline{B}_{-j}$.

If they satisfy the assumptions of Part 2, then:
\begin{align}   
\Lambda_{j}^{in}	&=\begin{pmatrix}  M_{\mathcal{N}^{\text{up}}(i)} & -\overline{B}_{i}\boldsymbol{e}^{up}_{i}\\
-\overline{B}_{i}(\boldsymbol{e}^{up}_{i})' & \overline{B}_{i}+\Lambda_{-i}^{-1}
\end{pmatrix}_{in(j),in(j)}^{-1}=\left(\Lambda_{-i}^{-1}+\overline{B}_{i}-\overline{B}_{i}^{2}\left( M_{\mathcal{N}^{\text{up}}(i)} \right)_{in(i), in(i)}^{-1}\right)^{-1} ,
\end{align}
where $0\le \Lambda_{-i}^{-1}<\overline{B}_{-i}$.

    \label{lemma_downstream}
\end{lemma}

\textbf{Part 1}

We want to prove that $\mu_i^{out}>\mu_j^{out}$. Since $q_i^{out}\ge q_j^{out}$, it is sufficient to prove that $\Lambda_i^{out}>\Lambda_j^{out}$. Using Equation \eqref{downstream_j} to solve for $\left(M_{\mathcal{N}^{\text{down}}(j)}^{-1}\right)_{out(j),out(j)}^{-1}=(\Lambda_j^{out})^{-1}+\overline{B}_j$, and substituting into \eqref{downstream_i}, we can express $\Lambda_i^{out}$ as function of $\Lambda_j^{out}$:
\[
\Lambda_i^{out}=\left(\Lambda_{-j}^{-1}+\left(\dfrac{1}{\overline{B}_j}+\Lambda_j^{out}\right)^{-1}\right)^{-1}
\]
From which it follows that $\Lambda_i^{out}>\Lambda_j^{out}$ if $\Lambda_{-j}^{-1}$ is small enough. Since $\Lambda_{-j}^{-1}\le \overline{B}_{-j}$, this is true if $\overline{B}_{-j}$ is small enough.

\textbf{Part 2}

Using the second part of Lemma \ref{lemma_downstream}, if $\overline{B}_{-i}$ is small enough all the reasoning can be repeated, obtaining $\Lambda_i^{in}<\Lambda_j^{in}$.
\qed

\subsection{Proof of Proposition \ref{prop:markup}
}
\label{proof_markups}

The proof follows from the following lemmas, proven in the Supplementary Appendix.

\begin{lemma}
    Define $\overline{BR}(\overline{\Lambda},n,k)$ as the unique positive solution of:
\[
X=\left(k^{-1}+(\overline{\Lambda}^{-1}+(n-1)X)^{-1}\right)^{-1}
\]
Then, $\overline{BR}$ is increasing in $n$ and $k$, and decreasing in $\overline{\Lambda}$.

Moreover, define $\overline{\Lambda}_i(\overline{B}_{-i}):=\overline{\Lambda}_i^{out}+\overline{\Lambda}_i^{in}$. The best reply equations \eqref{bestreply_chain} can be expressed as, for all $i$:
\[
\overline{B}_i=\overline{BR}(\overline{\Lambda}_i(\overline{B}_{-i}),k_i,n_i)
\]
 \label{lemma:layer_best_reply}
\end{lemma}

\begin{lemma}

\label{compstat}
In equilibrium, each $\overline{B_i}^*$ is increasing in each $n_j$.
Moreover, if $k_i=k$ for all $i$, then for all $i,j$ $n_i \ge n_j$ implies $\overline{B_i}^* \ge \overline{B_j}^*$. 

\end{lemma}
\begin{lemma}
   If $N=1$ or $N\ge 2$ and $n_i\ge 2$, then the equilibrium slopes are nonzero.
    
    \label{lemma:nontrivial}
\end{lemma}

To arrive at the expression of the inverse residual schedule $\pv_i^r$, it is easier to start from the \emph{direct} residual demand and supply.
For firms in layer $1$, the slope of the (direct) residual demand is $B_{c,1}+(n_1-1)\overline{B}_1$. Firms in the upstream layer 2 face a demand $n_1\overline{B}_1(p_1-p_2)+(n_2-1)\overline{B}_2(p_2-p_3)$, where it is now necessary to solve the first layer equations for $p_1$ as a function of $p_2$. Proceeding iteratively, we can show that the direct residual demand and supply for firm $i$ are:
\begin{align}
q_i^{out}&=\underbrace{(\overline{\Lambda}_i^{out})^{-1}(A-p_i)}_{\substack{\text{Demand}\\\text{from customers}}}-\underbrace{(n_i-1)\overline{B}_i(p_i-p_{i+1}-f_{i,L})}_{\substack{\text{Supply of}\\\text{competitors}}} \nonumber\\
q_{i,i+1}&= \underbrace{(\overline{\Lambda}_i^{in})^{-1}p_{i+1}}_{\substack{\text{Supply}\\\text{from suppliers}}}-\underbrace{(n_i-1)\overline{B}_i(p_i-p_{i+1}-f_{i,L})}_{\substack{\text{Demand of}\\\text{competitors}}}
\label{residual_schedule_chain}
\end{align}
where: 
\begin{align}
\overline{\Lambda}_i^{out}&=\dfrac{1}{B_c}+\sum_{j<i} \dfrac{1}{n_j\overline{B}_j} \quad 
\overline{\Lambda}_i^{in}&=\sum_{j>i} \dfrac{1}{n_j\overline{B}_j}
\label{priceimpacts}
\end{align}
represent the slopes of the \emph{aggregate} demand for good $i$ and supply for good $i+1$. 

The right-hand sides of Equations \ref{residual_schedule_chain} constitute the (direct) residual demand and supply. Inverting this system, we obtain $\pv_i^r$ for firm $i$. To get the price impact we can just compute the Jacobian of the right-hand side and invert, i.e.:
\begin{align}
\Lambda_i&=\left(\begin{array}{cc}
(\overline{\Lambda}_i^{out})^{-1}+(n_i-1)\overline{B}_i & -(n_i-1)\overline{B}_i\\
-(n_i-1)\overline{B}_i &  (\overline{\Lambda}_i^{in})^{-1}+(n_i-1)\overline{B}_i
\end{array}\right)^{-1} \quad \text{ if } i<N \nonumber \\
\Lambda_N&=
\big((\overline{\Lambda}_N^{out})^{-1}+(n_N-1)\overline{B}_N\big)^{-1}, 
\label{eq:priceimpact_supplychain}
\end{align}
So, the markup-markdown vector is:
\begin{align}
    \boldsymbol{\mu}_{i}&=q_i^{out}\left(\begin{array}{c}
\dfrac{\overline{\Lambda}_{i}^{out}}{1+\left(\overline{\Lambda}_{i}^{in}+\overline{\Lambda}_{i}^{out}\right)(n_{i}-1)\overline{B}_{i}}\\
-\dfrac{\overline{\Lambda}_{i}^{in}}{1+\left(\overline{\Lambda}_{i}^{in}+\overline{\Lambda}_{i}^{out}\right)(n_{i}-1)\overline{B}_{i}}
\end{array}\right) \quad \text{if } i<N \nonumber \\
\boldsymbol{\mu}_{N}&=\mu_N^{out}=q_N^{out}\dfrac{\overline{\Lambda}_{N}^{out}}{1+\left(\overline{\Lambda}_{N}^{in}+\overline{\Lambda}_{N}^{out}\right)(n_{N}-1)\overline{B}_{N}}
 \label{markups_supplychain}
\end{align}
and the first order conditions \eqref{bestreply} reduce to:
\begin{equation}
\overline{B}_i=\left(\dfrac{1}{k_i}+\dfrac{1}{(n_i-1)\overline{B}_i+\left(\overline{\Lambda}_{i}^{in}+\overline{\Lambda}_{i}^{out}\right)^{-1}}\right)^{-1},
\label{bestreply_chain}
\end{equation}

\begin{enumerate}
\item Suppose $n_i=n_j:=n^*$. Denote $Q$ the quantity consumed by the consumer in equilibrium. By Lemma  \ref{compstat}, in equilibrium we have homogeneous slopes $\overline{B}_i=\overline{B}_j:=B^*$. By homogeneity and market clearing $q_i^{out}=q_j^{out}:=Q/n^*$ for any $i$, $j$.  As a consequence, by the expressions \eqref{priceimpacts}, we have that the quantity  $\overline{\Lambda}:=\overline{\Lambda}_i^{in}+\overline{\Lambda}_i^{out}=B_c^{-1}+(n^*)^{-1}\sum_{j\neq i}(B^*)^{-1}=B_c^{-1}+(n^*)^{-1}(N-1)(B^*)^{-1}$ is independent of $i$. 
From Equation \eqref{priceimpacts} we have that $\overline{\Lambda}_i^{out}$ is increasing in $i$ (i.e., increasing upstream) while $\overline{\Lambda}_i^{in}$ is decreasing in $i$ (i.e., increasing downstream). By Equations \eqref{markups_supplychain}, the same is true of, respectively, the markup and markdown.
    \item The aggregate profit in layer $i$ is $\Pi_i=n_i\dfrac{Q^2}{n_i^2}\left(\dfrac{1}{\overline{B}_i}-\dfrac{1}{2k}\right)=\dfrac{Q^2}{n_i}\left(\dfrac{1}{\overline{B}_i}-\dfrac{1}{2k}\right)$ so, by Lemma \ref{compstat}, $\Pi_i\le \Pi_j$ if and only if  $n_i \ge n_j$.  \qed

\end{enumerate}

\section{Proofs of Section \ref{sec:multilateral}}

\subsection{Proof of Proposition \ref{teo_comparative}}	
\label{proof_teo_comparative}	

The assumptions on $\Lambda_i$ and Equation \ref{bestreply} show that $\overline{B}_i\in [0,k_i]$ and $\Lambda_i$ is continuous, so the best reply map is also continuous. So, by Brouwer's fixed-point theorem, there exists an equilibrium.
Moreover,  the best reply is increasing in the profile of slopes of other firms $\overline{B}_{-i}$. By Topkis' Theorem, the equilibrium set is a lattice, so it has a maximal and minimal element.
\qed



\subsection{Proof of Proposition \ref{thm:specialcases}}
\label{proof:thm:specialcases}

Both the unilateral and the local versions are cases in which there is a subset of goods whose prices are taken as given. Define $\overline{\mathcal{N}}(i)\subseteq \mathcal{M}$ the set of goods whose price changes are not taken as given by firm $i$.  Local market power is the case $\overline{\mathcal{N}}(i)=\mathcal{N}(i)$. Unilateral market power is the case in which $\overline{\mathcal{N}}(i)=\mathcal{N}^{down}(i)$.

The problem of the firm is the analogue of  \eqref{BestReply}, except that the firm only internalizes the effect of quantity on prices of goods in the set $\overline{\mathcal{N}}(i)$: so, the firm only internalizes as constraints the subset of market-clearing conditions for goods with indices in $\overline{\mathcal{N}}(i)$:
\begin{equation}
    [\hat{\qv}_i+(M-\hat{B}_i)\pv]_{\overline{\mathcal{N}}(i)}=\overline{\boldsymbol{A}}_{\overline{\mathcal{N}}(i)}
    \label{market_clearing_restricted}
\end{equation}
In particular, since the $\pv_{\mathcal{N}(i)\setminus \overline{\mathcal{N}}(i)}$ are taken as given, the only prices that are allowed to vary are $\pv_{\mathcal{N}(i)\cap \overline{\mathcal{N}}(i)}$ Hence:
\begin{align*}
\max_{\qv_i,\pv_{\mathcal{N}(i)\cap \overline{\mathcal{N}}(i)}} & \qv_i'\pv_i-\ell_i\\
& \text{subject to: the restricted market-clearing } \eqref{market_clearing_restricted}, \nonumber\\
& \text{and the technology constraint } \eqref{tech_constraint_min}
\end{align*}
As in \eqref{BestReply}, we ignore the schedule constraint, showing that the optimum already satisfies it.
 We can express the equations more conveniently by decomposing the matrix $M$ as follows:
\[
M=\begin{pmatrix}
    M_{\overline{\mathcal{N}}(i)} & M_{\overline{\mathcal{N}}(i),\overline{\mathcal{N}}(i)^c}\\
    M_{\overline{\mathcal{N}}(i),\overline{\mathcal{N}}(i)^c}' & M_{\overline{\mathcal{N}}(i)^c}
\end{pmatrix}
\]
We solve for the prices in $\overline{\mathcal{N}}(i)$:
\begin{align*}
\pv_{\overline{\mathcal{N}}(i)}&=(M_{\overline{\mathcal{N}}(i)}-\hat{B}_{i,\overline{\mathcal{N}}(i)})^{-1} \times \\
&\left([\boldsymbol{\overline{A}}-\hat{\qv}_i]_{\overline{\mathcal{N}}(i)}-(M_{\overline{\mathcal{N}}(i),\overline{\mathcal{N}}(i)^c}-\hat{B}_{i,\overline{\mathcal{N}}(i),\overline{\mathcal{N}}(i)^c})\pv_{\overline{\mathcal{N}}(i)^c} \right)
\end{align*}
The vector $[\hat{\qv}_i]_{\overline{\mathcal{N}}(i)}$ only depends on quantities of goods in $\mathcal{N}(i)\cap \overline{\mathcal{N}}(i)$, all the other entries are zero.
So, selecting the vector of prices of traded goods $\pv_{\mathcal{N}(i)\cap \overline{\mathcal{N}}(i)}$, we can write the residual schedule as: 
\[
\pv^r_{\mathcal{N}(i)\cap \overline{\mathcal{N}}(i)}(\qv_i)=[(M_{\overline{\mathcal{N}}(i)}-\hat{B}_{i,\overline{\mathcal{N}}(i)})^{-1}]_{\mathcal{N}(i)\cap \overline{\mathcal{N}}(i)}(\boldsymbol{\overline{A}}_{\mathcal{N}(i)\cap \overline{\mathcal{N}}(i)}-\qv_{\mathcal{N}(i)\cap \overline{\mathcal{N}}(i)})+const
\]
where $const$ denotes terms that do not depend on quantities. So, the price impact on the relevant prices is:
\[
\dfrac{\partial \pv^r_{\mathcal{N}(i)\cap \overline{\mathcal{N}}(i)}}{\partial\qv_{\mathcal{N}(i)\cap \overline{\mathcal{N}}(i)}}=-[(M_{\overline{\mathcal{N}}(i)}-\hat{B}_{i,\overline{\mathcal{N}}(i)})^{-1}]_{\mathcal{N}(i)\cap \overline{\mathcal{N}}(i)}
\]
So, the $\Lambda_i$ matrix is:
\begin{equation}
\Lambda_i=-\dfrac{\partial \pv^r_i}{\partial \qv_i}=\begin{pmatrix}
    [(M_{\overline{\mathcal{N}}(i)}-\hat{B}_{i,\overline{\mathcal{N}}(i)})^{-1}]_{\mathcal{N}(i)\cap \overline{\mathcal{N}}(i)} & 0\\
     0 & 0
\end{pmatrix}
\label{priceimpact_constrained}
\end{equation}
where the zeros are present only if $\mathcal{N}(i)\setminus \overline{\mathcal{N}}(i)\neq 0$ (as in the unilateral case).
 \qed

\subsection{Proof of Theorem \ref{thm:main_comparative}}
\label{proof:main_comparative}

\textbf{Part 1}
First, we prove that for all $i$ and all profiles $\overline{B}$ we have, in the p.s.d. ordering:
\begin{align*}
    \Lambda^{\text{local}}_i(\overline{B}_{-i})&\le \Lambda^{\text{multilateral}}_i(\overline{B}_{-i})\\
    \Lambda^{\text{unilateral}}_i(\overline{B}_{-i})&\le \Lambda^{\text{multilateral}}_i(\overline{B}_{-i})
\end{align*}

Now, in both cases we can partition the goods in two sets $\overline{\mathcal{N}}(i)$ and $\overline{\mathcal{N}}(i)^c$, as in the proof of Lemma \ref{thm:specialcases}. Reordering, we can write $M-\hat{B}_i$ as a block matrix:
\[
M-\hat{B}_i=\begin{pmatrix}
      A_1 & A_2\\
    A_2' & A_3
\end{pmatrix}, \quad \text{so that } \Lambda_i^{multilateral}= \left[\begin{pmatrix}
      A_1 & A_2\\
    A_2' &  A_3
\end{pmatrix}^{-1}\right]_{\mathcal{N}_i}
\]
Using block inversion, the price impact in both constrained models can also be written as:
\begin{align}
    \Lambda_i=&\lim_{T\to \infty} \left[\begin{pmatrix}
      A_1 & A_2\\
    A_2' & T A_3
\end{pmatrix}^{-1}\right]_{\mathcal{N}_i}
\end{align}
Now note that, for $T>1$:
\[
\begin{pmatrix}
      A_1 & A_2\\
    A_2' &  TA_3
\end{pmatrix}- \begin{pmatrix}
      A_1 & A_2\\
    A_2' &  A_3
\end{pmatrix}=\begin{pmatrix}
     0 & 0\\
    0 & (T-1) A_3
\end{pmatrix}
\]
is positive semidefinite, so $\begin{pmatrix}
      A_1 & A_2\\
    A_2' &  TA_3
\end{pmatrix}\ge  \begin{pmatrix}
      A_1 & A_2\\
    A_2' &  A_3
\end{pmatrix}$ in the p.s.d. order. The inverse matrices have the opposite ranking. By passing to the limit
we obtain $\Lambda_i\le\Lambda_i^{multilateral}$, which is what we wanted to show.

Define $BR^{m}:\prod_i[0,k_i]\to \prod_i[0,k_i]$ the best reply map for the multilateral model and $BR$ the one for either the local or unilateral model.
From the results above, we have that for any profile $\overline{B}$ we have: $BR(\overline{B})>BR^{m}(\overline{B})$. Call $(\overline{B}^*)^m$ the unique equilibrium in the multilateral case
 and $\overline{B}^*$ the minimal equilibrium in either the unilateral or local model. We have:
\[
\overline{B}^*=BR(\overline{B}^*)>BR^m(\overline{B}^*)
\]
By Proposition \ref{teo_comparative}, the best reply is monotonic. So, iterating the best reply $BR^m$ starting from $B^*$ we eventually reach the equilibrium of the multilateral model, so:
\[
\overline{B}^*>BR^m(\overline{B}^*)>\cdots>(\overline{B}^*)^m
\]
which is what we wanted to show. 

\textbf{Part 2}
We need the following Lemma, proven in the Supplemental Appendix. Define $BR^{uni}$ the best reply map for the unilateral model. 

\begin{lemma}

 Consider any profile of slopes $\overline{B}$, and suppose that $i$, $j$ satisfy the assumptions of the Theorem, and in particular $\overline{B}_i\le \overline{B}_j$. Compute the best reply using the unilateral model $BR^{\text{uni}}_k(\overline{B}_{-k})$. We have that: $BR^{\text{uni}}_i(\overline{B}_{-i})<BR^{\text{uni}}_j(\overline{B}_{-j})$.
 \label{lemma_BR_unilateral}

\end{lemma}

We want to prove that in the unilateral equilibrium $\pi_{i}>\pi_{j}$, which is true if and only if $\overline{B}_{i}<\overline{B}_{j}$.

Consider the unique equilibrium with multilateral market power: $\overline{B}^{m}$. By Lemma \ref{thm:specialcases}, for every firm $k$ $\Lambda_k^{\text{uni}}(\overline{B}_{-k}^{m})\le \Lambda_k^{m}(\overline{B}_{-k}^{m})$. So, it follows $BR^{\text{uni}}_k(\overline{B}_{-k}^{m})\ge BR^{m}_k(\overline{B}^{m}_{-k})=\overline{B}^{m}_{k} $.  Applying iteratively the operator   $BR^{\text{uni}}_k$ , we obtain an increasing sequence of profiles, that must necessarily converge to the minimal equilibrium with unilateral market power $\overline{B}^{\text{uni}}$.

Moreover, Lemma \ref{lemma_BR_unilateral} proves that $BR^{\text{uni}}_i(\overline{B}_{-i}^{m})< BR^{uni}_j(\overline{B}^{m}_{-j}) $ remains true at every iteration: so, it must be (weakly) true in the minimal equilibrium with unilateral market power: $\overline{B}^{\text{uni}}_i\le \overline{B}^{uni}_j$. \qed 


\subsection{Proof of Corollary \ref{localteo}}


If there is a unique final good, say good $0$, then the vector $\boldsymbol{A}$ has just one nonzero entry, corresponding to good $0$. So, $\boldsymbol{A}'\pv=A_0p_0$, and as a consequence:
\[
p_0=\dfrac{1}{A_{0}}\boldsymbol{A}'M^{-1}\boldsymbol{A}
\]
Since $M$ is increasing in each matrix $B_i$ (or, equivalently, each coefficient $\overline{B}_i$), $p_0$ is decreasing in each $B_i$. 

In the supply chain with layers, we can explicitly compute total welfare using the utility function, and the fact that $Q=n_iq_i$:
\begin{align}
W&= \frac{A_c}{B_c}Q-\frac{1}{2B_c}Q^2-\frac{1}{2}\sum_in_i\dfrac{q_i^2}{k_i} =Q\left(\frac{A_c}{B_c}-\frac{1}{2B_c}Q-\frac{1}{2}\left(\sum_i\frac{1}{n_ik_i}\right)Q\right),
\label{welfare_chain}
\end{align}
where:
\[
Q=A\left(B_c\left(\sum_i \dfrac{1}{n_i\overline{B}_i}\right)+1\right)^{-1}
\]
The welfare is increasing in $Q$ for each $Q\le Q^*$, where $Q^*=A\left(B_c\left(\sum_i \dfrac{1}{n_ik_i}\right)+1\right)^{-1}$, which is the quantity in the perfect competition benchmark. So, for any parameters of a Generalized SDFE, welfare is increasing in $Q$, which means it is increasing in each equilibrium slope $\overline{B}_i$.
\qed

\subsection{Proof of Proposition \ref{relativeloss}}
\label{proof_relativeloss}

We need the following Lemma, proven in the Supplemental Appendix.
\begin{lemma}
    In the equilibrium with multilateral market power and $n_i=2$ for each $i$, the equilibrium slope $B^*$ satisfies $\lim_{N\to \infty}B^*=0$. 

    With local market power, there is a positive number $\underline{B}$, independent of $N$, such that each equilibrium slope $\overline{B}_i$ satisfies $\overline{B}_i>\underline{B}$.
    \label{lemma:limit}
\end{lemma}

Write for brevity $Q^l:=Q^{\text{local}}$, $Q^m:=Q^{\text{multilateral}}$ and the same for welfare. Using the expression for welfare derived in \eqref{welfare_chain}, we can write the ratio of welfare in the two models as:
\[
\frac{W^m}{W^l}=\frac{Q^m}{Q^l}\frac{\frac{A_c}{B_c}-\frac{1}{2B_c}Q^m-\frac{1}{2}\frac{N}{2}Q^m}{\frac{A_c}{B_c}-\frac{1}{2B_c}Q^l-\frac{1}{2}\frac{N}{2}Q^l},
\]
where:
\[
Q^l=A\left(B_c\left(\sum_i \dfrac{1}{2\overline{B}_i^{\text{local}}}\right)+1\right)^{-1},\quad Q^m=A\left(B_c\left(\dfrac{N}{2B^*}\right)+1\right)^{-1}
\]
Both go to zero as $N\to \infty$, so that $\lim_{N\to \infty} W^m/W^l=\lim_{N\to \infty} Q^m/Q^l$. Now, using the lower bound for the local slope:
\[
\lim_{N\to \infty} \frac{W^m}{W^l}=\lim_{N\to \infty}\frac{Q^m}{Q^l}\le \lim_{N\to \infty}\frac{B_c\left( \dfrac{N}{2\underline{B}}\right)+1}{B_c\left(\dfrac{N}{2B^*}\right)+1}
=\lim_{N\to \infty} 2 B^* \frac{B_c \frac{1}{2 \underline{B}}+1/N}{B_c+2B^*/N} =0
\]
and so also the ratio of welfares goes to 0. \qed

{
\small

\bibliographystyle{chicago}
\bibliography{biblio.bib}
}

\pagebreak

\section*{Supplementary Appendix of \\\enquote{Multilateral market power in input-output networks}}

\section{Additional Proofs}
\label{app_proofs}

\subsection{Proof of Lemma \ref{lemma-invertibility}}

We want to prove that $\xv'\left(M-\hat{B}_i\right)\xv=0$ implies $\xv=0$: so $M-\hat{B}_i$ is positive definite, hence invertible. Now, the quadratic form can be decomposed as: $\xv'\left(M-\hat{B}_i\right)\xv=\xv'\hat{B}_c\xv+\sum_{j\neq i}\overline{B}_j\xv'\hat{\boldsymbol{v}}_{j}\hat{\boldsymbol{v}}_{j}'\xv$. For this to be zero, it must be that $\xv_{\mathcal{C}}=0$ and, if $\xv\neq 0$, $\xv$ is orthogonal to all $\hat{\boldsymbol{v}}_j$ for $j\neq i$.

Now, we show that the set of vectors $V=\{\hat{\vv}_j\}\subset \mathbb{R}^m$, $j=1,\ldots,n$ must contain a basis of $\mathbb{R}^m$. Select a subset of firms $\mathcal{N}_0$ such that each firm produces a distinct good. Then, select the set of vectors $V_0=\{\hat{\vv}_j:\, j \in \mathcal{N}_0\}$. Each vector in this set has a 1 in a different position (the one corresponding to the output, which is distinct for all firms by construction). Now consider the submatrix $\overline{F}$ by selecting from $F$ only the rows corresponding to firms in $\mathcal{N}_0$. Then, $\overline{F}\in \mathbb{R}^{m \times m}$. 
Each vector $\hat{\vv}_j'$, with $j \in \mathcal{N}_0$, is a row of the matrix $I-\overline{F}$.  
By the assumption of viability, $I-\overline{F}$ is an M-matrix, and in particular is invertible \citep{horn1994topics}: it follows that the set $V_0$ (and so, $V$) contains a basis of the whole space. 

We have that $\xv$ must be orthogonal to $V_0\setminus\{\hat{\vv}_i\}$. Now there are two cases: either $V\setminus\{\hat{\vv}_i\}$ contains a basis, or it spans an $m-1$-dimensional space. In the first case, $\xv=0$ and the proof is concluded. In the second case it means that $xv$ must belong to the orthogonal orthogonal complement of $V\setminus\{\hat{\vv}_i\}$, which is one-dimensional. The orthogonal complement is the subspace parallel to the $i-$th column of the inverse matrix $L:=(I-\overline{F})^{-1}$. Now, the assumption of connectedness ensures that $L_{gi}>0$ if $g\in \mathcal{C}$. In fact, since $\overline{F}$ is an M-matrix, we have the series representation: $L=\sum_k [\overline{F}^k]_{ig}$. By the connectedness assumption, if $g\in \mathcal{C}$, then there is at least a path from $i$ to $g$: so, it follows $L_{gi}\neq 0$. Since $\xv$ is parallel to $L_i$, it must be $x_{ig}\neq 0$ for some $g\in \mathcal{C}$: but this contradicts $\xv_{\mathcal{C}}=0$. So, it must be that $\xv=0$. This proves that $M-\hat{B}_i$ is positive definite. \qed

\subsection{Proof of Lemma \ref{linear-residual}}
\label{proof_linear_residual}

I prove the result under the more general assumptions that the other firms $j\neq i$ use general linear schedules of the form \eqref{linear_schedule}. The fact that this does not affect the result justifies the restriction of the game to \eqref{schedule_expression}. 
 
\begin{lemma}
Suppose all players except $i$ use general linear schedules as in \eqref{linear_schedule}, such that $M-\hat{B}_i$ is invertible. Consider the monopoly problem \eqref{BestReply}, which I rewrite here:
\begin{align}
    &\max_{\qv_i, \ell_i,\pv_i}\boldsymbol{q}_{i}'\pv_{i}-\ell_i     \nonumber\\ 
\text{subject to:  }    &  \pv_i=\pv_i^r(\qv_i)=[(M-\hat{B}_i)^{-1}]_{\mathcal{N}(i)}(\boldsymbol{\tilde{A}}_i-\qv_i) \label{eq:res_schedule}\\
    &\text{the technology constraints \eqref{tech_constraint_min}}  \nonumber 
\end{align}
where $\boldsymbol{\tilde{A}}_i$ is a function of the schedule coefficients of other firms.

Then, if this problem has a solution, it is equivalent to the best-reply problem of player $i$, in the sense formalized in the main text.

\end{lemma}

The version in the main text follows by noting that, when the schedules of other players are restricted to the functional form \eqref{schedule_expression}, then Lemma \ref{pricingfunction} guarantees that $M-\hat{B}_i$ is invertible and $\Lambda_i$ is positive definite. These imply that the above assumption is satisfied, and the optimization has a solution because is concave.

\paragraph{Proof}

The market-clearing equations \eqref{mktclear_linear} can equivalently be rewritten \enquote{fixing} quantities $\qv_i$, as:
 \begin{align}
	& \hat{\qv}_i+(M-\hat{B}_i)\pv=\overline{\boldsymbol{A}}-\hat{B}_{i,f}, \label{eq:clearing} \\
    & \qv_i= B_i\pv_i-B_{i,f} \label{eq:schedule}
\end{align}
where $\overline{\boldsymbol{A}}:=\hat{\boldsymbol{A}}+M_f$. 
Since these implicitly define the pricing function, we can rewrite the optimization of \eqref{profits} by adding the above equations \eqref{eq:schedule} and \eqref{eq:clearing} as constraints:
\begin{align}
\max_{\qv_i,\pv_i, \overline{B}_i} \,& \qv_i'\pv_i-\ell_i\\
\text{subject to: }\quad 
&\qv_i=\overline{B}_i(\vv_i\pv_i-f_{i,L})\vv_i \label{eq:schedule3} \\
&\hat{\qv}_i+(M-\hat{B}_i)\pv=\overline{\boldsymbol{A}}-\hat{B}_{i,f} \nonumber \\
&\text{the technology constraint } \eqref{tech_constraint_min} \nonumber
\end{align}
To show that this is equivalent to \eqref{BestReply}, we have to show that Equation \eqref{eq:clearing} is the residual schedule $\pv_i^r(\cdot)$, and that the schedule constraint \eqref{eq:schedule3} is redundant. Now, the equation \eqref{eq:clearing} can solved for prices $\pv$:
\begin{equation}
\pv=(M-\hat{B}_i)^{-1}(\overline{\boldsymbol{A}}-\hat{B}_{i,f}-\hat{\qv}_i)
\label{prices_total}
\end{equation}
Now, reordering the equations so to have all the rows associated with inputs and outputs of $i$ first, the blocks of the matrix $M$ are:
\[
M-\hat{B}_i=\left(\begin{array}{cc}
	M_{\mathcal{N}(i)}-B_i & M_{\mathcal{N}(i),\mathcal{N}(i)^c}\\
	M_{\mathcal{N}(i),\mathcal{N}(i)^c}' & M_{\mathcal{N}(i)^c}
\end{array}\right)
\]
We can partially solve for the subset $\pv_i$:
\begin{align}
\pv_i& =  [(M-\hat{B}_i)^{-1}]_{\mathcal{N}(i)}\left([\overline{\boldsymbol{A}}-\hat{B}_{i,f}]_{\mathcal{N}(i)}-\qv_i-M_{\mathcal{N}(i),\mathcal{N}(i)^c}\pv_{-i}\right) \nonumber\\
&=[(M-\hat{B}_i)^{-1}]_{\mathcal{N}(i)}\left([\overline{\boldsymbol{A}}-\hat{B}_{i,f}]_{\mathcal{N}(i)}-\qv_i-M_{\mathcal{N}(i),\mathcal{N}(i)^c}(M_{\mathcal{N}(i)^c})^{-1}[\overline{\boldsymbol{A}}-\hat{B}_{i,f}]_{\mathcal{N}(i)^c}\right),\nonumber
\end{align}
where in the last step we use \eqref{prices_total} to solve express $\pv_{-i}$.
Now, defining $\boldsymbol{\tilde{A}}_i:=[\overline{\boldsymbol{A}}-\hat{B}_{i,f}]_{\mathcal{N}(i)}-M_{\mathcal{N}(i),\mathcal{N}(i)^c}(M_{\mathcal{N}(i)^c})^{-1}[\overline{\boldsymbol{A}}-\hat{B}_{i,f}]_{\mathcal{N}(i)^c}$, the right-hand side above becomes the expression of the residual schedule $\pv_i^r(\qv_i;\overline{B}_{-i})$ in  \eqref{eq:res_schedule}.

It remains to show that the schedule constraint is redundant. To do so, note that the FOC remains exactly the same as in the main text, implying the schedule equation  \eqref{eq:schedule3}.
In particular, if a solution of \eqref{BestReply} exists, it must satisfy \eqref{eq:schedule3}, which is the same as \eqref{schedule_expression}: so, the optimization \eqref{BestReply} is indeed equivalent to the best reply problem. \qed 

\subsection{Proof of Lemma \ref{linear_algebra}}
\label{proof_lemma_lambdas}

Both identities are actually true for any firm $i$ such that the price impact matrix is diagonal. In this case, in the unilateral SDFE the block relative to output is: $\Lambda_{i}^{out}= 
\left(M_{\mathcal{N}^{\text{down}}(i)} - \overline{B}_{i} [\hat{\boldsymbol{v}}_{i}\hat{\boldsymbol{v}}_{i}']_{\mathcal{N}^{\text{down}}(i)}\right)_{out(i),out(i)}^{-1}$. Now, the only good both in $\mathcal{N}^{\text{down}}(i)$ and $\mathcal{N}(i)$ is the output of $i$, so $[\hat{\boldsymbol{v}}_{i}]_{\mathcal{N}^{\text{down}}(i)} =\ev_i^{down}$, the relevant canonical basis vector. Using this notation, we can also rewrite:
\begin{align*}
\Lambda_i^{out}&=\left(M_{\mathcal{N}^{\text{down}}(i)} - \overline{B}_{i} \ev_i^{down}(\ev_i^{down})'\right)_{out(i),out(i)}^{-1} \\
&=(\ev_i^{down})'\left(M_{\mathcal{N}^{\text{down}}(i)} - \overline{B}_{i} \boldsymbol{e}_{i}^{down}(\ev_i^{down})'\right)^{-1}\ev_i^{down}
\end{align*}
Now for simplicity call $A:=M_{\mathcal{N}^{\text{down}}(i)} $, $x:=\overline{B}_{i}$ and $\boldsymbol{e}_i:= \ev_i^{down}$.
Using the Sherman-Morrison formula: 
\begin{align}
\left(A - x \boldsymbol{e}_{i}\boldsymbol{e}_{i}'\right)^{-1}_ {ii}&= \ev_i'\left(A - x \boldsymbol{e}_{i}\boldsymbol{e}_{i}'\right)^{-1}\ev_i  \\
&=
\ev_i'\left[ A^{-1}
+x A^{-1}\boldsymbol{e}_{i}
\left(1-x \boldsymbol{e}_{i}'A^{-1} \boldsymbol{e}_{i}\right)^{-1}
\boldsymbol{e}_{i}'A^{-1}
\right]\ev_i \\
&=
\ev_i'\left(A^{-1}\right)\ev_i
\left(1 + 
\dfrac{x\,\boldsymbol{e}_{i}' A^{-1} \boldsymbol{e}_{i}}
{1 - x\,\boldsymbol{e}_{i}' A^{-1} \boldsymbol{e}_{i}}
\right) \\
&=
\dfrac{\left(A^{-1}\right)_{ii}}
{1 - x\left(A^{-1}\right)_{ii}}\\
&=\left(\left(A^{-1}\right)_{ii}^{-1}-x\right)^{-1}.
\end{align}
For the equation on the input side, repeat the same calculations with:
\[
\Lambda_i^{in}=\left(  M_{\mathcal{N}^{\text{up}}(i)}-\overline{B}_{i}\boldsymbol{e}^{up}_{i}(\boldsymbol{e}^{up}_{i})'\right)_{in(i),in(i)}^{-1}
\]
\qed

\subsection{Proof of Lemma \ref{lemma_downstream}}

The fact that $\Lambda_i$ is diagonal means that the matrix $M-\hat{B}_i$ is block diagonal:
\[
   M-\hat{B}_i=\begin{pmatrix}
       M^{down}_i-\overline{B}_i\ev_i^{down}(\ev_i^{down})' & O \\
       O & M_i^{up}-\overline{B}_i\ev_i^{up}(\ev_i^{up})'
   \end{pmatrix}
   \]
so, to compute the output price impact we simply have to compute the inverse of the block $M^{down}_i-\overline{B}_i\ev_i^{down}(\ev_i^{down})'$.

The fact that $\Lambda_i$ is diagonal means that there is a partition $(\mathcal{N}^{up}(i),\mathcal{N}^{down}(i))$ of the set of goods $\mathcal{M}$ such that $out(i)\in \mathcal{N}^{down}(i)$ and $\mathcal{N}^{in}(i)\subseteq \mathcal{N}^{up}(i)$, and for any pair of goods $g \in \mathcal{N}^{down}(i)$ and $h \in \mathcal{N}^{up}(i)$, if $g,h\in \mathcal{N}(k)$ for some firm $k$, then $k=i$. The same must be true for $j$, so the partitions have to be nested: $ \mathcal{N}^{down}(j)\subseteq\mathcal{N}^{down}(i)$ and $ \mathcal{N}^{up}(i)\subseteq\mathcal{N}^{up}(j)$. 
So, the submatrix $M_i^{down}$ has the form:
\[
M^{down}_i=\begin{pmatrix}
M_{-j}^{down} & M_{-j,j}^{down} & M_{-j,i}   \\
  (M_{-j,j}^{down})'  & M_{\mathcal{N}^{\text{down}}(j)} & -\overline{B}_{j}\boldsymbol{e}_{j} \\
(M_{-j,i})'  &  -\overline{B}_{j}\boldsymbol{e}_{j}' & \overline{B}_i+ \overline{B}_j+\overline{B}_{-j} 
\end{pmatrix}
\]
where for brevity I call $M_{-j}^{down}:= M_{\mathcal{N}^{\text{down}}(i)\setminus\mathcal{N}^{\text{down}}(j)}$, $M_{-j,i}:=M_{\mathcal{N}^{\text{down}}(i)\setminus\mathcal{N}^{\text{down}}(j), \,out(i)}$ and $M_{-j,j}^{down}:=M_{\mathcal{N}^{\text{down}}(i)\setminus\mathcal{N}^{\text{down}}(j),\, \mathcal{N}^{\text{down}}(j)}$. 
To obtain Equation \eqref{downstream_i}, we do block inversion:
\begin{align*}
&(M_{\mathcal{N}^{\text{down}}(i)}-\hat{B}_i)^{-1}_{out(i),out(i)}=\\
&\left(\begin{pmatrix}
 M_{\mathcal{N}^{\text{down}}(j)} & -\overline{B}_{j}\boldsymbol{e}_{j}^{down} \\
-\overline{B}_{j}(\boldsymbol{e}_{j}^{down})' & \overline{B}_j+\overline{B}_{-j} 
\end{pmatrix}-\begin{pmatrix}
  M_{-j,j}^{down} \\
M_{-j,i}  
\end{pmatrix}(M_{-j}^{down})^{-1} (M_{-j,j}^{down},M_{-j,i} )  \right)^{-1}_{out(i),out(i)}\\
=& \left(\Lambda_{-j}^{-1}+\overline{B}_{j}-\overline{B}_{j}^{2}(\boldsymbol{e}_{j}^{down})'M_{\mathcal{N}^{\text{down}}(j)}^{-1}\boldsymbol{e}_{j}^{down}\right)^{-1} \\
=& \left(\Lambda_{-j}^{-1}+\overline{B}_{j}-\overline{B}_{j}^{2}\left(M_{\mathcal{N}^{\text{down}}(j)}\right)_{out(j), out(j)}^{-1}\right)^{-1},
\end{align*}
where $\Lambda_{-j}^{-1}=\overline{B}_{-j}-(M_{-j,i})'(M_{-j}^{down})^{-1} M_{-j,i} $. This is a diagonal element of the inverse of the submatrix relative to only $i$ and "not $j$" subnetwork:
\[
\sum_{k\in \mathcal{N}^{down}(i)\setminus \mathcal{N}^{down}(j)\,k\neq i} \overline{B}_k[\hat{\vv}_k\hat{\vv}_k']_{\mathcal{N}^{down}(i)\setminus \mathcal{N}^{down}(j)}
\]
Since this is by construction positive definite, it follows that $\Lambda_{-j}^{-1}\ge 0$.

The proof of the identities for the second part is analogous, taking the proper subnetworks into account.
\qed

\subsection{Proof of Lemma \ref{lemma:layer_best_reply}}

  The equation $X=\left(k^{-1}+(\overline{\Lambda}^{-1}+(n-1)X)^{-1}\right)^{-1}$ is equivalent to the quadratic equation:
  \begin{equation}
	(n-1)X^2+(\overline{\Lambda}^{-1}-(n-2)k)X-\overline{\Lambda}_i^{-1}k=0
\end{equation}
\begin{align*}
\Delta&= \left((n-2)k-\overline{\Lambda}^{-1}\right)^2+4 (n-1)\overline{\Lambda}^{-1}k>0
\end{align*}
The quadratic formula gives:
\begin{equation}
X=\frac{\left((n-2)k-\overline{\Lambda}^{-1}\right) \pm\sqrt{\Delta} }{2 (n-1)}
\label{quadratic}
\end{equation}
and since $n\ge 2$, we have that $\Delta>\left((n-2)k-\overline{\Lambda}^{-1}\right)^2$, so this has a positive and a negative root. So, the positive root is unique and the function $\overline{BR}$ is well-defined.

The monotonicity follows from the implicit function theorem applied to the function $F(\overline{BR},n,k):=\overline{BR}^{-1}-\left(k^{-1}+\left(\overline{\Lambda}^{-1}+(n-1)\overline{BR}\right)^{-1}\right)$.
In particular, since:
\begin{align*}
    \dfrac{\partial F}{\partial \overline{BR}}&= -\overline{BR}^{-2} +(n-1)\left(\overline{\Lambda}^{-1}+(n-1)\overline{BR}\right)^{-2}\\
    &=\overline{BR}^{-2}\left(-1+\dfrac{1}{n-1}\dfrac{(n-1)^2\overline{BR}^2}{\left(\overline{\Lambda}^{-1}+(n-1)\overline{BR}\right)^2}\right)<0,
\end{align*}
we have:
\begin{align*}
    \dfrac{\partial \overline{BR}}{\partial k} & =  \dfrac{k^{-2}}{\overline{BR}^{-2} -(n-1)\left(\overline{\Lambda}^{-1}+(n-1)\overline{BR}\right)^{-2}}>0\\
   \dfrac{\partial \overline{BR}}{\partial n }  & =  \dfrac{\left(\overline{\Lambda}^{-1}+(n-1)\overline{BR}\right)^{-2}\overline{BR}}{\overline{BR}^{-2} -(n-1)\left(\overline{\Lambda}^{-1}+(n-1)\overline{BR}\right)^{-2}} >0\\
   \dfrac{\partial \overline{BR}}{\partial \overline{\Lambda} } &= -\dfrac{\left(\overline{\Lambda}^{-1}+(n-1)\overline{BR}\right)^{-2}\overline{\Lambda}^{-2}}{\overline{BR}^{-2} -(n-1)\left(\overline{\Lambda}^{-1}+(n-1)\overline{BR}\right)^{-2}} <0
\end{align*}
from which the monotonicities follow.  \qed 

\subsection{Proof of Lemma \ref{compstat}}

The proof needs the following Lemma.
 
\begin{lemma}

Suppose $k_i=k$ for all layers $i$. Consider a profile $B$. 
If $n_i>n_j$ and $X\le \overline{BR}(\overline{\Lambda}_j(X,B_{-i,j}),n_j,k)$, then $\overline{BR}(\overline{\Lambda}_i(X,B_{-i,j}),n_i,k)> \overline{BR}(\overline{\Lambda}_j(X,B_{-i,j}),n_j,k)$.

\label{compstat-bestrep}
\end{lemma}
\begin{proof}
By definition we have:
\[
\overline{\Lambda}_i(X,B_{-i,j})=\mathcal{L}+\frac{1}{n_jX},\quad \overline{\Lambda}_j(X,B_{-i,j})=\mathcal{L}+\frac{1}{n_iX}
\]
where $\mathcal{L}=\sum_{k\neq i,j}\dfrac{1}{n_kB_k}$. So, $\overline{\Lambda}_i(X,B_{-i,j})\ge \overline{\Lambda}_j(X,B_{-i,j})$ if and only if $n_i\ge n_j$. We want to prove that $\overline{BR}(\overline{\Lambda}_i(X,B^*_{-i,j}),n_i,k)>\overline{BR}(\overline{\Lambda}_j(X,B^*_{-i,j}),n_j,k)$.
Now define $\Delta n=n_i-n_j$, $\overline{\Lambda}(h):=\mathcal{L}+\frac{1}{(n_i-h\Delta n)X}$, and the function $F(h)$ as:
\[
F(h):=\overline{BR}\left(\overline{\Lambda}(h),n_j+h\Delta n,k\right)
\]
For $h=0$, we have $F(0)=\overline{BR}(\overline{\Lambda}_j,n_j,k)$, while for $h=1$ we have $F(1)=\overline{BR}(\overline{\Lambda}_i, n_i,k)$. Now we are  going to prove that $F$ is increasing in $h$ when $n_i>n_j$ and $X\le F(0)$, so proving our thesis. So, assume that $n_i>n_j$, or $\Delta n>0$.
Using the calculations in Lemma \ref{lemma:layer_best_reply}:
\begin{align*}
    \dfrac{\partial F}{\partial h}&=\dfrac{\partial \overline{BR}}{\partial n}\Delta n+\dfrac{\partial \overline{BR}}{\partial \overline{\Lambda}}\dfrac{\Delta n}{\left(n_i-h\Delta n\right)^2X}\\
    &=\dfrac{\Delta n \left(\overline{\Lambda}(h)^{-1}+(n_j+h\Delta n-1)F(h)\right)^{-2}}{\left(F(h)^{-2}-\dfrac{(n_j+h\Delta n-1) }{\left(\overline{\Lambda}^{-1}(h)+(n_j+h\Delta n-1)F(h)\right)^2}\right)}\left(F(h)-\dfrac{\overline{\Lambda}(h)^{-2}}{\left(n_i-h\Delta n\right)^2X}\right)
\end{align*}
Because of Lemma \ref{lemma:layer_best_reply}, the denominator is positive, and so the first fraction is positive. Hence, the derivative is positive if and only if 
\begin{align}
F(h) & >\dfrac{\overline{\Lambda}(h)^{-2}}{\left(n_i-h\Delta n\right)^2X}\\
&=\left(\mathcal{L}+\dfrac{1}{(n_i-h\Delta n)X}\right)^{-2}\dfrac{1}{\left(n_i-h\Delta n\right)^2X}\\
    &=\dfrac{X}{\left((n_i-h\Delta n)X\mathcal{L}+1\right)^{2}}
\end{align} 
Since $\mathcal{L}>0$, this condition is always satisfied if $X \le  F(h)$. By assumption, we know that $X\le F(0)$. By the above reasoning, $F$ is strictly increasing in 0 and, by continuity of the derivative, it is increasing at least in an open interval containing 0.
Now call $h^*>0$ the supremum of the interval where $F$ is increasing. Since $F$ is increasing in 0, the supremum exists. Since $F$ is increasing between $0$ and $h^*$, it must be $F(h^*)> F(0)\ge X$. So, also $h^*$ satisfies the condition, with the consequence that $F$ is increasing in $h^*$. Since $F$ is continuous, there is an interval including $h^*$ where $F$ is increasing, so $h^*$ cannot be the supremum. Then, we proved that, if $X\le  F(0)$,  $F$ is increasing for any $h\ge 0$ and so we can conclude $F(0)<F(1)$, which is the thesis. 

\end{proof}

    Consider the equilibrium profile $B^*$. The best reply function \eqref{bestreply_chain} is increasing in $k_i$ and $n_i$ by Lemma \ref{lemma:layer_best_reply}. So, since the game is supermodular, it follows that in equilibrium each coefficient $\overline{B}_j$ is increasing in each $k_i$ and $n_i$. 

    To compare the equilibrium values of the coefficients, we apply the theory of monotone comparative statics to modified best reply functions. The argument will follow Lemma 1 of \cite{lazzati2013comparison}, using the ranking of best replies provided in Lemma \ref{compstat-bestrep} above.\footnote{The argument is a slight modification of \cite{lazzati2013comparison}, because in our context we cannot assume, as in that paper, that $\overline{BR}_i(X,B_{-i,j})>\overline{BR}_j(X,B_{-i,j})$ for \emph{all} $X$ and $B_{i,j}$ (this is, in general, false). It turns out that the proof works with the weaker assumption that $\overline{BR}_i(X,B^*_{-i,j})>\overline{BR}_j(X,B^*_{-i,j})$ for all $X\le B_i^*,B_j^*$, that is what Lemma \ref{compstat-bestrep} provides.}

Consider the unique equilibrium profile $B^*$. All firms in the same layer are identical, and use an identical coefficient: so, we index the profile directly with the layers. Assume $n_i> n_j$. 
Suppose by contradiction that $B_i^*\le B_j^*$. Now using Lemma \ref{compstat-bestrep}
and letting $X=B_i^*$, we have $\overline{BR}(\overline{\Lambda}_j(X,B^*_{-i,j}),n_j,k)=B_j^*$, and so the assumption $X\le  \overline{BR}(\overline{\Lambda}_j(X,B^*_{-i,j}),n_j,k)$ is satisfied by assumption. So, we can conclude that:
\begin{equation}
\overline{BR}(\overline{\Lambda}_i(B_i^*,B^*_{-i,j}),n_i,k)>\overline{BR}(\overline{\Lambda}_j(B_i^*,B^*_{-i,j}),n_j,k)
\label{contradiction}
\end{equation}
Moreover, since the best reply is increasing we get: 
\[
B_i^*=\overline{BR}(\overline{\Lambda}_i(B_j^*,B^*_{-i,j}),n_i,k)\ge \overline{BR}(\overline{\Lambda}_i(B_i^*,B^*_{-i,j}),n_i,k)
\]
Using again the assumption, we also conclude that:
\[
B_j^*=\overline{BR}(\overline{\Lambda}_j(B_i^*,B^*_{-i,j}),n_j,k)\ge B_i^*
\]
Combining the last two inequalities we find: 
\[
\overline{BR}(\overline{\Lambda}_j(B_i^*,B^*_{-i,j}),n_j,k)\ge \overline{BR}(\overline{\Lambda}_i(B_i^*,B^*_{-i,j}),n_i,k),
\] that contradicts Equation \eqref{contradiction}. Hence, it cannot be that $B_i^*\le B_j^*$ and so we conclude that $B_i^*> B_j^*$. An analogous reasoning proves that the same is true when $n_i=n_j=n^*$ and we vary $k_i$.
\qed

\subsection{Proof of Lemma \ref{lemma:nontrivial}}

To see that the equilibrium slopes are nonzero, consider the profile in which for all $i$ $\overline{B}_i=\varepsilon$. In this case, $\overline{\Lambda}_i^{out}+\overline{\Lambda}_i^{in}=\dfrac{1}{\varepsilon}\sum_{j\neq i} \dfrac{1}{n_j}+\dfrac{1}{B_c}$. So, the best reply is:
\begin{align*}
\overline{B}_i'&=\left(k_i^{-1}+\left(\left(\dfrac{1}{\varepsilon}\sum_{j\neq i} \dfrac{1}{n_j}+\dfrac{1}{B_c}\right)^{-1}+(n_i-1)\varepsilon\right)^{-1}\right)^{-1}\\
&=\varepsilon \left(\varepsilon k_i^{-1}+\left(\left(\sum_{j\neq i} \dfrac{1}{n_j}+\dfrac{\varepsilon}{B_c}\right)^{-1}+(n_i-1)\right)^{-1}\right)^{-1}
\end{align*}
If $N=1$, then $\sum_{j\neq i} \dfrac{1}{n_j}=0$, so $\overline{B}'_i\to (k_i^{-1}+B_c^{-1})^{-1}$.

If $N>1$:
\[
\lim_{\varepsilon\to 0}\dfrac{\overline{B}_i'}{\varepsilon}=\left(\sum_{j\neq i} \dfrac{1}{n_j}\right)^{-1}+(n_i-1)
\]
and if $n_i\ge 2$ then the above is larger than 1, meaning that $\overline{B}_i'>\varepsilon$. Iterating the best reply, we obtain an increasing sequence that must converge to the equilibrium $\overline{B}_i$. So, under our assumptions there always exists an $\underline{\varepsilon}>0$ such that for all $i$ $\overline{B}_i>\underline{\varepsilon}$, and so the equilibrium is nonzero.  \qed

\subsection{Proof of Lemma \ref{lemma_BR_unilateral}}

Using Lemma \ref{lemma_downstream}, specifically Equation \eqref{downstream_j}, we have that for any $i$, $j$, $\Lambda_{i}^{down}>\Lambda_{j}^{down}$ if and only if 
\begin{equation}
\left(M_{\mathcal{N}^{down}(i)}^{-1}\right)_{out(i),out(i)}^{-1}-\overline{B}_{i}<\left(M_{\mathcal{N}^{down}(j)}^{-1}\right)_{out(j),out(j)}^{-1}-\overline{B}_{j}.
\label{eq_downstream}
\end{equation}
Now, by definition, the set $\mathcal{N}^{\text{down}}(i)$ only includes the goods produced using directly or indirectly $g$, 
but not their inputs and their supply chain. Moreover, $j$ is the only customer of $i$, and so in $\mathcal{N}^{\text{down}}(i)$ the only good connected to the output of $i$ is the input of $j$. So, the matrices satisfy:
\begin{align}
M_{\mathcal{N}^{\text{down}}(i)} =
\begin{pmatrix}
(n^*-1)\overline{B}_{i} + n^*\overline{B}_{j} & -n^*\overline{B}_{j}\boldsymbol{e}_{1}' \\
-n^*\overline{B}_{j}\boldsymbol{e}_{1} & M_{\mathcal{N}^{\text{down}}(j)}
\end{pmatrix}.
\end{align}
where $\ev_1$ is the first canonical basis vector. 

\medskip

Now we can compute $\left(M_{\mathcal{N}^{\text{down}}(i)}^{-1}\right)_{out(i),out(i)}$ using block inversion:
\begin{align}
(M^{-1}_{\mathcal{N}^{\text{down}}(i)})_{out(i),out(i)}
&= 
\left( n^{*}\overline{B}_{i} + n^{*}\overline{B}_{j} - n^{*2}\overline{B}_{j}^{2}(M^{-1}_{\mathcal{N}^{\text{down}}(j)})_{out(j),out(j)} \right)^{-1}.
\end{align}

So:
\begin{align}
\left( (M^{-1}_{\mathcal{N}^{\text{down}}(i)})_{out(i),out(i)} \right)^{-1} - \overline{B}_{i}
&= n^{*}\overline{B}_{i} + n^{*}\overline{B}_{j} - (n^{*})^{2}\overline{B}_{j}^{2}(M^{-1}_{\mathcal{N}^{\text{down}}(j)})_{out(j),out(j)} - \overline{B}_{i} \\
&= (n^{*}-1)\overline{B}_{i} + n^{*}\overline{B}_{j} - (n^{*})^{2}\overline{B}_{j}^{2}(M^{-1}_{\mathcal{N}^{\text{down}}(j)})_{out(j),out(j)}.
\label{inverse_down}
\end{align}
Hence, the inequality \eqref{eq_downstream} becomes:
\begin{align}
(n^{*}-1)\overline{B}_{i} + n^{*}\overline{B}_{j} - (n^{*})^{2}\overline{B}_{j}^{2}(M^{-1}_{\mathcal{N}^{\text{down}}(j)})_{out(j),out(j)}
<
\left( (M^{-1}_{\mathcal{N}^{\text{down}}(j)})_{out(j),out(j)} \right)^{-1} - \overline{B}_{j}.
\end{align}
Let $x = (M_{\mathcal{N}^{\text{down}}(j)}^{-1})_{out(j),out(j)}$. Then the above inequality is equal to:
\begin{align}
(n^{*}-1)\overline{B}_{i} + (n^{*}+1)\overline{B}_{j} - (n^{*})^{2}\overline{B}_{j}^{2}x < \frac{1}{x}.
\end{align}

Rearranging,
\begin{align}
0 < 1 - \left( (n^{*}-1)\overline{B}_{i} + (n^{*}+1)\overline{B}_{j} \right)x + \left( n^{*}\overline{B}_{j} \right)^{2}x^{2}.
\end{align}

Now, if we use $\overline{B}_{i} \le  \overline{B}_{j}$, we get:
\begin{align}
&1 - \left( (n^{*}-1)\overline{B}_{i} + (n^{*}+1)\overline{B}_{j} \right)x + \left( n^{*}\overline{B}_{j} \right)^{2}x^{2}
\ge \\
&1 - 2 n^{*}\overline{B}_{j} x + \left( n^{*}\overline{B}_{j} \right)^{2}x^{2}=
\left( n^{*}\overline{B}_{j}x - 1 \right)^{2} > 0.
\end{align}
which proves that inequality \eqref{eq_downstream}
 is satisfied. It is satisfied with a strict inequality, because $x\overline{B}_j=1$ by Equation \eqref{inverse_down} would imply $\Lambda^{down}_i=((n^*-1)\overline{B}_i)^{-1}$, that is impossible because it could be only if the price of $out(i)$ where constant, that is not the case in the unilateral model.\qed

\subsection{Proof of Proposition \ref{prop:markup_generalized}}

Part 1 follows directly from part 2 of Theorem \ref{thm:main_comparative}.   For part 2, the price impact is the analogous to Equation \eqref{chain_output}, but on inputs:
\begin{align}
	\Lambda^{unilateral}_i=\begin{pmatrix}
		0  & 0 \\
		0 & \left( (\overline{\Lambda}_i^{in})^{-1}+(n_i-1)\overline{B}_i\right)^{-1}
	\end{pmatrix}
	\label{chain_input}
\end{align}
Markups are zero, while the markdowns are:
\[
\mu_i^{in}=\dfrac{Q}{n}
\dfrac{\overline{\Lambda}_{i}^{in}}{1+\overline{\Lambda}_{i}^{in}\overline{B}_{i}(n-1)}=\dfrac{Q}{n}\left(\overline{B}_i^{-1}-k^{-1}\right)
\]
Plugging the price impact \eqref{chain_input} in the best reply equation we get:
\begin{equation}
\overline{B}^{unilateral}_i=\left(k^{-1}+\dfrac{1}{(\overline{\Lambda}_{i}^{in})^{-1}+\overline{B}_{i}(n-1)}\right)^{-1}=\overline{BR}(\overline{\Lambda}_i^{in},k,n)
\label{bestreply_chain_unilateral}
 \end{equation}
 where $\overline{BR}$ is the same function defined in Lemma \ref{lemma:layer_best_reply}.
By Lemma \ref{lemma:layer_best_reply} and the fact that $\overline{\Lambda}_i^{in}$ is decreasing upstream we conclude that, in equilibrium, $\overline{B}_i$ is increasing upstream. As a consequence, both markdowns and profits are increasing downstream.\qed

\subsection{Proof of Proposition \ref{sequential_cournot}}

Consider the setting of the layered supply chain of Example \ref{ex:line}. Solve it as a standard Sequential Oligopoly as in \cite{belleflamme2015industrial}: firms in layer 1 play a Cournot game on outputs, taking the input price $p_2$ as given. Then, compute the demand for good 2 from layer 1 implied by the Cournot equilibrium between downstream firms. Finally, firms in layer 2 play a Cournot game in outputs. If there are more layers, Firms take as given the input price $p_{3}$, and so on for all the layers.

Consider first the downstream layer, 1. The inverse demand is: $P_1(\sum_j q_{j,1})=\dfrac{1}{B_c}(A-\sum_jq_{j,1})$ The FOC is:
\[
\dfrac{\partial}{\partial q_i}\left((P_1-P_2)q_{i,1} -\dfrac{1}{2k_i}q_{i,1}^2 \right)=P_1-P_2-\dfrac{1}{B_c}q_{i,1}-\dfrac{1}{k_1}q_{i,1}=0
\]
So, the quantity and the price have the relation:
\[
q_i=\left(\dfrac{1}{B_c}+\dfrac{1}{k_i}\right)^{-1}(P_1-P_2)
\]
This is a linear schedule with slope $\left(\dfrac{1}{B_c}+\dfrac{1}{k_i}\right)^{-1}$, which is exactly Equation \eqref{bestreply}, where the price impact is modified so that $\overline{\Lambda}^{in}=0$ and the slope of competitors are set to 0, so that the price impact is:
\[
\Lambda_i^{sequential}=\begin{pmatrix}
    \overline{\Lambda}_i^{out} & 0\\
    0 & 0
\end{pmatrix}
\]

Finally, the demand from layer 1 to the upstream layer can be obtained by solving the market-clearing conditions:
\begin{align}
 P_1\left(\sum_j q_{j,1}\right)&=\dfrac{1}{B_c}\left(A-\sum_jq_{j,1}\right)\\
 q_i&=\left(\dfrac{1}{B_c}+\dfrac{1}{k_1}\right)^{-1}(P_1-P_2)
\end{align}
which, since they are linear, have the same solution as the market clearing in the competition in schedules, and it is linear:
\[
P_2=A_2-\overline{B}_2\sum_j q_{j,2}
\]

Since the demand has the same form as the first layer, just with different coefficients, the reasoning can be repeated for all layers.
\qed

\subsection{Proof of Lemma \ref{lemma:limit}}
\label{proof_lemma:limit}
Consider first the multilateral case. By Proposition \ref{prop:markup}, we know that in equilibrium all slope coefficients are the same: $\overline{B}_i=B^*$, from which it follows that $\overline{\Lambda}_i=(N-1)/(n^*B^*)+B_c^{-1}$. Since $B^*$ is bounded, in the limit for $N\to \infty$ we have that $\overline{\Lambda}_i^{-1}\to 0$. The best reply equation \eqref{bestreply_chain} becomes:
\[
B^*=(k^{-1}+((n^*-1)B^*)^{-1})^{-1}
\]
and collecting the $B^*$ terms, this can be solved analytically, yielding $B^*\to k\frac{n^*-2}{n^*-1}$. If $n^*=2$, then $B^*\to 0$. In this case, the above best reply equation contains a division by zero, but it is still satisfied in the limit, because $(B^*)^{-1}\to \infty$ and so the right-hand side goes to zero.

In the local case, instead, we show that in equilibrium there is a $\underline{B}$ that is a lower bound to all slope coefficients. The best reply Equation remains of the form \eqref{bestreply_chain}
, with the difference that the coefficients $\overline{\Lambda}_i$ satisfy:
\[
\overline{\Lambda}_i=\dfrac{1}{2B_{i-1}}+\dfrac{1}{2B_{i+1}}
\]
with the convention that $B_0=B_c/2$, and $B_{N+1}=\infty$. Since the first and last do not have the same functional form, the equilibrium is not homogeneous. However, for $\varepsilon$ small enough, we can consider $\varepsilon<\min_i B_i$, and so for all $i$, $\overline{\Lambda}_i>(2\varepsilon)^{-1}$. Moreover, there is an $\varepsilon$ satisfying the best reply equation:
\[
\varepsilon=\left(\dfrac{1}{k}+\dfrac{1}{2 \varepsilon+\varepsilon}\right)^{-1}
\]
which is a lower bound on the equilibrium profile: $\overline{B}^{\text{local}}>\varepsilon$. Solving the above equation we get that the lower bound is $\overline{B}:=k\dfrac{2}{3}$, which is strictly positive and independent of $N$. \qed

\subsection{Proof of Proposition \ref{prop_vertical}}

\label{proof_prop_vertical}

After the merger, by assumption the merged firm forecloses the other downstream rivals, so only one firm survives, and is a monpolist. Since the upstream firm has constant marginal cost ($k_{2}\to\infty$) and the downstream firm has quadratic marginal cost with curvature $k$, the monopolist now has a technology with marginal cost with curvature k.

The monopoly price in the after-merger setting is:
\[
p^{post}=A\left(B_c+B_M \right)^{-1}
\]
where $B_M$ is the equilibrium coefficient of the supply of the only firm.

In the pre-merger equilibrium instead the final price is:
\[
p^{pre}=\frac{A}{B_c+\left(\dfrac{1}{n_1\overline{B}_1}+\dfrac{1}{\overline{B}_2}\right)^{-1}}
\]
These expressions are valid under any Generalized SDFE, because they are agnostic on the slopes.

Now, in the post-merger case, since there is only one firm, the price impact is $\Lambda_M=1/B_c$. The firm has no inputs, nor competitors so, multilateral, unilateral and sequential Cournot all give the same “best reply”, which is equivalent to the monopolist optimization: $B_{M}=(k^{-1}+\Lambda_{M})^{-1}=(k^{-1}+B_{c}^{-1})^{-1}$. So, the price is:
\[
p^{post}=A\left(B_c+B_M \right)^{-1}=A\left(B_c+\frac{1}{1/B_c+k^{-1}} \right)^{-1}
\]
In the pre-merger case, since firm 2 has no suppliers, the price impact on inputs, $\overline{\Lambda}_2^{in}$, is zero both with unilateral and multilateral market power. Moreover, since it is the only firm upstream, the “slopes of competitors” do not matter. The price impact on outputs is $\overline{\Lambda}_2^{out}=B_c^{-1}+(n_1\overline{B}_1)^{-1}$. So, using \eqref{bestreply_chain}, the best reply equation for $\overline{B}_{2}$ is the same with multilateral, unilateral market power, or Sequential Cournot. (remember that $k_2\to \infty$):
\begin{equation}
\overline{B}_2=\left(B_c^{-1}+(n_1\overline{B}_1)^{-1}\right)^{-1}
\label{BR_2}
\end{equation}
This implies that the price is:
\[
p^{pre}=A\left(B_c+\frac{1}{1/B_c+\frac{2}{n_1\overline{B}_1}} \right)^{-1}
\]
Comparing the two equations, we find that the price is higher after the merger if and only if $2k <n_1 \overline{B}_1$.

Now in all cases, $n_1\overline{B}_1(n_1)$ is increasing and unbounded, so the equation $2k=n_1\overline{B}_1(n_1)$ has a unique solution. Moreover, if $n_1=2$ the RHS is lower, while for $n_1$ sufficiently large the RHS is higher. Define $n^{\ast}$ as the solution in case of multilateral market power, and $n_{\ast}$ as the solution with unilateral market power.
 By Theorem \ref{thm:main_comparative}, for any $n_1$, $\overline{B}_1^{\text{unilateral}}(n_1)>\overline{B}_1^{\text{multilateral}}(n_1)$, so that $n_{\ast}<n^{\ast}$.  Hence, it follows that there is a range of $n_1 \in (n_{\ast},n^{\ast})$ the merger is welfare-increasing with multilateral market power, but welfare-decreasing with unilateral market power.
 
 Moreover, the same is true for Sequential Cournot if $B_c$ is sufficiently large. To see it, first note that by Equation \eqref{BR_2}, in all considered equilibria $\overline{B}_2\to n_1\overline{B}_1$ if $B_c\to \infty$. 
  The equation for $\overline{B}_1$ in the Sequential Cournot is:
 \[
 \overline{B}_1^{\text{Cournot}}=\left(\dfrac{1}{k}+\dfrac{1}{B_c}\right)^{-1}\to k \text{ for } B_c\to \infty 
 \]
 In the multilateral case, instead, the best reply Equation \eqref{bestreply_chain} converges to:
 \[
\overline{B}_{1}^{\text{multilteral}}=\left(\dfrac{1}{k}+\dfrac{1}{\overline{B}_{2}+(n_{1}-1)\overline{B}_{1}^{\text{multilteral}}}\right)^{-1}
 \]
 Now using the fact that, in the limit, $\overline{B}_2^{\text{multilateral}}= n_1\overline{B}_1^{\text{multilateral}}$, we find that the above equation is solved by $\overline{B}_1^{\text{multilateral}}=\dfrac{2n_1-2}{2n_2-1}k<k$. It follows that, if $B_c$ is sufficiently large, $\overline{B}_1^{\text{Cournot}}(n_1)>\overline{B}_1^{\text{multilateral}}(n_1)$.
 \qed

\section{General model: non-linear schedules and uncertainty}
\label{sec:full_blown}

In this section I show that the linear equilibrium analyzed in the main text remains an equilibrium even if firms are free to choose arbitrary (possibly non-linear) schedules. Moreover, it remains an equilibrium also introducing uncertainty in the intercept of the consumer demand $\boldsymbol{A}$ and the labor cost parameters $\boldsymbol{f}_L$. In particular, the linear schedule becomes the unique best reply, and the equilibrium is ex-post, extending the \cite{klemperer1989supply} selection argument to the input-output network setting.

\subsection{The general game}

The technology is the one defined by the Equations \eqref{tech_constraint_min}. 
The key difference is that the vectors $\boldsymbol{A}$ and $\boldsymbol{f}_L$ are stochastic. Formally, assume that the vector $\eps:=(\boldsymbol{A},\boldsymbol{f}_L)$ has a joint distribution $F_{\eps}$ with finite mean, and that $\eps$ has support $\mathcal{E}$, which we leave generic for now. The support of $f_{i,L}$ is simply the projection of $\mathcal{E}$ and is denoted $\mathcal{E}_i$. For the uniqueness result in Lemma \ref{expost_general}, we need $\mathcal{E}$ such that the full price space is covered. The easiest way to do so is assuming $\mathcal{E}=\mathbb{R}^{n+m}$, and maintain the assumption of the main text that negative prices and quantities are allowed and simply represent trade flowing in the opposite direction. 
The schedules chosen by the firms $\mathcal{S}_i$ now are maps from $\mathbb{R}^{d_i}\times \mathcal{E}_i\to \mathbb{R}^{d_i}$, mapping $\pv_i,f_{i,L}$ to $\mathcal{S}_i(\pv_i,f_{i,L})$, and the schedules are chosen \emph{before} the realization of the vectors $\boldsymbol{A}$ and $\boldsymbol{f}_L$. Firms can condition their schedule on their value $f_{i,L}$: this can be equivalently interpreted by saying that firms are able to observe their own realization of the cost $f_{i,L}$. In this context, the market-clearing conditions are analogous to \eqref{sub_mktclear}, remembering that schedules are also functions of the realization of $f_{i,L}$, and the consumer demand a function of $\boldsymbol{A}$:
\begin{align}
	\sum_{i\,:\,g \in \mathcal{N}_i}\mathcal{S}_{ig}(\pv_i,f_{i,L})&=D_{c,g}(\pv_{c},\boldsymbol{A})\quad \forall g \in \mathcal{C} \quad \sum_{i\,:\,g \in \mathcal{N}_i}\mathcal{S}_{ig}(\pv_i,f_{i,L})=0 \quad \forall g \in \mathcal{M}\setminus \mathcal{C}
    \label{mktclear_general}
\end{align}

The set of feasible schedules for firm $i$ is denoted $\mathcal{A}_i$, and $\mathcal{A}=\prod_{i \in \mathcal{N}}\mathcal{A}_i$. The set $\mathcal{A}_i$ is the set of schedules that:
\begin{enumerate}
    
    \item satisfy Equations \eqref{tech_constraint_min};

    \item are such that the market-clearing conditions \eqref{mktclear_general} uniquely define a pricing function.\footnote{Here we do not state explicitly sufficient conditions for this to be the case: in the parallel paper \cite{bizzarri2025general}, the reader can find a set of sufficient conditions.}
    
\end{enumerate}

Formally, we have the following definition.
\begin{defi}[Pricing function and payoffs]
	

A function $\pv:\mathcal{E}\times \mathcal{A}\to \RR^m$ that solves \eqref{mktclear_general} is called a pricing function. 

	
The payoff of firm (player) $i$ is the mapping from supply and demand schedules in $\mathcal{A}_i$ to real numbers defined by the expected profits:
	\begin{align*}
		\pi_{i}(\mathcal{S}_{i}, \mathcal{S}_{-i})&=\EE_F\big(\pv_i'(\eps)\mathcal{S}_i(\pv_i(\eps),f_{i,L})-\mathcal{S}_{\ell,i}(\pv_i(\eps),f_{i,L})\big)
	\end{align*} 
	\label{pricing}	
	
\end{defi}

\subsection{Equilibrium}

\paragraph{Linear equilibrium}

We look for an equilibrium in which the schedules are linear, namely they have the expression of Equation \eqref{linear_schedule}:
\begin{equation}
\mathcal{S}_i(\pv_i,f_{i,L})=B_i\pv_i-B_{i,f}f_{i,L}
\label{schedules}
\end{equation}
with the difference that $f_{i,L}$ is factored out of the intercept vector, because it is interpreted as a stochastic variable realized after the schedule is chosen.

\paragraph{Market-clearing and residual demand}

The best reply problem of firm $i$ is:
\begin{equation}
	\label{BestReply_general}
	\max_{\mathcal{S}_i\in \mathcal{A}_i} 	\pi_{i}(\mathcal{S}_{i},\mathcal{S}_{-i})
\end{equation}	
If the schedules of other players $\mathcal{S}_{-i}$ are linear as in Equation \eqref{schedules},
 then the market-clearing equations and the residual schedule follow, respectively Equation \eqref{mktclear_linear}, \eqref{residual_schedule} of the main text, with the only difference that the intercept $\boldsymbol{\tilde{A}}_i$ is stochastic. 

The key result is the following, which generalizes Lemma \ref{linear-residual}. It shows that the linear schedules represent the unique best reply among all the schedules satisfying the assumptions for which the game is well defined. The reason is the same as in \cite{klemperer1989supply}: with support large enough, uncertainty forces the firms to choose a specific locus of prices and quantities, and all the choices are relevant for some realization of the stochastic parameters.

\begin{lemma}
The schedule $\mathcal{S}^*_i$ solves the best reply problem \eqref{BestReply_general}
 if and only if for any $\eps$ the quantity vector $\qv_i^*=\mathcal{S}^*_i(\pv_i(\eps),f_{i,L})$ solves:
\begin{equation}
    \max_{\qv_i, \ell_i}\boldsymbol{q}_{i}'\boldsymbol{p}^r_{i}(\qv_i, \eps;B_{-i})-\ell_i
    \label{expost_general}
\end{equation}
subject to the technology constraint \eqref{tech_constraint_min}. Moreover, it is the unique such schedule if the support $\mathcal{E}$ is such that for all $i$, $[M^{-1}\overline{\boldsymbol{A}}]_i$ spans $\mathbb{R}^{d_i}$.

\label{equivalent}
\end{lemma}

The condition on the support is satisfied if, for example $\mathcal{E}=\mathbb{R}^{n+m}$, corresponding to the case in which all goods are directly sold to the consumer. However, this is not necessary. For example, for the Supply chain with layers of Example \ref{ex:line}, the condition is easier to satisfy: it is sufficient that the consumer intercept $A_0$ and the cost parameter for the last layer $f_N$ have full support.

Now the solution of \eqref{expost_general} is immediate, as in the main text, and yields:
\begin{equation}
\mathcal{S}_i(\pv_i,f_{i,L})=\overline{B}_i(\vv_i'\pv_i-f_{i,L})\vv_i
\label{eq:general_schedule}
\end{equation}
where the coefficient $\overline{B}_i$ satisfies \eqref{bestreply}.

So, we conclude that the linear equilibrium of the generalized model of this section is unique and identical to the equilibrium of the main text, characterized by Equations \eqref{bestreply} for all $i$. Moreover, since it solves \eqref{expost_general}, it is an ex-post equilibrium: firms would not want to revise their schedule even if they learned the realization of $\eps$.

\subsection{Proof of Lemma \ref{equivalent}}
The proof uses the following Lemma, stating in words that optimizing over a schedule ex-ante is equivalent to pointwise optimizing ex-post, provided the stochastic parameters span the whole domain of the schedule: this is the same property used in \cite{klemperer1989supply} and \cite{malamud2017decentralized}.
\begin{lemma}
	
	Suppose $F$ is a distribution on $\mathcal{B}(\mathcal{E})$. Consider the set of feasible schedules $\mathcal{F}$, a subset of measurable functions $f: \mathcal{E}\to  \RR^n$. Suppose $g:\mathbb{R}^n\times \mathcal{E}\to \mathcal{E}$ is another measurable function on the standard product space. Define two maximizations:
	\begin{equation}
\max_{f \in \mathcal{F}} \EE_{\varepsilon} g(f(\varepsilon),\varepsilon)
\label{global}
\end{equation}
\begin{equation}
\max_{y \in \mathcal{F}(\mathcal{E})}g(y,\varepsilon) 
\label{pointwise}
\end{equation}
Say that a function $f^{**}$ \enquote{maximizes} \eqref{pointwise} if $f^{**}(\varepsilon)$ maximizes \eqref{pointwise} for any $\varepsilon$. Assume that if $f^{**}$ maximizes \eqref{pointwise}, then $f^{**}\in \mathcal{F}$.

\begin{enumerate}

\item If $f^*$ maximizes \eqref{global} and the support of $F$ includes all $\mathcal{E}$, then it also maximizes \eqref{pointwise}. In particular, if the solution of \eqref{pointwise} is unique for any $\varepsilon$, then also the solution of \eqref{global} is unique.

\item If $f^{**}$ maximizes \eqref{pointwise}, then it also maximizes \eqref{global}.

\end{enumerate}


In words: optimizing over a function of $\varepsilon$ is equivalent to optimizing pointwise and ex-post for each realization of $\varepsilon$. 

	\label{lemma_pointwise}
	
\end{lemma}
\begin{proof}
	
\textbf{Part 1}

Suppose by contradiction that $f^*\in \arg\max_{f \in \mathcal{F}} \EE_{\varepsilon} g(f(\varepsilon),\varepsilon)$, but it does not solve \eqref{pointwise}. This means that if $f^{**}$ is a solution of \eqref{pointwise}, there is some set $A$ of positive measure such that $g(f^{**}(\varepsilon'),\varepsilon')>g(f^*(\varepsilon'),\varepsilon')$ in $A$. Hence, since the distribution has full support:
	\begin{align}
	\int g(f^{**}(\varepsilon),\varepsilon) \dd F(\varepsilon)&=\int_A g(f^{**}(\varepsilon),\varepsilon) \dd F(\varepsilon)+\int_{A^c} g(f^{**}(\varepsilon),\varepsilon) \dd F(\varepsilon)\\
&	=\int_A g(f^{**}(\varepsilon),\varepsilon) \dd F(\varepsilon)+\int_{A^c} g(f^*(\varepsilon),\varepsilon)\dd F(\varepsilon)\\
&> \int g(f^*(\varepsilon),\varepsilon)\dd F(\varepsilon)
\end{align}
	which is a contradiction, because $f^{*}$ is the maximizer of \eqref{global}. So, for all $\varepsilon$ except a null set, it must be $g(f^*(\varepsilon),\varepsilon)=g(f^{**}(\varepsilon),\varepsilon)$, and so $f^{**}$ is also a maximizer of \eqref{pointwise}.
	
	\textbf{Part 2}
By definition of $f^{**}$: $g(f^{**}(\varepsilon),\varepsilon)\ge g(y,\varepsilon)$ for all $y\in \mathcal{F}(\mathcal{E})$, $\varepsilon$. Hence, it follows that for any function $f\in \mathcal{F}$:
\[
\int g(f^{**}(\varepsilon),\varepsilon) \dd F(\varepsilon)\ge \int g(f(\varepsilon),\varepsilon)\dd F(\varepsilon)
\]
and so $f^{**}$ solves \eqref{global}.

\end{proof}

First, we observe that by the proof in the main text, if $\mathcal{S}_i(\pv_i,f_{i,L})=\qv_i$, then: 
\[
\boldsymbol{S}_{i}'\boldsymbol{p}_{i}(\mathcal{S}_i, \eps;\overline{B}_{-i})-\ell_i=\boldsymbol{q}_{i}'\boldsymbol{p}_{i}^r(\qv_i, \eps; \overline{B}_{-i})-\ell_i.
\]
So, the objective function in  \eqref{BestReply_general} can be substituted by the one above. Now, choosing a schedule $\mathcal{S}_i$ while being aware of the market-clearing conditions is equivalent to choosing a schedule of price-quantity pairs $\qv_i(\eps),\ell_i(\eps),\pv_i(\eps)$ as a function of $\eps$, under the constraint that $\pv_i=\pv_i^r(\qv_i,\eps)$. Indeed, since the prices are fully pinned down by the inverse demand, it is sufficient to choose a schedule of quantities.
 Call this schedule $S_i(\eps):=\qv_i(\eps),\ell_i(\eps)$. Not all schedules are feasible, because of the technology constraint, and because $B_i=\overline{B}_i\vv_i\vv_i'$. Call the set of feasible quantity schedules for $i$ $\mathcal{F}_i$.
Using the Lemma \ref{lemma_pointwise}, we obtain that the optimization \eqref{BestReply_general} is equivalent to the pointwise ex-post optimization:
\[
  \max_{\qv_i, \ell_i \in \mathcal{F}_i(\mathcal{E})}\boldsymbol{q}_{i}'\boldsymbol{p}^r_{i}(\qv_i, \eps;B_{-i})-\ell_i
\]
under the relevant constraints, and denoting $\mathcal{F}_i(\mathcal{E})$ the codomain of the schedule of firm $i$:
\[
\mathcal{F}_i(\mathcal{E})=\{\qv_i, \ell_i \mid \exists \eps \text{ s.t. } \mathcal{S}_i(\pv_i(\eps),f_{i,L})=\qv_i,\ell_i=\mathcal{S}_{i,\ell}(\pv_i(\eps),f_{i,L}), \text{and } \eqref{tech_constraint_min}\}
\]

Now, define as $\mathcal{T}_i$ the set of quantities satisfying the technology constraints \eqref{tech_constraint_min}. The calculation in the main text is the solution of the pointwise optimization:
\[
\max_{\qv_i, \ell_i \in \mathcal{T}_i}\boldsymbol{q}_{i}'\boldsymbol{p}^r_{i}(\qv_i, \eps; \overline{B}_{-i})-\ell_i
\]
where we removed the constraint on the quantities. This problem has a unique solution, which is the schedule 
\eqref{eq:general_schedule}. This schedule satisfies the functional form restriction, so the optimal schedule $S^{**}_i(\eps)$ belongs to the feasible set $\mathcal{F}_i$. To check that the two optimizations are equivalent, it remains to check that $\mathcal{F}_i(\mathcal{E})$ includes $\mathcal{T}_i$. For this to be the case, because of the technology constraints, it is sufficient that $q_i^{out}(\eps)$ spans $\mathbb{R}$. We have:
\[
q_i^{out}(\eps)=\overline{B}_i(\hat{\vv}_i'\pv(\eps)-f_{i,L})
\]
A sufficient condition for this to span $\mathbb{R}$ is that $\pv$ spans the price space. But this is true if $[M^{-1}\overline{\boldsymbol{A}}]_{\mathcal{N}(i)}$ spans $\mathbb{R}^{d_i}$, which is what we wanted to show.
\qed 

\section{Markups and markdowns}

We show that the markups defined in Equation \eqref{eq:markup} agree with the standard approach of computing the gap between price and marginal cost (markup), or price and marginal revenue products (markdowns). 

\begin{defi}
     The \emph{total cost} of firm $i$ is: $C_i(q_i^{out})=-\sum_{g\in \mathcal{N}(i)^{in}} p^r_g(\qv_i)q_{ig}+\ell_i$, where $\qv_i$ and $\ell_i$ have to be expressed as a function of $q_i^{out}$ using the technology constraints. Define the (absolute) \emph{markup} as $\mu_i:=p_i-\dfrac{\partial C_i}{\partial q_i^{out}}$. 
     
     The \emph{revenue product} of input $g$ is: $R_{ig}(q_{ig})=\sum_{h\in \mathcal{N}(i)\setminus\{ g\}}p_h^rq_{ih}-\ell_i$ where, again, $\qv_i$ and $\ell_i$ must be expressed as a function of $q_{ig}$ using the technology constraints. The \emph{markdown on input $g$} is: $\mu_{ig}:=\dfrac{\partial R_{ig}}{\partial (-q_{ig})}-p_g$.\footnote{The derivative is with respect to $-q_{ig}$ because with our convention input quantities are negative: $q_{ig}<0$}
\end{defi}

\begin{lemma}
    The vector $\boldsymbol{\mu}_i=(\mu_i,-\mu_{ig})$ satisfies: 
\begin{equation}
     \boldsymbol{\mu}_i=q_i^{out}\Lambda_i\vv_i=\Lambda_i\qv_i,
\end{equation} 

    \label{lemma:markup}
\end{lemma}

\subsection{Proof of Lemma \ref{lemma:markup}}

The total cost to produce $q_{i}^{out}$ is:
\begin{align}
	C_i(q_i^{out}):&=-\sum_{g\in \mathcal{N}(i)^{in}} p^r_g(\qv_i)q_{ig}+\ell_i\\
	&=\sum_{g\in \mathcal{N}(i)^{in}} p_g^r\left(q_i^{out}\vv_i\right)q_{i}^{out}f_{ig}+f_{i,L}q_i^{out}+\dfrac{1}{2k_i}(q_i^{out})^2
\end{align}
where everything is expressed as a function of $q_i^{out}$ using the technology constraints, including the prices where we write explicitly the argument of the residual inverse schedule: $\pv_i^r\left(q_i^{out}\vv_i\right)$. In particular, $\partial \pv_i^r\left(q_i^{out}\vv_i\right)/\partial q_i^{out}=-\Lambda_i \times \vv_i$.
Remember that if $g$ is an input of $i$ then $q_{ig}<0$ by convention, so the input quantity is $-q_{ig}$. Analogously, the (net) revenue product generated by input $g$ is:
\begin{align*}
R_{ig}(q_{ig})	&:=\sum_{h\in \mathcal{N}(i)\setminus\{ g\}}p_h^rq_{ih}-\ell_i=\boldsymbol{q}_{i}'\boldsymbol{p}^r_{i}-p^r_{g}q_{ig}-f_{i,L}q_{i}^{out}-\dfrac{1}{2k_i}\left(q_{i}^{out}\right)^{2}\\
&	=q_{i}^{out}\boldsymbol{v}_{i}'\boldsymbol{p}^r_{i}-p^r_{g}q_{ig}-f_{i,L}q_{i}^{out}-\dfrac{1}{2k_i}\left(q_{i}^{out}\right)^{2}\\
&	=\dfrac{(-q_{ig})}{f_{ig}}\boldsymbol{v}_{i}'\boldsymbol{p}^r_{i}((-q_{i,g})f_{i,g}^{-1}\vv_i)+p^r_{g}((-q_{i,g})f_{i,g}^{-1}\vv_i)(-q_{ig})-f_{i,L}\dfrac{(-q_{ig})}{f_{ig}}-\dfrac{1}{2k_i}\left(\dfrac{q_{ig}}{f_{ig}}\right)^{2},
\end{align*}
where again everything is expressed as a function of $q_{ig}$ using the technology constraint $q_i^{out}=f_{ig}^{-1}(-q_{ig})$.
Now, compute the marginal cost and the marginal revenue product:
\begin{align*}
    \dfrac{\partial C_i}{\partial q_{i}^{out}}&=\sum_{g}p^r_{g}f_{ig}-\sum_{g}\left[\Lambda_{i}\boldsymbol{v}_{i}\right]_{g}f_{ig}q_{i}^{out}+f_{i,L}+\dfrac{1}{k_i}q_{i}^{out}\\
    &=-\boldsymbol{v}_{i}'\boldsymbol{p}^r_{i}+p_i^{out}+\boldsymbol{v}_{i}'\Lambda_{i}\boldsymbol{v}_{i}q_{i}^{out}-\left[\Lambda_{i}\boldsymbol{v}_{i}\right]_{i}q_{i}^{out}+f_{i,L}+\dfrac{1}{k_i}q_{i}^{out}\\
   &=-\boldsymbol{v}_{i}'\boldsymbol{p}^r_{i}+f_{i,L}+p_i^{out}+\left(\boldsymbol{v}_{i}'\Lambda_{i}\boldsymbol{v}_{i}+\dfrac{1}{k_i}\right)q_{i}^{out}-\left[\Lambda_{i}\boldsymbol{v}_{i}\right]_{i}q_{i}^{out}\\
   &=-\left(\boldsymbol{v}_{i}'\boldsymbol{p}^r_{i}-f_{i,L}-\overline{B}_i^{-1}q_{i}^{out}\right)+p_i^{out}-\left[\Lambda_{i}\boldsymbol{v}_{i}\right]_{i}q_{i}^{out}=p_i^{out}-\left[\Lambda_{i}\boldsymbol{v}_{i}\right]_{i}q_{i}^{out}\\
    \dfrac{\partial R_{ig}}{\partial (-q_{ig})}&	=\dfrac{1}{f_{ig}}\boldsymbol{v}_{i}'\boldsymbol{p}_{i}-\dfrac{(-q_{ig})}{f^2_{ig}}\boldsymbol{v}_{i}'\Lambda_{i}\boldsymbol{v}_{i}-\dfrac{f_{i,L}}{f_{ig}}-\dfrac{(-q_{ig})}{k_if_{ig}^{2}}-\dfrac{1}{f_{ig}}\left[\Lambda_{i}\boldsymbol{v}_{i}\right]_{g}(-q_{ig})+p_{g}^r \\
    &=\dfrac{1}{f_{ig}}\left(\boldsymbol{v}_{i}'\boldsymbol{p}^r_{i}-f_{i,L}-\dfrac{(-q_{ig})}{f_{ig}}\left(\boldsymbol{v}_{i}'\Lambda_{i}\boldsymbol{v}_{i}+\dfrac{1}{k_i}\right)\right)-\dfrac{1}{f_{ig}}\left[\Lambda_{i}\boldsymbol{v}_{i}\right]_{g}(-q_{ig})+p_{g}^r\\
    &=\dfrac{1}{f_{ig}}\left(\boldsymbol{v}_{i}'\boldsymbol{p}^r_{i}-f_{i,L}-q_{i}^{out}\overline{B}_i^{-1}\right)-\dfrac{1}{f_{ig}}\left[\Lambda_{i}\boldsymbol{v}_{i}\right]_{g}(-q_{ig})+p_{g}^r	=-\dfrac{1}{f_{ig}}\left[\Lambda_{i}\boldsymbol{v}_{i}\right]_{g}q_{ig}+p_{g}^r,
\end{align*}
where, in both calculations, the terms with $\overline{B}_i$ disappear because of the best reply equation.
So, the markup and markdowns are:
\begin{align*}
\mu_i^{out}&=p_{i}^{out}-\dfrac{\partial C_i}{\partial q_{i}^{out}}	\\
&=p_{i}^{out}-\left(p_i^{out}-\left[\Lambda_{i}\boldsymbol{v}_{i}\right]_{i}q_{i}^{out}\right)\\	&=\left[\Lambda_{i}\boldsymbol{q}_{i}\right]_{i}\\
  \mu_{i,g}^{in}&=   \dfrac{\partial R_{ig}}{\partial (-q_{ig})}-p_g\\
  &=  \left(-\dfrac{(-q_{ig})}{f_{ig}}\left[\Lambda_{i}\boldsymbol{v}_{i}\right]_{g}+p_{g}\right)-p_{g}\\
    &=-\left[\Lambda_{i}\boldsymbol{v}_{i}\right]_{g}q_{i}^{out}=-\left[\Lambda_{i}\boldsymbol{q}_{i}\right]_{g}
\end{align*}
So, the (signed) markup-markdown vector is: $(\mu_i^{out},-\boldsymbol{\mu}_i^{in})=\Lambda_{i}\boldsymbol{q}_{i}$. \qed

\section{Substitute intermediate inputs}
\label{sec:substitutes}
The analysis in the main text assumed perfect complementarity among intermediate inputs, mainly to simplify notation and ensure linear-quadratic profit functions. This appendix shows that many results of the paper generalize to the case of intermediate inputs that are imperfect complements or substitutes.
To keep tractability, we need the equilibrium to be in linear schedules. This means that we must preserve the linear-quadratic nature of the objective function in \eqref{BestReply}. Here,
I introduce a parametric functional form for the technology that allows inputs to be imperfect complements or substitutes, while keeping a linear-quadratic expression for the profits. The technology of the main text is a limiting case in which inputs converge to be perfect complements, and the linear labor terms $f_{i,L}$ are set to zero.

To define the production function, define a \emph{labor allocation} as a subdivision of labor $\boldsymbol{\ell}_i=(\ell_{ij})_{j\in \mathcal{N}^{in}(i)}$. Consider the function $\boldsymbol{\phi}_i=(\phi_{ij})_{j \in \mathcal{N}^{in}(i)}:\RR_+^{d_i-1} \to \RR^{d_i-1}_+$ implicitly defined by the equations:
\begin{align}
	\ell_{ij}= \sum_{h}\Sigma_{i,jh}\phi_{ij}\phi_{ih}, \quad \forall j
	\label{implicit_substitutes}
\end{align}
where $\Sigma_i \in \mathbb{R}^{(d_i-1)\times (d_i-1)}$  is a diagonally dominant and positive definite matrix. Lemma \ref{technology2} below proves that Equations \eqref{implicit_substitutes} indeed define a function $\boldsymbol{\phi}_{i}$.
Denote $\qv_i^{in}=\qv_{\mathcal{N}^{in}(i)}$ the vector of input quantities of firm $i$. Denote $\boldsymbol{\omega}_i=(\omega_{ij})_{j \in \mathcal{N}^{in}(i)}$ a nonnegative parameter vector, representing the intensity of each input $j$ in the technology of firm $i$. 

\begin{description}
	
	\item[Assumption T: Technology]
 The production function of firm $i$ is:
\begin{align}
	\Phi(\qv_i^{in},\tilde{\ell}_i)=\max_{\boldsymbol{\ell}_i: \, \sum_j\ell_j+\ell_{i0}=\tilde{\ell}_i} \sum_{j \in \mathcal{N}^{in}(i)} \omega_{ij} \min \{ \phi_{ij}(\boldsymbol{\ell}_i), -q_{ij} \} +\alpha_{i} \sqrt{2\ell_{i0}}
	\label{production_function}
\end{align}
where $\phi_{ij}$ is defined implicitly by \eqref{implicit_substitutes}. Define the renormalized (residual) labor as: $\ell_i:=\sqrt{2\ell_{i,0}}$.
With these definitions, the quantity vector $\qv_i$ satisfies the technology constraints: 
\item \begin{align}
		\boldsymbol{u}_i'\qv_i&=\alpha_i\ell_i \nonumber\\
		\tilde{\ell}_i&=	(\qv_i^{in})'\Sigma_i\qv_i^{in}+\frac{1}{2}\ell^2_{i}
		\label{tech_constraint_substitution} 
	\end{align}	
\end{description}


Expression \eqref{tech_constraint_substitution} follows from the fact that if $\tilde{\ell}_i$ is high enough, the labor allocation must be such that $-q_{ij}=\phi_{ij}(\boldsymbol{\ell}_i)$ for each $j$ (as a reminder, $q_{ij}$ represents the net sales: since $j$ is an input, $-q_{ij}$ is the quantity bought). Then, summing \eqref{implicit_substitutes} across $j$, we obtain the expression for labor $\tilde{\ell}_i$. This can be interpreted as follows. The firm hires a number of workers $\tilde{\ell}_i$. The number of workers has to be divided into tasks: each group $\ell_{ij}$ deals primarily with a specific intermediate input $j$, and $\ell_{i0}$ workers deal with tasks unrelated to specific inputs. Equation \eqref{implicit_substitutes} specifies how many workers are needed to be allocated primarily on input $j$. If $\Sigma_{i}$ is diagonal, then this number depends only on the quantity of input $j$, and \eqref{implicit_substitutes} yields $\phi_{ij}(\ell_{ij})=-q_{ij}=\sqrt{2\ell_{ij}}$. In general, off-diagonal terms $\Sigma_{i,jh}$ represent the interaction among intermediate inputs: if $\Sigma_{i,jh}>0$ then for the same input prices using both inputs $j$ and $k$ has a higher cost than only using one of the two: this captures substitutability; on the contrary, $\Sigma_{i,jk}<0$ captures complementarity. Here, the labor cost has no linear term analogous to $f_{i,L}$ in the main text: it could be added by modifying the definition in Equation \eqref{implicit_substitutes}, but I avoid it for simplicity.

\begin{lemma}
	
There exists a function $\boldsymbol{\phi}_i$ defined by the relations \eqref{implicit_substitutes}.
	
	\label{technology2}
	
\end{lemma}
The proof is in paragraph \ref{proof:technology2} below.

\subsection{The game}

The game is the same defined in Section \ref{sec:thegame}, with the only difference that the technology constraints are given by the Equations \ref{tech_constraint_substitution}, so that the coefficient matrices $B_i$ are not restricted to have rank 1.
We make the following assumption instead.

\begin{description}
	\item[Assumption A]: for all $i$, $B_i$ is symmetric and positive definite. 
\end{description}

Call $B=(B_i)_{i\in \mathcal{N}}$ the profile of coefficients chosen by firms, and $\mathcal{B}_i$ the set of matrices satisfying the above assumptions, so that the set of feasible action profiles is $\mathcal{B}=\prod_i \mathcal{B}_i$.

The definition of Generalized SDFE is the same as Definition \ref{def_generalized}, where the price impact function $\Lambda_i$ now has to be decreasing in each $B_j$ in the positive semidefinite ordering. 

\subsection{Results}

The following Lemma has the role of Lemma \ref{lemma-invertibility}. The proof is in \ref{proof:invertible_substitutes}.
\begin{lemma}
	The matrix  $M-\hat{B}_i$ is positive definite, so invertible.
	\label{invertible_substitutes}
\end{lemma}
Given this, the proof of Lemma \ref{linear-residual} is still valid, because it only uses the linearity of the schedules. In particular, this shows that the special case of the standard SDFE realizes for $\Lambda_i$ with the same expression as in Lemma \ref{linear-residual}. So, the best reply problem is analogous to \eqref{BestReply}, with different technology constraints. In particular, using the technology constraint \eqref{tech_constraint_substitution} to eliminate $\tilde{\ell}_i$, the best reply problem can be written as:
\begin{equation}
	\max_{\qv_i,\ell_{i}} (\pv^r_i)'\qv_i-(\qv_i^{in})'\Sigma_i\qv_i^{in}-\frac{1}{2}\ell^2_{i}
\end{equation}	
subject to $\boldsymbol{u}_i'\qv_i=\alpha_i \ell_i$.
Since $\alpha_i>0$, we can further eliminate the variable $\ell_i$ using the constraint: $\ell_i=\dfrac{1}{\alpha_i}(q_i^{out}+\oomega_i'\qv_i^{in})$, and rewrite the problem as:
\begin{equation}
	\max_{\qv_i,\ell_{i}} (\pv^r_i)'\qv_i-\dfrac{1}{2}\qv_i'C^{-1}_i\qv_i,
\end{equation}	
where we define:
\[
C_i:=\begin{pmatrix}
	\dfrac{1}{\alpha_i^2} & \dfrac{1}{\alpha_i^2}\oomega_i'\\
	\dfrac{1}{\alpha_i^2}\oomega_i &2\Sigma_i+\dfrac{1}{\alpha_i^2}\oomega_i\oomega_i'
\end{pmatrix}^{-1}
\]
Since $\Sigma_i$ is positive definite, also $C_i$ is.
The Hessian of the problem is $-2\Lambda_i-C_i^{-1}$, so the problem is strictly concave.
The FOCs are:
\begin{equation}
\pv_i^r-\Lambda_i\qv_i-C_i^{-1}\qv_i=0
\label{FOC_subs}
\end{equation}
So, we get:
\[
\mathcal{S}_i=(C_i^{-1}+\Lambda_i)^{-1}\pv_i
\]
So the coefficient matrix of firm $i$, in equilibrium, satisfies the expression:
\begin{equation}
	B_i=(C_i^{-1}+\Lambda_i)^{-1}
	\label{bestreply_substitution}
\end{equation}
Notice that $C_i$ is the matrix of coefficients of the schedules that represent the competitive equilibrium of the economy, because price-taking corresponds to $\Lambda_i=0$. Moreover, from \eqref{FOC_subs}, since $C^{-1}\qv_i$ is the marginal cost of labor, the markup-markdown vector still satisfies $\boldsymbol{\mu}_i=\Lambda_i\qv_i$. So, the characterization of markups in terms of centrality in the goods network is still valid: $\boldsymbol{\mu}_i=\Lambda_i\qv_i$, and also the centrality expression in Remark \ref{rmk_centrality}.

Existence of a solution of \eqref{bestreply_substitution} needs a different proof. In particular, action spaces are not unidimensional, which means that the game is not (necessarily) a potential game anymore. However, the best replies are still increasing, in the positive semidefinite ordering (by the assumption on $\Lambda_i$).
The following theorem states conditions for existence and a characterization of the equilibrium, generalizing Proposition \ref{teo_comparative}. The proof is in \ref{proof:existence_substitutes}.

\begin{teo}
	
	\begin{enumerate}

		
		\item  The best reply of firm $i$ is \eqref{bestreply_substitution}. Moreover, it is increasing in $B_{-i}$ in the positive semidefinite ordering.
		
		\item 	There exists a minimal and a maximal Nash equilibrium.
		
	\end{enumerate}
	\label{existence_substitutes}
\end{teo}

Proposition \ref{thm:specialcases} is still valid. The proof of Proposition \ref{thm:specialcases}
 only uses the fact that $M-\hat{B}_i$ is invertible, which is still true by Lemma \ref{invertible_substitutes}. 

Part 1 of Theorem \ref{thm:main_comparative} is still valid, provided the comparison is done using the maximal equilibrium of the multilateral model, because now uniqueness is not necessarily guaranteed.
The proof relies on the fact that the price impact is higher in the positive semidefinite sense for the multilateral model, and that iterating the best reply we obtain a decreasing sequence converging to the higher equilibrium of the multilateral model. Both facts are still true, with essentially the same proofs: where we simply have to consider the equilibrium profile the profile of matrices $B^*$ rather than $\overline{B}^*$, and the inequalities have to be intended in the positive semidefinite ordering.

\subsection{Proof of Lemma \ref{invertible_substitutes}}
\label{proof:invertible_substitutes}

If $\hat{B}_c$ is positive definite (which is the case when $\hat{B}_c=B_c$, that is all goods are consumed by the consumer), the thesis follows immediately, because $M-\hat{B}_i\ge B_c$.
If $\hat{B}_c$ has some zero rows, consider a nonzero vector $\xv$ such that $\xv'(M-\hat{B}_i) \xv=0$. Since $B_c$ and $B_j$ are positive definite, it must be that $x_g=0$ for all $g \in \mathcal{C}$, and $\xv_j=0$ for every $j\neq i$.
Since every good has both a producer and a buyer, for all $g$ in $\mathcal{N}(i)$ there is another firm $j$ such that $g\in \mathcal{N}(j)$. So, by the previous reasoning, $x_g=0$ for any good. \qed

\subsection{Proof of Theorem \ref{existence_substitutes}}
\label{proof:existence_substitutes}

\paragraph{Part 1}
The payoffs are strictly concave, so Equation \eqref{bestreply_substitution} is necessary and sufficient for optimization. We have to show that a profile of matrices satisfying it exists.
The best reply map $BR_i$ defined by equation \eqref{bestreply_substitution} is continuous and increasing in the p.s.d. order, by the assumption on $\Lambda_i$. 

\paragraph{Part 2}
The best reply is continuous. Moreover, the best reply belongs to the set $\{B_i \in \mathcal{B}_i\mid \lVert B_i\rVert \le \lVert C_i \rVert\}$, so it is bounded. Because of this, we have that the profile of best replies to $(C_i)_{i\in \mathcal{N}}$ satisfies $BR_i(C_{-i})\le C_i$. Since the best reply is monotonic, starting from the profile $(C_i)_{i\in \mathcal{N}}$ and iterating the best reply, we obtain a decreasing sequence. So, it must converge to a profile $B^*_i$, that is the maximal equilibrium. Analogously, if $B_i^0=0$ for all $i$, we have $BR_i(B_{-i}^0)\ge B_i^0$, and iterating again we find that the sequence converges to the minimal equilibrium $B_{i,*}$. So, there are a maximal and a minimal equilibrium, possibly identical.

\subsection{Proof of Lemma \ref{technology2}}
\label{proof:technology2}

We want to show that there always exist
a function $\boldsymbol{\phi}_i=(\phi_{ij})_{j \in \mathcal{N}^{in}(i)}:\RR_+^{d_i-1} \to \RR_+^{d_i-1}$, implicitly defined by \eqref{implicit_substitutes}.
We can rewrite such expression as: $\Phi_{i}(\boldsymbol{\ell}_{i},\boldsymbol{\phi}_i)=0$, for all $j$, where $\Phi_i:\RR^{d_i-1}\times \RR^{d_i-1} \to \RR^{d_i-1}$ is the map:
\[
\Phi_{ij}(\boldsymbol{\ell}_{i},\boldsymbol{\phi}_i)=-\ell_{ij}+ \sum_{h}\Sigma_{i,jh}\phi_{ij}\phi_{ih}
\]
We want to prove that this function of $\boldsymbol{\phi}_i$ has a zero for any $\boldsymbol{\ell}_{i}$. We use the intermediate value theorem, in the generalized form of the Poincaré-Miranda Theorem, \cite{kulpa1997poincare}. To do this, consider the function $\Phi_{i}(\boldsymbol{\ell_i},\, \cdot)$ limited to the domain $[0,\overline{F}(\boldsymbol{\ell}_i)]^{d_i-1}$, for some positive value $\overline{F}(\boldsymbol{\ell}_i)$.
If $\phi_{ij}=0$, we have $\Phi_{ij}\le 0$ for any $\phi_{ik}$ with $k\neq j$. If $\phi_{ij}=\overline{F}(\boldsymbol{\ell}_i)$, since $\Sigma_i$ is strictly diagonally dominant, and $\Sigma_{i,jh}\ge -|\Sigma_{i,jh}|$, 
\[
\Sigma_{i,jj}\overline{F}(\boldsymbol{\ell}_i)+\sum_{h \neq j}\Sigma_{i,jh}\phi_{ih}\ge \Sigma_{i,jj}\overline{F}(\boldsymbol{\ell}_i)-\sum_{h \neq j}|\Sigma_{i,jh}|\phi_{ih}> \sum_{h \neq j} |\Sigma_{i,jh}|(\overline{F}(\boldsymbol{\ell}_i)-\phi_{ih})\ge 0.
\]
Then:
\[
\Phi_{ij}(\boldsymbol{\ell_i},\,\overline{F}(\boldsymbol{\ell}_i), (\phi_{ih})_{h\neq j})=-\ell_{ij}+\overline{F}(\boldsymbol{\ell}_i)(\Sigma_{i,jj}\overline{F}(\boldsymbol{\ell}_i)+\sum_{h \neq j}\Sigma_{i,jh}\phi_{ih})
\]
Now we want this to be nonnegative for any choice of $\phi_{ih}$ in $[0,\overline{F}(\boldsymbol{\ell}_i)]$. The configuration of $\phi_{ih}$ that minimizes $\Phi_{ij}(\boldsymbol{\ell_i},\,\overline{F}(\boldsymbol{\ell}_i), (\phi_{ih})_{h\neq j})$ is such that $\phi_{ih}=0$ if $\Sigma_{i,jh}\ge 0$, and $\phi_{ih}=\overline{F}(\boldsymbol{\ell}_i)$ if $\Sigma_{i,jh}<0$. So:
\[
\min_{\phi_{ih} \in [0,\overline{F}(\boldsymbol{\ell}_i)] \text{ for } h\neq i} \Phi_{ij}(\boldsymbol{\ell_i},\,\overline{F}(\boldsymbol{\ell}_i), (\phi_{ih})_{h\neq j})=-\ell_{ij}+\overline{F}(\boldsymbol{\ell}_i)^2\left(\Sigma_{i,jj}-\sum_{h \neq j,\, \Sigma_{i,jh}<0}|\Sigma_{i,jh}|\right)
\] 
By diagonal dominance, $(\Sigma_{i,jj}-\sum_{h \neq j,\, \Sigma_{i,jh}<0}|\Sigma_{i,jh}|)>0$. Then, if we set, for any fixed vector $\boldsymbol{\ell}_i$:
\[
\overline{F}(\boldsymbol{\ell}_i):=\dfrac{\max_j \ell_j}{\Sigma_{i,jj}-\sum_{h \neq j,\, \Sigma_{i,jh}<0}|\Sigma_{i,jh}|}
\]
then, by the Poincaré-Miranda Theorem it follows that $\Phi_{i}(\boldsymbol{\ell_i}\, \cdot)$ has a zero in $[0,\overline{F}(\boldsymbol{\ell}_i)]^{d_i-1}$. So, there exist a function satisfying the condition, which is what we wanted to show. \qed 

\end{document}